\documentclass[runningheads,11pt]{llncs}

\usepackage{geometry}
\usepackage[activate=true,final,
tracking=true,
kerning=true,
spacing=true,
factor=1100,
stretch=10,
shrink=10]{microtype}

\usepackage{hyperref}
\usepackage{amsmath,mathtools}
\usepackage{amssymb}
\usepackage{enumerate}
\usepackage{esvect}
\usepackage{cite}
\usepackage{pgf,tikz,caption}
\usetikzlibrary{arrows,arrows.meta,automata,positioning,calc}
\renewcommand{\qed}{\hfill \blacksquare}
\usepackage[utf8]{inputenc}
\usepackage[T1]{fontenc}
\usepackage{lmodern}
\usepackage{mathalfa}
\numberwithin{equation}{section}
\usepackage{bm}
\usepackage{spverbatim}
\usepackage{breakcites}
\usepackage{enumitem}
\usepackage{complexity}
\hypersetup{
	colorlinks,
	citecolor=black,
	filecolor=black,
	linkcolor=black,
	urlcolor=black
}

\DeclareRobustCommand\sampleline[1]{%
	\tikz\draw[#1] (0,0) (0,\the\dimexpr\fontdimen22\textfont2\relax)
	-- (2em,\the\dimexpr\fontdimen22\textfont2\relax);%
}
\usepackage{nccmath}
\usepackage{multicol,xparse}
\usepackage{cleveref}
\usepackage{xcolor}
\usepackage{upgreek}

\makeatletter
\newcommand{\customlabel}[2]{\protected@write @auxout {}{\string \newlabel {#1}{{#2}{\thepage}{#2}{#1}{}}}\hypertarget{#1}{}}
\newcommand{\cv}[1]{\reflectbox{\ensuremath{\vv{\reflectbox{\ensuremath{#1}}}}}}	

\makeatother

\makeatletter
\def\@fnsymbol#1{\ensuremath{\ifcase#1\or \dagger\or \ddagger\or
		\mathsection\or \mathparagraph\or \|\or **\or \dagger\dagger
		\or \ddagger\ddagger \else\@ctrerr\fi}}
\makeatother
\newcommand*\samethanks[1][\value{footnote}]{\footnotemark[#1]}

\title{Extremal Set Theory and LWE Based Access Structure Hiding Verifiable Secret Sharing with Malicious-Majority and Free Verification}
\author{Vipin Singh Sehrawat \inst{1}\thanks{Part of this work was done while the author was a PhD candidate at The University of Texas at Dallas, USA.}\thanks{Research partially supported by NPRP award NPRP8-2158-1-423 from the Qatar National Research Fund (a member of The Qatar Foundation). The statements made herein are solely the responsibility of the authors.} \and Foo Yee Yeo\inst{1} \and Yvo Desmedt \inst{2,3}\samethanks}
\institute{Seagate Technology, Singapore \\ \email{\{vipin.sehrawat.cs@gmail.com\}} \and 
	The University of Texas at Dallas, Richardson, USA \and
	University College London, London, UK}

\titlerunning{Access Structure Hiding Verifiable Secret Sharing with Malicious-Majority and Free Verification}
\authorrunning{V. S. Sehrawat, F. Y. Yeo, and Y. Desmedt}

\begin{document}
\maketitle	
	\begin{abstract} \normalsize
		Secret sharing allows a dealer to distribute a secret among a set of parties such that only authorized subsets, specified by an access structure, can reconstruct the secret. Sehrawat and Desmedt (COCOON 2020) introduced \emph{hidden access structures}, that remain secret until some authorized subset of parties collaborate. However, their scheme assumes semi-honest parties and supports only restricted access structures. We address these shortcomings by constructing a novel access structure hiding verifiable secret sharing scheme that supports \emph{all} monotone access structures. Our scheme is the first secret sharing solution to support malicious behavior identification and share verifiability in malicious-majority settings. Furthermore, the verification procedure of our scheme incurs no communication overhead, and is therefore ``free''. As the building blocks of our scheme, we introduce and construct the following:
		\begin{itemize}
			\item a set-system with greater than $\exp\left(c\frac{2(\log h)^2}{(\log\log h)}\right)+2\exp\left(c\frac{(\log h)^2}{(\log\log h)}\right)$ subsets of a set of $h$ elements. Our set-system, $\mathcal{H}$, is defined over $\mathbb{Z}_m$, where $m$ is a non-prime-power. The size of each set in $\mathcal{H}$ is divisible by $m$ while the sizes of the pairwise intersections of different sets are not divisible by $m$ unless one set is a (proper) subset of the other,
			\item a new variant of the learning with errors (LWE) problem, called \textsf{PRIM-LWE}, wherein the secret matrix is sampled such that its determinant is a generator of $\mathbb{Z}_q^*$, where $q$ is the LWE modulus. 
		\end{itemize}
		Our scheme arranges parties as nodes of a directed acyclic graph and employs modulus switching during share generation and secret reconstruction. For a setting with $\ell$ parties, our (non-linear) scheme supports all $2^{2^{\ell - O(\log \ell)}}$ monotone access structures, and its security relies on the hardness of the LWE problem. Our scheme's maximum share size, for any access structure, is: $$(1+ o(1)) \dfrac{2^{\ell}}{\sqrt{\pi \ell/2}}(2 q^{\varrho + 0.5} + \sqrt{q} + \mathrm{\Theta}(h)),$$ where $\varrho \leq 1$ is a constant. We provide directions for future work to reduce the maximum share size to:
		\[\dfrac{1}{l+1} \left( (1+ o(1)) \dfrac{2^{\ell}}{\sqrt{\pi \ell/2}}(2 q^{\varrho + 0.5} + 2\sqrt{q}) \right),\]
		where $l \geq 2$. We also discuss three applications of our secret sharing scheme.
	\end{abstract}
	
\section{Introduction}\label{sec1}
A secret sharing scheme is a method by which a dealer distributes shares of a secret to a set of parties such that only authorized subsets of parties, specified by an access structure, can combine their shares to reconstruct the secret. As noted by Shamir~\cite{Shamir[79]}, the mechanical approach to secret sharing, involving multiple locks to a mechanical safe, was already known to researchers in combinatorics (see \cite{Liu[68]}, Example 1-11). Digital secret sharing schemes were introduced in the late 1970s by Shamir~\cite{Shamir[79]} and Blakley~\cite{Blakley[79]} for the $t$-out-of-$\ell$ threshold access structure, wherein all subsets of cardinality at least $t ~(t \in [\ell])$ are authorized. Ito et al.~\cite{Ito[87]} showed the existence of a secret sharing scheme for every monotone access structure. A number of strengthenings of secret sharing, such as verifiable secret sharing \cite{Chor[85]}, identifiable secret sharing \cite{McSar[81]}, robust secret sharing \cite{Rabin[89]}, rational secret sharing \cite{HalTea[04]}, ramp secret sharing \cite{BlaCat[84]}, evolving secret sharing \cite{KomaMoni[16]}, proactive secret sharing \cite{Amir[95]}, dynamic secret sharing \cite{ChiLein[89]}, secret sharing with veto capability \cite{Beut[89]}, anonymous secret sharing \cite{Stinson[87]}, evolving ramp secret-sharing \cite{AmHuss[20]}, locally repairable secret sharing \cite{AgaMazu[16]} and leakage-resilient secret sharing \cite{FrabAks[21],VipAshK[18]}, have been proposed under varying settings and assumptions. Quantum versions have also been developed for some secret sharing variants (e.g., see \cite{MarkBui[99],QinLiu[21],YaoGuo[21],QinWall[18],CleDan[99],LuMiao[18],LuHou[20],JoySab[20],SutraOm[20]}). Secret sharing is the foundation of multiple cryptographic constructs and applications, including threshold cryptography~\cite{YvoFrank[89],Yvo[89],Santis[94],TalR[06]}, (secure) multiparty computation~\cite{Ben[88],Chaum[88],Cramer[00],CramDam[15],HirtMau[97],GolMic[87],HaoRon[06]}, secure distributed storage \cite{GustSimm[89]}, attribute-based encryption~\cite{Goyal[06],Waterss[11],ConstantDr[15]}, generalized oblivious transfer~\cite{Tassa[11],Shankar[08]}, perfectly secure message transmission~\cite{Danny[93],AshArpAsh[11],MartPate[11],YangYvo[10]}, access control~\cite{Naor[06],HuLi[18],IlanMark[18]}, anonymous communications~\cite{Sehrawat[17]}, leakage-resilient circuit compilers \cite{SebaTal[10],YuvAmi[03],RothG[12]}, e-voting~\cite{Berry[99],Yung[04],Sorin[07]}, e-auctions~\cite{Mic[98],Peter[09]}, secure cloud computing~\cite{MehrDoug[12],SatKei[13]}, witness pseudorandom functions \cite{IlanMark[18]}, cloud data security~\cite{VaruDar[17],NungJia[13]}, distributed storage blockchain \cite{SihAhm[20],RaviLav[18],RavLavR[18],KimKiran[19],DaiZha[18]}, copyright protection \cite{HsiHsu[07],JonYan[10]}, indistinguishability obfuscation \cite{IlanMark[18]}, multimedia applications \cite{AdnanEsam[19]}, and private (linear and logistic) regression \cite{Gasc[17],ShiChao[16],ChaoJun[20]}, tree-based models \cite{FangChen[20]} and general machine learning algorithms \cite{DanTho[15],PayPet[18],PayYup[17]}. 

\subsection{\textbf{Motivation}}
\subsubsection{Hidden Access Structures.}
Traditional secret sharing models require the access structure to be known to the parties. Since secret reconstruction requires shares of any authorized subset from the access structure, having a public access structure reveals the high-value targets, which can lead to compromised security in the presence of malicious parties. Having a public access structure also implies that some parties must publicly consent to the fact that they themselves are not trusted. As a motivating example, consider a scenario where Alice dictates her will/testament and instructs her lawyer that each of her 10 family members should receive a valid share of the will. In addition, the shares should be indistinguishable from each other in terms of size and entropy. She also insists that to reconstruct her will, \{Bob, Tom, Catherine\} or \{Bob, Cristine, Brad, Roger\} or \{Rob, Eve\} must be part of the collaborating set. However, Alice does not want to be in the bad books of her other, less trusted family members. Therefore, she demands that the shares of her will and the procedure to reconstruct it back from the shares must not reveal her ``trust structures'', until after the will is successfully reconstructed. This problem can be generalized to secret sharing with \textit{hidden} access structures, that remain secret until some authorized subset of parties assembles. However, the (only) known access structure hiding secret sharing scheme does not support all $2^{2^{\ell - O(\log \ell)}}$ monotone access structures, where $\ell$ denotes the number of parties~\cite{Vipin[20],VipinThesis[19]}, but only those access structures where the smallest authorized subset contains at least half of the total number of parties. 

\subsubsection{Superpolynomial Size Set-Systems and Efficient Cryptography.}
In this work, we consider the application of set-systems with specific intersections towards enhancing existing cryptographic protocols for distributed security. To minimize the overall computational and space overhead of such protocols, it is desirable that the parameters such as exponents, moduli and dimensions do not grow too large. For a set-system whose size is superpolynomial in the number of elements over which it is defined, achieving a sufficiently large size requires a smaller modulus and fewer elements, which translates into smaller dimensions, exponents and moduli for its cryptographic applications. Hence, quickly growing set-systems are well-suited for the purpose of constructing (relatively) efficient cryptographic protocols. 

\subsubsection{Lattice Based Secret Sharing for General Access Structures.} Lattice-based cryptosystems are among the leading ``post-quantum'' cryptographic candidates that are plausibly secure from large-scale quantum computers. For a thorough review of the various implementations of lattice-based cryptosystems, we refer the interested reader to the survey by Nejatollahi et al.~\cite{Hamid[19]}. With NIST's latest announcements~\cite{NIST[20]}, the transition towards widespread deployment of lattice-based cryptography is expected to pick up even more steam. However, existing lattice-based secret sharing schemes support only threshold access structures~\cite{Stein[07],Pila[17]}. Hence, there is a need to develop lattice-based secret sharing schemes for general (i.e., all monotone) access structures. 

\subsubsection{(Im)possibility of Verifiable Secret Sharing for Malicious-Majority.}
In its original form, secret sharing assumes a fault-free system, wherein the dealer and parties are honest. Verifiable secret sharing (VSS) relaxes this assumption, guaranteeing that there is some unique secret that a malicious dealer must ``commit'' to. The objective of VSS is to resist malicious parties, which are classified as follows: 
\begin{itemize}
\item a dealer sending incorrect shares, 
\item malicious parties submitting incorrect shares for secret reconstruction. 
\end{itemize}

VSS is a fundamental building block for many secure distributed computing protocols, such as (secure) multiparty computation and byzantine agreement~\cite{Abraham[08],Canetti[93],Feld[88],Katz[06],Patra[14]}. Tompa and Woll~\cite{Tompa[89]}, and McEliece and Sarwate~\cite{McEliece[81]} gave the first (partial) solutions to realize VSS, but the notion was defined and fully realized first by Chor et al.~\cite{Chor[85]}. Since then, multiple solutions, under various assumptions, have been proposed~\cite{Chor[85],Cachin[02],Chaum[88],Ben[88],Feld[87],Oded[91],Tor[01],Rabin[89],Genn[98],Basu[19],Kate[10],Cascudo[17],Backes[13],Stadler[96]}. VSS typically assumes that the parties are connected pairwise by authenticated private channels and they all have a broadcast channel, which allows one party to send a consistent message to all other parties, guaranteeing consistency even if the broadcaster itself is malicious. However, even probabilistically, broadcast cannot be simulated on a point-to-point network when more than a third of the parties are malicious. Therefore, it is infeasible to construct VSS protocols when more than a third of the parties are malicious~\cite{Lamport[82]}. Hence, relaxed definitions of verifiability must be explored to design efficient schemes that: 
\begin{itemize}
	\item do not fail when more than a third of the parties are malicious,
	\item unlike VSS and related concepts, do not require additional communication or cryptographic protocols.
\end{itemize} 

\subsection{\textbf{Related Work}}
A limited number of attempts have been made to introduce privacy-preserving features to secret sharing. The first solution that focused on bolstering privacy for secret sharing was called anonymous secret sharing, wherein the secret can be reconstructed without the knowledge of which parties hold which shares~\cite{Stinson[87]}. In such schemes, secret reconstruction can be performed by giving the shares to a black box that does not know the identities of the parties holding those shares. As pointed out by Guillermo et al. \cite{Mida[03]}, anonymous secret sharing does not provide cryptographic anonymity. Existing anonymous secret sharing schemes either operate in restricted settings (e.g., $\ell$-out-of-$\ell$, $2$-out-of-$\ell$ threshold) or use difficult to generate underlying primitives \cite{Stinson[87],Phillips[92],Blundo[96],Kishi[02],Deng[07],HongWei[12]}. For instance, the constructions from~\cite{Stinson[87],Blundo[96]} use resolvable Steiner systems~\cite{Steiner[53]}, which are non-trivial to achieve and have only a few known results in restricted settings~\cite{Bryant[17],Teir[94],Col[92],Lou[20],Kwan[20],Pat[17],Pipp[89],Chau[71],Fer[19],Kra[95],Yuca[99],Yuca[02]}. There are also known impossibility results concerning the existence of certain desirable Steiner systems~\cite{Pat[08]}. For an introduction to Steiner systems, we refer the interested reader to~\cite{Tri[99],Charles[06]}. Kishimoto et al.~\cite{Kishi[02]} employed combinatorics to realize anonymous secret sharing, thereby avoiding the difficult to generate primitives. However, their scheme also works for only certain specific thresholds. 

Daza and Domingo-Ferrer \cite{Daza[07]} aimed at achieving a weaker form of anonymous secret sharing wherein the notion of privacy is analogous to that for ring signatures~\cite{RonAdi[01]}, i.e., instead of a party's identity, only its subset membership is leaked. Recently, Sehrawat and Desmedt~\cite{Vipin[20]} introduced \textit{access structure hiding} secret sharing for restricted access structures, wherein no non-negligible information about the access structure gets revealed until some authorized subset of parties assembles. They constructed novel set-systems and vector families to ``encode'' the access structures such that deterministic and \emph{private} assessments can be conducted to test whether a given subset of parties is authorized for secret reconstruction. 

\subsection{\textbf{Our Contributions}}
The access structure hiding secret sharing scheme from~\cite{Vipin[20]} has the following limitations:
\begin{enumerate}
	\item It assumes semi-honest polynomial-time adversaries, which try to gain additional information while correctly following the protocol. Hence, the scheme fails in the presence of malicious adversaries, which are not guaranteed to follow the protocol correctly.
	\item It requires that the smallest authorized subset contain at least half of the total number of parties.
\end{enumerate}
We address these limitations by introducing access structure hiding verifiable secret sharing, which supports all monotone access structures and remains ``verifiable'' even when a majority of the parties are malicious. Our detailed contributions follow: 

\subsubsection{Novel Superpolynomial Sized Set-Systems and Vector Families.}
In order to build our access structure hiding verifiable secret sharing scheme, we construct a set-system that is described by \Cref{mainThm} in the following text. 
\begin{definition}\label{def1} 
	\emph{We say that a family of sets $\{G_1,\,G_2,\,\ldots,G_t\}$ is \emph{non-degenerate} if there does not exist $1\leq i\leq t$ such that $G_i\subseteq G_j$ for all $1\leq j\leq t$.}
\end{definition}

\begin{definition}\label{def2}
	\emph{Let $m\geq 2$, $t\geq 2$ be integers and $\mathcal{H}$ be a set-system. We shall say that $\mathcal{H}$ has \emph{$t$-wise restricted intersections} modulo $m$ if the following two conditions hold:
	\begin{enumerate}
		\item $\forall H\in\mathcal{H}$, $\vert H\vert = 0\bmod m$,
		\item $\forall t^\prime$ satisfying $2\leq t^\prime\leq t$, and $\forall H_1,\,H_2,\,\ldots,\,H_{t^\prime}\in\mathcal{H}$ with $\{H_1,\,H_2,\,\ldots,\,H_{t^\prime}\}$ non-degenerate, it holds that:
		$$\left\vert\bigcap_{\tau=1}^{t^\prime} H_\tau\right\vert\neq 0\bmod m.$$
	\end{enumerate}}
\end{definition} 

\begin{theorem}\label{mainThm}
	Let $\{\alpha_i\}_{i=1}^r$ be $r > 1$ positive integers and $m = \prod_{i=1}^{r} p_i^{\alpha_i}$ be a positive integer with $r$ different odd prime divisors: $p_1, \ldots, p_r$, and $l\geq 2$ be an integer such that $l<\min(p_1,\,\ldots,\,p_r)$. Then, there exists $c>0$ such that for all integers $t\geq 2$ and $h\geq lm$, there exists an explicitly constructible non-uniform\footnote{member	sets do not	all have equal size} set-system $\mathcal{H}$, defined over a universe of $h$ elements, such that
	\begin{enumerate}
		\item $\vert\mathcal{H}\vert>\exp\left(c\dfrac{l(\log h)^r}{(\log\log h)^{r-1}}\right)+l\exp\left(c\dfrac{(\log h)^r}{(\log\log h)^{r-1}}\right)$,
		\item $\forall H_1, H_2\in\mathcal{H}$, either $\vert H_1\vert=\vert H_2\vert$, $\vert H_1\vert=l\vert H_2\vert$ or $l\vert H_1\vert=\vert H_2\vert$,
		\item $\mathcal{H}$ has $t$-wise restricted intersections modulo $m$.
	\end{enumerate}
\end{theorem}

Recall that $a \bmod m$ denotes the smallest non-negative $b = a \bmod m$. Since the access structure $\Gamma$ is monotone, it holds that if $\mathcal{B} \supseteq \mathcal{A}$, for some $\mathcal{A} \in \Gamma$, then $\mathcal{B} \in \Gamma$. We derive a family of vectors $\mathcal{V} \in (\mathbb{Z}_m)^h$ from our set-system $\mathcal{H}$, that captures the superset-subset relations in $\mathcal{H}$ as (vector) inner products in $\mathcal{V}$. This capability allows us to capture \textit{special} information about each authorized subset $\mathcal{A} \in \Gamma$ in the form of an inner product, enabling us to devise an efficient test for whether a given subset of parties $\mathcal{B}$ is a superset of any $\mathcal{A} \in \Gamma$. 

\subsubsection{\textsf{PRIM-LWE}.} Informally, (the multi-secret version of) the learning with errors (LWE) problem~\cite{Reg[05]} asks for the solution of a system of noisy linear modular equations: given positive integers $n$, $w = \poly(n)$ and $q \geq 2$, an LWE sample consists of $(\textbf{A}, \textbf{B} = \textbf{A} \textbf{S} + \textbf{E} \bmod q)$ for a fixed secret $\textbf{S} \in \mathbb{Z}^{n \times n}_q$ with small entries, and $\textbf{A} \xleftarrow{\; \$ \;} \mathbb{Z}^{w \times n}_q$. The error term $\textbf{E} \in \mathbb{Z}^{w \times n}$ is sampled from some distribution supported on small numbers, typically a (discrete or rounded) Gaussian distribution with standard deviation $\alpha q$ for $\alpha = o(1)$. We introduce a new variant of the LWE problem, called \textsf{PRIM-LWE}, wherein the matrix $\textbf{S}$ can be sampled from the set of matrices whose determinants are generators of $\mathbb{Z}_q^\ast$. We prove that, up to a constant factor, \textsf{PRIM-LWE} is as hard as the plain LWE problem.

\subsubsection{Access Structure Hiding Verifiable Secret Sharing Scheme.} We use our novel set-system and vector family to generate \textsf{PRIM-LWE} instances, and thereby construct the first access structure hiding verifiable (computational) secret sharing scheme that guarantees secrecy, correctness and verifiability (with high probability) even when a majority of the parties are malicious. To detect malicious behavior, we postpone the verification procedure until after secret reconstruction. The idea of delaying verification till secret reconstruction is also used in identifiable secret sharing \cite{McSar[81]} wherein parties only interact with a trusted external \emph{stateless} server and the goal is to inform each honest player of the correct set of cheaters. However, unlike the identifiable secret sharing solutions \cite{AshC[11],KaoSat[95],SatoOb[11],McSar[81],SatTos[06],IshRafa[12],MasaTake[18],ChenLi[15]}, our scheme supports share verification and does not require any digital signature or message authentication subroutines. Furthermore, our scheme does not require any dedicated round to verify whether the reconstructed secret is consistent with all participating shares. Our scheme is graph-based with the parties represented by nodes in a directed acyclic graph (DAG). For a setting with $\ell$ parties, our (non-linear) scheme supports \emph{all} monotone access structures, and its security relies on the hardness of the LWE problem. The maximum share size of our scheme is $(1+ o(1)) \dfrac{2^{\ell}}{\sqrt{\pi \ell/2}}(2 q^{\varrho + 0.5} + \sqrt{q} + \mathrm{\Theta}(h))$, where $q$ is the LWE modulus and $\varrho \leq 1$ is a constant. We also describe improvements that will lead to an access structure hiding verifiable secret sharing scheme with maximum share size equal to:
\[\dfrac{1}{l+1} \left( (1+ o(1)) \dfrac{2^{\ell}}{\sqrt{\pi \ell/2}}(2 q^{\varrho + 0.5} + 2\sqrt{q}) \right),\]
where $l \geq 2$ (as defined by Theorem~\ref{mainThm}). 

\subsection{\textbf{``Free'' Verification at the Expense of Larger Shares}}
In the first secret sharing scheme for general (monotone) access structures~\cite{Ito[87]}, the share size is proportional to the depth 2 complexity of the access structure when viewed as a Boolean function; hence, shares are exponential for most access structures. While for specific access structures, the share size of the later schemes~\cite{Brick[89],Karch[93],Gusta[88]} is less than the share size for the scheme from~\cite{Ito[87]}, the share size of all schemes for general access structures remained $2^{\ell-o(\ell)}~(\ell$ denotes the number of parties) until 2018, when Liu and Vaikuntanathan~\cite{LiuStoc[18]} (using results from~\cite{Liu[18]}) constructed a secret sharing scheme for general access structures with a share size of $2^{0.944 \ell}$. Applebaum et al.~\cite{Benny[20]} (using the results of~\cite{Apple[19],Liu[18]}) constructed a secret sharing scheme for general access structures with a share size of $2^{0.637\ell + o(\ell)}$. Whether the share size for general access structures can be improved to $2^{o(\ell)}$ (or even smaller) remains an important open problem. On the other hand, multiple works~\cite{Blundo[92],Capo[93],Csi[96],Csi[97],Dijk[95]} have proved various lower bounds on the share size of secret sharing for general access structures, with the best being $\mathrm{\Omega}(\ell^2/\log \ell)$ from Csirmaz~\cite{Csi[96]}. 

The maximum share size of our access structure hiding verifiable secret sharing scheme is: $$(1+ o(1)) \dfrac{2^{\ell}}{\sqrt{\pi \ell/2}}(2 q^{\varrho + 0.5} + \sqrt{q} + \Theta(h)),$$ where $q$ is the LWE modulus and $\varrho \leq 1$ is a constant. Therefore, the maximum share size of our scheme is larger than the best known upper bound of $2^{0.637\ell + o(\ell)}$ on the share size for secret sharing over general access structures. However, at the expense of the larger share size, our scheme achieves ``free'' verification because unlike the existing VSS protocols, whose verification procedures incur at least $O(\ell^2)$ communication overhead~\cite{Cachin[02],Backes[13]}, the verification procedure of our scheme does not incur any communication overhead. 

\subsection{\textbf{Applications}}
In this section, we discuss three example applications of our access structure hiding verifiable secret sharing scheme.

\subsubsection{Frameproof Secret Sharing.}
In secret sharing, any authorized subset of parties can compute the shares of another authorized subset of parties and use the latter's shares to perform non licet activities. For example, they can reconstruct a key and sign a message on behalf of their organization, and later, during audit, they can put the blame on the other authorized subset. By doing this they may escape the accountability of using the key inappropriately. Recently, Desmedt et al.~\cite{Yvo[21]} captured this threat by defining \textit{framing} in secret sharing schemes as the ability of some subset $\mathcal{A} \subset \mathcal{P}$ to compute the share of any participant $P_i \in \mathcal{P} \setminus \mathcal{A}$. In our access structure hiding verifiable secret sharing scheme, the share of each party $P_i$ is sealed as a \textsf{PRIM-LWE} instance such that the lattice basis, $\textbf{A}_i$, used to generate it is known only to $P_i$. Since $\textbf{A}_i$ is required to generate $P_i$'s share, it is infeasible for any coalition of polynomial-time parties $\mathcal{A} \subset \mathcal{P}$ to compute the share of $P_i \in \mathcal{P} \setminus \mathcal{A}$ without solving the LWE problem. Hence, our access structure hiding verifiable secret sharing scheme is not vulnerable to framing, and is therefore \textit{frameproof}.

\subsubsection{Eternity Service.} Eternity service aims to use redundancy and scattering techniques to replicate data across a large set of machines (such as the Internet), and add anonymity mechanisms to increase the cost of selective service denial attacks~\cite{Ross[96]}. We know that secret sharing is a provably secure scattering technique. Moreover, hidden access structures guarantee that neither an insider nor outsider (polynomial) adversary can know the access structure without collecting all shares of some authorized subset, making it impossible for the adversary to identify targets for selective service denial attacks. Hence, access structure hiding secret sharing fits the requirements for realizing eternity service. 

\subsubsection{Undetectable Honeypots.} Honeypots are information systems resources conceived to attract, detect, and gather attack information. Honeypots serve several purposes, including the following: 
\begin{itemize}
	\item distracting adversaries from more valuable machines on a network, 
	\item providing early warning about new attack and exploitation trends, 
	\item allowing in-depth examination of adversaries during and after exploitation of the honeypot. 
\end{itemize}
The value of a honeypot is determined by the information that we can obtain from it. Monitoring the data that enters and leaves a honeypot lets us gather information that is not available to network intrusion detection systems. For example, we can log the key strokes of an interactive session even if encryption is used to protect the network traffic. Although the concept is not new~\cite{Stoll[89]}, interest in protection and countermeasure mechanisms using honeypots has become popular only during the past two decades~\cite{Beham[13],Kulkarni[12],Hussein[15],Fred[14]}. For an introduction to the topic, we refer the interested reader to~\cite{Lance[03]}. Unfortunately, honeypots are easy to detect and avoid~\cite{Joni[17],Kraw[04],Eyal[17],Wang[10],Zou[06],Osama[12],Rowe[06],Neil[06],Holz[05],Dorn[04],Vrable[05]}.

In scenarios wherein secret sharing is used to distribute a secret (e.g., encryption keys) among multiple servers, hidden access structures would allow the dealer to provide all servers with \emph{legitimate} shares while enforcing zero or negligible information leakage without some authorized subset's shares. Moreover, each share corresponds to the same secret and the entropy of all shares is equal. Since access structures are hidden, the dealer can keep servers out of the minimal authorized subsets without revealing this information. Shares from such servers are ``useless'' since their participation is only optional for successful secret reconstruction. Because their shares do not hold any value without the participation of an authorized subset, these servers can be exposed to attackers and turned into honeypots. Furthermore, since the protocol allows all servers to participate in secret reconstruction, identifying honeypots is impossible until successful secret reconstruction. 

\subsection{\textbf{Organization}} 
The rest of the paper is organized as follows: Section~\ref{Sec2} recalls necessary definitions and constructs that are required for our constructions and solutions. Section~\ref{Sec3} formally defines access structure hiding verifiable secret sharing scheme. In Section~\ref{construction}, we construct the first building block for our secret sharing scheme, i.e., our superpolynomial size set-systems and vector families. Section~\ref{work} establishes that our set-systems can be operated upon via the vector families. Section~\ref{sec5} extends the idea from Section~\ref{work} by introducing \emph{access structure tokens} and giving an example procedure to generate access structure tokens to ``encode'' \emph{any} monotone access structure. In Section~\ref{prim}, we introduce a new variant of LWE, called \textsf{PRIM-LWE}. We present our access structure hiding verifiable secret sharing scheme in Section~\ref{sec6}. We conclude with a conclusion in Section~\ref{conclude}.

\section{Preliminaries}\label{Sec2}
For a positive integer $n$, let $[n]$ denote the set of the first $n$ positive integers, i.e., $[n] = \{1,\dots,n\}$. 

\begin{theorem}[Dirichlet's Theorem]\label{Dr}
	For all coprime integers $c$ and $q$, there are infinitely many primes, $p$, of the form $p = c \bmod q.$
\end{theorem}

\begin{theorem}[Fermat's Little Theorem]\label{Fermat} 
	If $p$ is a prime and $c$ is any number coprime to $p$, then $c^{p-1} = 1 \bmod p$.
\end{theorem}

\begin{theorem}[Euler's Theorem]\label{Euler}
	Let $y$ be a positive integer and $\mathbb{Z}_y^*$ denote the multiplicative group modulo $y$. Then for every integer $c$ that is coprime to $y$, it holds that: $c^{\varphi(y)} = 1 \bmod y,$ where $\varphi(y) = |\mathbb{Z}_y^*|$ denotes Euler's totient function.
\end{theorem}

For a detailed background on \Cref{Dr,Fermat,Euler}, we refer the interested reader to~\cite{Hardy[80]}.

\begin{definition}[Hadamard/Schur product]\label{Hada}
	\emph{Hadamard/Schur product of two vectors $\textbf{u}, \textbf{v} \in \mathcal{R}^n$, denoted by $\textbf{u} \circ \textbf{v}$, returns a vector in the same linear space whose $i$-th element is defined as: $(\textbf{u} \circ \textbf{v})[i] = \textbf{u}[i] \cdot \textbf{v}[i],$ for all $i \in [n].$}
\end{definition}

\begin{definition}[Negligible Function]\label{Neg}
	\emph{For security parameter $\omega$, a function $\epsilon(\omega)$ is called \textit{negligible} if for all $c > 0$, there exists a $\omega_0$ such that $\epsilon(\omega) < 1/\omega^c$ for all $\omega > \omega_0$.}
\end{definition}

\begin{definition}[Computational Indistinguishability~\cite{Gold[82]}]
	\emph{Let $X = \{X_\lambda\}_{\lambda \in \mathbb{N}}$ and $Y = \{Y_\lambda\}_{\lambda \in \mathbb{N}}$ be ensembles, where $X_\lambda$'s and $Y_\lambda$'s are probability distributions over $\{0,1\}^{\kappa(\lambda)}$ for some polynomial $\kappa(\lambda)$. We say that $\{X_\lambda\}_{\lambda \in \mathbb{N}}$ and $\{Y_\lambda\}_{\lambda \in \mathbb{N}}$ are polynomially/computationally indistinguishable if the following holds for every (probabilistic) polynomial-time algorithm $\mathcal{D}$ and all $\lambda \in \mathbb{N}$:
	\[\Big| \Pr[t \leftarrow X_\lambda: \mathcal{D}(t) = 1] - \Pr[t \leftarrow Y_\lambda: \mathcal{D}(t) = 1] \Big| \leq \epsilon(\lambda),\]
	where $\epsilon$ is a negligible function.}
\end{definition}	

\begin{definition}[Access Structure] 
	\emph{Let $\mathcal{P} = \{P_1, \dots, P_\ell\}$ be a set of parties. A collection $\Gamma \subseteq 2^{\mathcal{P}}$ is monotone if $\mathcal{A} \in \Gamma$ and $\mathcal{A} \subseteq \mathcal{B}$ imply that $\mathcal{B} \in \Gamma$. An access structure $\Gamma \subseteq 2^{\mathcal{P}}$ is a monotone collection of non-empty subsets of $\mathcal{P}$. Sets in $\Gamma$ are called authorized, and sets not in $\Gamma$ are called unauthorized.}
\end{definition}

If $\Gamma$ consists of all subsets of $\mathcal{P}$ with size greater than or equal to a fixed threshold $t$ $(1 \leq t \leq \ell)$, then $\Gamma$ is called a $t$-threshold access structure. In its most general form, an access structure can be any monotone NP language. This was first observed by Steven Rudich in private communications with Moni Naor \cite{Beimel[11],MoniNaor[06]}.

\begin{definition}[Closure]
	    \emph{Let $\mathcal{P}$ be a set of participants and $\mathcal{A} \in 2^{\mathcal{P}}$ . The closure of $\mathcal{A}$, denoted by cl$(\mathcal{A})$, is the set
		\[\text{cl}(\mathcal{A}) = \{\mathcal{C}: \mathcal{C}^* \subseteq \mathcal{C} \subseteq \mathcal{P} \text{ for some } \mathcal{C}^* \in \mathcal{A}\}.\]} 
\end{definition}

\begin{definition}[Minimal Authorized Subset]\label{GammaDef}
	\emph{For an access structure $\Gamma$, a family of minimal authorized subsets $\Gamma_0 \in \Gamma$ is defined as:
	\[
		\Gamma_0 = \{\mathcal{A} \in \Gamma: \mathcal{B} \not\subset \mathcal{A} \text{ for all } \mathcal{B} \in \Gamma \setminus \{ \mathcal{A} \}\}.
	\]   }
\end{definition}

Hence, the family of minimal access subsets $\Gamma_0$ uniquely determines the access structure $\Gamma$, and it holds that: $\Gamma =$ cl$(\Gamma_0)$, where cl denotes closure.

\begin{definition}[Computational Secret Sharing~\cite{Hugo[93]}]\label{def.1}
	\emph{A computational secret sharing scheme with respect to an access structure $\Gamma$, security parameter $\omega$, a set of $\ell$ polynomial-time parties $\mathcal{P} = \{P_1, \dots, P_\ell \}$, and a set of secrets $\mathcal{K}$, consists of a pair of polynomial-time algorithms, {\fontfamily{cmtt}\selectfont (Share,Recon)}, where: 
	\begin{itemize}
		\item {\fontfamily{cmtt}\selectfont Share} is a randomized algorithm that gets a secret $k \in \mathcal{K}$ and access structure $\Gamma$ as inputs, and outputs $\ell$ shares, $\{\mathrm{\Pi}^{(k)}_1, \dots, \mathrm{\Pi}^{(k)}_\ell\},$ of $k$,
		\item {\fontfamily{cmtt}\selectfont Recon} is a deterministic algorithm that gets as input the shares of a subset $\mathcal{A} \subseteq \mathcal{P}$, denoted by $\{\mathrm{\Pi}^{(k)}_i\}_{i \in \mathcal{A}}$, and outputs a string in $\mathcal{K}$,
	\end{itemize}
	such that, the following two requirements are satisfied:
	\begin{enumerate}
		\item \textit{Perfect Correctness:} for all secrets $k \in \mathcal{K}$ and every authorized subset $\mathcal{A} \in \Gamma$, it holds that:\\ Pr[{\fontfamily{cmtt}\selectfont Recon}$(\{\mathrm{\Pi}^{(k)}_i\}_{i \in \mathcal{A}}, \mathcal{A}) = k] = 1,$ 
		\item \textit{Computational Secrecy:} for every unauthorized subset $\mathcal{B} \notin \Gamma$ and all different secrets $k_1, k_2 \in \mathcal{K}$, it holds that the distributions $\{\mathrm{\Pi}_i^{(k_1)}\}_{i \in \mathcal{B}}$ and $\{\mathrm{\Pi}_i^{(k_2)}\}_{i \in \mathcal{B}}$ are computationally indistinguishable (w.r.t. $\omega)$.
	\end{enumerate}}
\end{definition}

\begin{remark}[Perfect Secrecy]\label{remark}
	If $\forall k_1, k_2 \in \mathcal{K}$ with $k_1 \neq k_2$, the distributions $\{\mathrm{\Pi}_i^{(k_1)}\}_{i \in \mathcal{B}}$ and $\{\mathrm{\Pi}_i^{(k_2)}\}_{i \in \mathcal{B}}$ are identical, then the scheme is called a perfect secret sharing scheme.
\end{remark}

Steven Rudich proved that if NP $\neq$ coNP, then efficient (i.e., polynomial-time) perfect secret sharing is impossible for Hamiltonian and monotone NP access structures, and efficient computational secret sharing is the best that we can do \cite{KomaYog[14]}.

\begin{definition}[Disjoint Union]\label{Dis}
	\emph{Let $n \geq 2$ be an integer and $\mathcal{H} = \{H_i: i \in [n]\}$ be a family of sets. Then, disjoint union of $\mathcal{H}$ is given as:
	\[\bigsqcup_{i \in [n]} H_i = \bigcup_{i \in [n]} \left\{(h,i): h \in H_i\right\}.\]}
\end{definition}

\subsubsection{\textbf{The Hybrid Argument}.}
The hybrid argument \cite{GoldMic[84]}, which is essentially the triangle inequality, is one of the most fundamental tools used in security proofs \cite{ArnoMarc[21]}. In cryptography, the canonical application of the hybrid argument is towards constructing the (inductive) arguments underlying various pseudorandom generators~\cite{BlumMicali[84],Yao[82],GoldLevin[89],HasRus[99],Naom[91],NaomAvi[92],Naom[92],RusNisan[94]}. Here, we give an informal introduction to the hybrid argument. For a formal account, we refer the interested reader to \cite{MarcArno[21]}. 

The hybrid argument is a technique to bound the closeness of two distributions, $D_0$ and $D_n$, via a polynomially long sequence of ``hybrids'', $D_0, D_1, \ldots, D_n$, which are constructed such that any two consecutive hybrids differ in exactly one feature. The central idea behind the hybrid argument is that if a (bounded or unbounded) distinguisher can distinguish the ``extreme hybrids'' $D_0$ and $D_n$, then it can also distinguish any adjacent hybrids $D_i$ and $D_{i+1}$, which it cannot do by the design of the hybrids. Therefore, the triangle inequality (for statistical or computational distance) can be used to obtain a bound on the distance between $D_0$ and $D_n$ by bounding the distance between neighboring distributions $D_i$ and $D_{i+1}$ for all $i \in \{0\} \cup [n - 1]$.

\subsubsection{\textbf{Set Systems with Restricted Intersections}.}
Extremal set theory is a field within combinatorics which deals with determining or estimating the size of set-systems, satisfying certain restrictions. The first result in extremal set theory was from Sperner~\cite{Sperner[28]} in 1928, establishing the maximum size of an antichain, i.e., a set-system where no member is a superset of another. But, it was Erd\H{o}s et al.'s pioneering work in 1961~\cite{Erdos[61]} that started systematic research on extremal set theory problems. Our work in this paper concerns a subfield of extremal set theory, called \textit{intersection theorems}, wherein set-systems under certain intersection restrictions are constructed, and bounds on their sizes are derived. We shall not give a full account of the known intersection theorems and mention only the results that are relevant to our set-system and its construction. For a broader account of intersection theorems over finite sets, we refer the interested reader to the comprehensive survey by Frankl and Tokushige~\cite{Frankl[16]}. For an introduction to intersecting and cross-intersecting families related to hypergraph coloring, please see~\cite{AMDD[2020]}. 

\begin{lemma}[\cite{Gro[00]}]\label{cor2}
	Let $m = \prod_{i=1}^{r} p_i^{\alpha_i}$ be a positive integer with $r > 1$ different prime divisors. Then there exists an explicitly constructible polynomial $Q$ with $n$ variables and degree $O(n^{1/r})$, which is equal to $0$ on $z = (1,1, \dots, 1) \in \{0,1\}^n$ but is nonzero $\bmod~ m$ on all other $z \in \{0,1\}^n$. Furthermore, $\forall z \in \{0,1\}^n$ and $\forall i \in \{1, \dots ,r\}$, it holds that: $Q(z) \in \{0,1\} \bmod p_i^{\alpha_i}$.
\end{lemma}

\begin{theorem}[\cite{Gro[00]}]\label{thm}
	Let $m$ be a positive integer, and suppose that $m$ has $r > 1$ different prime divisors: $m = \prod_{i=1}^{r} p_i^{\alpha_i}$. Then there exists $c = c(m) > 0$, such that for every integer $h > 0$, there exists an explicitly constructible uniform set-system $\mathcal{H}$ over a universe of $h$ elements such that:
	\begin{enumerate}
		\item $|\mathcal{H}| \geq \exp \left( c \dfrac{(\log h)^r}{(\log \log h)^{r-1}} \right)$,
		\item $\forall H \in \mathcal{H}:|H| = 0 \bmod m$,
		\item $\forall G, H \in \mathcal{H}, G \neq H:|G \cap H| \not= 0 \bmod m$.
	\end{enumerate}
\end{theorem}

\subsubsection{\textbf{Matching Vectors}.} 
A matching vector family is a combinatorial object that is defined as:

\begin{definition}[\cite{Zeev[11]}]
	\emph{Let $S \subseteq \mathbb{Z}_m \setminus \{0\}$, and $\langle \cdot, \cdot \rangle$ denote the inner product. We say that subsets $\mathcal{U} = \{\textbf{u}_i\}_{i=1}^N$ and $\mathcal{V} = \{\textbf{v}_i\}_{i=1}^N$ of vectors in $(\mathbb{Z}_m)^h$ form an $S$-matching family if the following two conditions are satisfied: 
		\begin{itemize}
			\item $\forall i \in [N],$ it holds that: $\langle \textbf{u}_i, \textbf{v}_i \rangle = 0 \bmod m$, 
			\item $\forall i,j \in [N]$ such that $i \neq j$, it holds that: $\langle \textbf{u}_i, \textbf{v}_j \rangle \bmod m \in S$.
	\end{itemize}}
\end{definition}

The question of bounding the size of matching vector families is closely related to the well-known extremal set theory problem of constructing set systems with restricted modular intersections. Matching vectors have found applications in the context of private information retrieval~\cite{Beimel[15],Beimel[12],Zeev[15],Zeev[11],Klim[09],Sergey[08],Liu[17]}, conditional disclosure of secrets~\cite{Liu[17]}, secret sharing~\cite{Liu[18]} and coding theory~\cite{Zeev[11]}. The first super-polynomial size matching vector family follows directly from the set-system constructed by Grolmusz~\cite{Gro[00]}. If each set $H$ in the set-system $\mathcal{H}$ defined by Theorem~\ref{thm} is represented by a vector $\textbf{u} \in (\mathbb{Z}_m)^h$, then it leads to the following family of $S$-matching vectors:
 
\begin{corollary}[to Theorem~\ref{thm}]
	For $h > 0$, suppose that a positive integer $m = \prod_{i=1}^{r} p_i^{\alpha_i}$ has $r > 1$ different prime divisors: $p_1, \ldots, p_r$. Then, there exists a set $S$ of size $2^{r}-1$ and a family of $S$-matching vectors \emph{$\{\textbf{u}_i\}$}${}^N_{i=1},$ \emph{$\textbf{u}_i$} $\in (\mathbb{Z}_m)^h$, such that, $N \geq \exp \left( c \dfrac{(\log h)^r}{(\log \log h)^{r-1}} \right)$.
\end{corollary}

\subsubsection{\textbf{Lattices}.}
A lattice $\mathrm{\Lambda}$ of $\mathbb{R}^w$ is defined as a discrete subgroup of $\mathbb{R}^w$. In cryptography, we are interested in integer lattices, i.e., $\mathrm{\Lambda} \subseteq \mathbb{Z}^w$. Given $w$-linearly independent vectors $\textbf{b}_1,\dots,\textbf{b}_w \in \mathbb{R}^w$, a basis of the lattice generated by them can be represented as the matrix $\mathbf{B} = (\textbf{b}_1,\dots,\textbf{b}_w) \in \mathbb{R}^{w \times w}$. The lattice generated by $\mathbf{B}$ is the following set of vectors:
\[\mathrm{\Lambda} = \mathcal{L}(\textbf{B}) = \left\{ \sum\limits_{i=1}^w c_i \textbf{b}_i: c_i \in \mathbb{Z} \right\}.\]
The lattices that are of particular interest in lattice-based cryptography are called \textit{q-ary} lattices, and they satisfy the following condition: $$q \mathbb{Z}^w \subseteq \mathrm{\Lambda} \subseteq \mathbb{Z}^w,$$ for some (possibly prime) integer $q$. In other words, the membership of a vector $\textbf{x}$ in $\mathrm{\Lambda}$ is determined by $\textbf{x}\bmod q$. Given a matrix $\textbf{A} \in \mathbb{Z}^{w \times n}_q$ for some integers $q, w, n,$ we can define the following two $n$-dimensional \textit{q-ary} lattices,

\[\mathrm{\Lambda}_q(\textbf{A}) = \{\textbf{y} \in \mathbb{Z}^n: \textbf{y} = \textbf{A}^T\textbf{s} \bmod q \text{ for some } \textbf{s} \in \mathbb{Z}^w \}, \]
\[\hspace{-28mm} \mathrm{\Lambda}_q^{\perp}(\textbf{A}) = \{\textbf{y} \in \mathbb{Z}^n: \textbf{Ay} = \textbf{0} \bmod q \}.\]

The first \textit{q-ary} lattice is generated by the rows of $\textbf{A}$; the second contains all vectors that are orthogonal (modulo $q$) to the rows of $\textbf{A}$. Hence, the first \textit{q-ary} lattice, $\mathrm{\Lambda}_q(\textbf{A})$, corresponds to the code generated by the rows of $\textbf{A}$ whereas the second, $\mathrm{\Lambda}_q^{\perp}(\textbf{A})$, corresponds to the code whose parity check matrix is $\textbf{A}$. For a complete introduction to lattices, we refer the interested reader to the monographs by Gr\"{a}tzer~\cite{Gratzer[03],Gratzer[09]}.

\subsubsection{\textbf{Lattices and Cryptography}.}
Problems in lattices have been of interest to cryptographers for decades with the earliest work dating back to 1997 when Ajtai and Dwork~\cite{Ajtai[97]} proposed a lattice-based public key cryptosystem following Ajtai's~\cite{Ajtai[96]} seminal worst-case to average-case reductions for lattice problems, wherein he showed that if there is no efficient algorithm that approximates the decision version of the Shortest Vector Problem (SVP) with a polynomial approximation factor, then it is hard to solve the associated search problem exactly over a random choice of the underlying lattice~\cite{Shafi[02]}. This reduction gave us the first cryptographically meaningful lattice-based hardness assumption, which became an essential component in proving the security of numerous lattice-based cryptographic constructions. For a detailed introduction to lattice-based cryptography, we refer the interested reader to \cite{KatzVadim[21],JiaZhen[20],MicciGold[02],DanMicci[09],PhongStern[01]}.

\subsubsection{\textbf{Learning with Errors}.}
The learning with errors (LWE) problem~\cite{Reg[05]} has emerged as the most popular hard problem for constructing lattice-based cryptographic solutions. The majority of practical LWE-based cryptosystems are derived from its variants such as ring LWE~\cite{Reg[10]}, module LWE~\cite{Ade[15]}, cyclic LWE~\cite{Charles[20]}, continuous LWE~\cite{Bruna[20]}, middle-product LWE~\cite{Miruna[17]}, group LWE~\cite{NicMal[16]}, entropic LWE \cite{ZviVin[16]} and polynomial-ring LWE~\cite{Damien[09]}. Many cryptosystems have been constructed whose security can be proved under the hardness of the LWE problem, including (identity-based, attribute-based, leakage-resilient, fully homomorphic, functional, public-key/key-encapsulation) encryption~\cite{AnaFan[19],KimSam[19],WangFan[19],Reg[05],Gen[08],Adi[09],Reg[10],Shweta[11],Vinod[11],Gold[13],Jan[18],Hayo[19],Bos[18],Bos[16],WBos[15],Brak[14],Fan[12],Joppe[13],Adriana[12],Lu[18]}, oblivious transfer~\cite{Pei[08],Dott[18],Quach[20]}, (blind) signatures~\cite{Gen[08],Vad[09],Markus[10],Vad[12],Tesla[20],Dili[17],FALCON[20]}, pseudorandom functions with special algebraic properties~\cite{Ban[12],Boneh[13],Ban[14],Ban[15],Zvika[15],Vipin[19],KimDan[17],RotBra[17],RanChen[17],KimWu[17],KimWu[19],Qua[18]}, hash functions~\cite{Katz[09],Pei[06]}, secure matrix multiplication~\cite{Dung[16],Wang[17]}, classically verifiable quantum computation~\cite{Urmila[18]}, noninteractive zero-knowledge proof system for (any) NP language~\cite{Sina[19]}, obfuscation~\cite{Huijia[16],Gentry[15],Hal[17],ZviVin[16],AnanJai[16],CousinDi[18]}, multilinear maps \cite{Grg[13],Gentry[15],Gu[17]}, lossy-trapdoor functions \cite{BellKil[12],PeiW[08],HoWee[12]}, and many more~\cite{Peikert[16],Hamid[19]}. 

\begin{definition}[Decision-LWE~\cite{Reg[05]}]\label{decisionLWE}
	\emph{For positive integers $n$ and $q \geq 2$, and an error (probability) distribution $\chi = \chi(n)$ over $\mathbb{Z}_q$, the decision-LWE${}_{n, q, \chi}$ problem is to distinguish between the following pairs of distributions: 
		\[(\textbf{A}, \textbf{A} \textbf{s} + \textbf{e}) \quad \text{and} \quad (\textbf{A}, \textbf{u}),\] 
		where $\textbf{A} \xleftarrow{\; \$ \;} \mathbb{Z}^{w \times n}_q$, $w = \poly(n)$, $\textbf{s} \in \mathbb{Z}^n_q$, $\textbf{e} \xleftarrow{\; \$ \;} \chi^w$ and $\textbf{u} \xleftarrow{\; \$ \;} \mathbb{Z}^w_q$.}
\end{definition} 

\begin{definition}[Search-LWE~\cite{Reg[05]}]\label{searchLWE}
	\emph{For positive integers $n$ and $q \geq 2$, and an error (probability) distribution $\chi = \chi(n)$ over $\mathbb{Z}_q$, the search-LWE${}_{n, q, \chi}$ problem is to recover $\textbf{s} \in \mathbb{Z}^n_q$, given $(\textbf{A}, \textbf{A} \textbf{s} + \textbf{e})$, where $\textbf{A} \xleftarrow{\; \$ \;} \mathbb{Z}^{w \times n}_q$, $\textbf{s} \in \mathbb{Z}^n_q$, $\textbf{e} \xleftarrow{\; \$ \;} \chi^w$ and $w = \poly(n)$.}
\end{definition}

Regev~\cite{Reg[05]} showed that for certain noise distributions and a sufficiently large $q$, the LWE problem is as hard as the worst-case SIVP and GapSVP under a quantum reduction (see~\cite{Pei[09],Zvika[13]} for classical hardness arguments). Regev's results were extended to establish that the fixed vector $\textbf{s}$ can be sampled from a low norm distribution (in particular, from the noise distribution $\chi)$ and the resulting problem is as hard as the original LWE problem~\cite{Benny[09]}. Later, it was discovered that $\chi$ can also be a simple low-norm distribution~\cite{Micci[13]}. Therefore, a standard hybrid argument can be used to get to \textit{multi-secret} LWE, which asks to distinguish $(\textbf{A}, \textbf{B} = \textbf{A} \textbf{S} + \textbf{E})$ from $(\textbf{A}, \textbf{U})$ for $\textbf{A} \xleftarrow{\; \$ \;} \mathbb{Z}^{w \times n}_q$, $\textbf{S} \in \mathbb{Z}_q^{n \times n} \text{ or } \textbf{S} \in \chi^{n \times n}$, $\textbf{E} \xleftarrow{\; \$ \;} \chi^{w \times n}$, and a uniformly sampled $\textbf{U} \in \mathbb{Z}^{w \times n}_q$. It is easy to verify that up to a $w$ factor loss in the distinguishing advantage, multi-secret LWE is equivalent to plain (single-secret) decision-LWE. Lattice reduction algorithms, which are the most powerful tools against LWE, remain (practically) inefficient in solving LWE~\cite{Ajtai[01],Fincke[85],Gama[06],Gama[13],Gama[10],LLL[82],Ngu[10],Vid[08],DanPan[10],DaniPan[10],Poha[81],Schnorr[87],Schnorr[94],Schnorr[95],Nguyen[09]}. 

\subsubsection{\textbf{Trapdoors for Lattices}.}\label{SecTrap}
Trapdoors for lattices have been studied in~\cite{Ajtai[99],Micci[12],Gen[08],Chen[19],Boyen[17],Hof[12],PeiW[08],Lyub[15]}. We recall the definition from~\cite{Micci[12]} as that is the algorithm used in our scheme.
\begin{definition}
	\emph{Let $n \geq wd$ be an integer and $\bar{n} = n - wd$. For $\textbf{A} \in \mathbb{Z}^{w \times n}_q$, we say that $\textbf{R} \in \mathbb{Z}^{\bar{n} \times wd}_q$ is a trapdoor for $\textbf{A}$ with tag $\textbf{H} \in \mathbb{Z}^{w \times w}_q$ if $\textbf{A}\begin{bsmallmatrix} {\scriptstyle\textbf{R}} \\ {\scriptstyle \textbf{I}} \end{bsmallmatrix} = \textbf{H} \cdot \textbf{G}$, where $\textbf{G} \in \mathbb{Z}^{w \times wd}_q$ is a primitive matrix.}
\end{definition}

Given a trapdoor $\textbf{R}$ for $\textbf{A}$, and an LWE instance $\textbf{B} = \textbf{A} \textbf{S} + \textbf{E} \bmod q$ for some ``short'' error matrix $\textbf{E}$, the LWE inversion algorithm from~\cite{Micci[12]} successfully recovers $\textbf{S}$ (and $\textbf{E}$) with overwhelming probability.

\subsubsection{\textbf{STCON}.}
STCON (s-t connectivity) in a directed graph can be defined as the following function: the input is a directed graph $G$. The graph contains two designated nodes, $s$ and $t$. The function outputs $1$ if and only if $G$ has a directed path from $s$ to $t$. Karchmer and Wigderson~\cite{Karch[93]} showed that there exists an efficient linear secret sharing scheme for the analogous function where the graph is undirected. In a linear secret sharing scheme~\cite{KarninGreene[83]}, share generation and secret reconstruction are performed by evaluating linear maps and solving linear systems of equations. Later, Beimel and Paskin~\cite{Beimel[08]} extended those results to linear secret sharing schemes for STCON in directed graphs. It is known that (linear) secret sharing schemes based on undirected STCON have strictly smaller share size than those based on directed STCON~\cite{AjtaiFagin[90],Karch[93],Beimel[08]}.

\section{Access Structure Hiding Verifiable Secret Sharing}\label{Sec3}
In this section, we give a formal definition of an access structure hiding verifiable (computational) secret sharing scheme. 
\begin{definition}\label{MainDef}
	 \emph{An access structure hiding verifiable (computational) secret sharing scheme with respect to an access structure $\Gamma$, a set of $\ell$ polynomial-time parties $\mathcal{P} = \{P_1, \dots, P_\ell\}$, a set of secrets $\mathcal{K}$ and a security parameter $\omega$, consists of two sets of polynomial-time algorithms, {\fontfamily{cmtt}\selectfont(HsGen, HsVer)} and {\fontfamily{cmtt}\selectfont(VerShr, Recon, Ver)}, which are defined as:
		\begin{enumerate}
			\item {\fontfamily{cmtt}\selectfont VerShr} is a randomized algorithm that gets a secret $k \in \mathcal{K}$ and access structure $\Gamma$ as inputs, and outputs $\ell$ shares, $\{\mathrm{\Psi}^{(k)}_1, \dots, \mathrm{\Psi}^{(k)}_\ell\},$ of $k$,
			\item {\fontfamily{cmtt}\selectfont Recon} is a deterministic algorithm that gets as input the shares of a subset $\mathcal{A} \subseteq \mathcal{P}$, denoted by $\{\mathrm{\Psi}^{(k)}_i\}_{i \in \mathcal{A}}$, and outputs a string in $\mathcal{K}$,
			\item {\fontfamily{cmtt}\selectfont Ver} is a deterministic Boolean algorithm that gets $\{\mathrm{\Psi}^{(k)}_i\}_{i \in \mathcal{A}}$ and a secret $k' \in \mathcal{K}$ as inputs and outputs $b \in \{0,1\}$, 
		\end{enumerate}
		such that the following three requirements are satisfied:
		\begin{enumerate}[label=(\alph*)]
			\item \textit{Perfect Correctness:} for all secrets $k \in \mathcal{K}$ and every authorized subset $\mathcal{A} \in \Gamma$, it holds that:\\ Pr[{\fontfamily{cmtt}\selectfont Recon}$(\{\mathrm{\Psi}^{(k)}_i\}_{i \in \mathcal{A}}, \mathcal{A}) = k] = 1,$ 
			\item \textit{Computational Secrecy:} for every unauthorized subset $\mathcal{B} \notin \Gamma$ and all different secrets $k_1, k_2 \in \mathcal{K}$, it holds that the distributions $\{\mathrm{\Psi}_i^{(k_1)}\}_{i \in \mathcal{B}}$ and $\{\mathrm{\Psi}_i^{(k_2)}\}_{i \in \mathcal{B}}$ are computationally indistinguishable (w.r.t. $\omega)$,
			\item \textit{Computational Verifiability:} every authorized subset $\mathcal{A} \in \Gamma$ can use {\fontfamily{cmtt}\selectfont Ver} to verify whether its set of shares $\{\mathrm{\Psi}^{(k)}_i\}_{i \in \mathcal{A}}$ is consistent with a given secret $k \in \mathcal{K}$. Formally, for a negligible function $\epsilon$, it holds that: 
			\begin{itemize}
				\item Pr[{\fontfamily{cmtt}\selectfont Ver}$(k, \{\mathrm{\Psi}^{(k)}_i\}_{i \in \mathcal{A}}) = 1] = 1 - \epsilon(\omega)$ if all shares $\mathrm{\Psi}^{(k)}_i \in \{\mathrm{\Psi}^{(k)}_i\}_{i \in \mathcal{A}}$ are consistent with the secret $k$, 
				\item else, if any share $\mathrm{\Psi}^{(k)}_i \in \{\mathrm{\Psi}^{(k)}_i\}_{i \in \mathcal{A}}$ is inconsistent with the secret $k$, then it holds that:\\ Pr[{\fontfamily{cmtt}\selectfont Ver}$(k, \{\mathrm{\Psi}^{(k)}_i\}_{i \in \mathcal{A}}) = 0] = 1 - \epsilon(\omega)$,
			\end{itemize}
		\end{enumerate}
		\begin{enumerate}
		 \setcounter{enumi}{3}			
			\item {\fontfamily{cmtt}\selectfont HsGen} is a randomized algorithm that gets $\mathcal{P}$ and $\Gamma$ as inputs, and outputs $\ell$ \textit{access structure tokens} $\{\mathrm{\mho}^{(\Gamma)}_1, \dots, \mathrm{\mho}^{(\Gamma)}_\ell\},$ 
			\item {\fontfamily{cmtt}\selectfont HsVer} is a deterministic algorithm that gets as input the \textit{access structure tokens} of a subset $\mathcal{A} \subseteq \mathcal{P}$, denoted by $\{\mathrm{\mho}_i^{(\Gamma)}\}_{i \in \mathcal{A}}$, and outputs $b \in \{0,1\}$,
		\end{enumerate}
		such that, the following three requirements are satisfied:
		\begin{enumerate}[label=(\alph*)]
			\item \textit{Perfect Completeness:} every authorized subset of parties $\mathcal{A} \in \Gamma$ can identify itself as a member of the access structure $\Gamma$, i.e., it holds that: Pr[{\fontfamily{cmtt}\selectfont HsVer}$(\{\mathrm{\mho}_i^{(\Gamma)}\}_{i \in \mathcal{A}}) = 1] = 1,$
			\item \textit{Perfect Soundness:} every unauthorized subset of parties $\mathcal{B} \notin \Gamma$ can identify itself to be outside of the access structure $\Gamma$, i.e., it holds that: Pr[{\fontfamily{cmtt}\selectfont HsVer}$(\{\mathrm{\mho}_i^{(\Gamma)}\}_{i \in \mathcal{B}}) = 0] = 1,$
			\item \textit{Statistical Hiding:} for all access structures $\Gamma, \Gamma' \subseteq 2^{\mathcal{P}}$, where $\Gamma \neq \Gamma'$, and each subset of parties $\mathcal{B} \notin \Gamma, \Gamma'$ that is unauthorized in both $\Gamma$ and $\Gamma'$, it holds that:
			\[\left| \Pr[\Gamma~|~ \{\mathrm{\mho}_i^{(\Gamma)}\}_{i \in \mathcal{B}}, \{\mathrm{\Psi}_i^{(k)}\}_{i \in \mathcal{B}}] - \Pr[\Gamma'~|~ \{\mathrm{\mho}_i^{(\Gamma)}\}_{i \in \mathcal{B}}, \{\mathrm{\Psi}_i^{(k)}\}_{i \in \mathcal{B}}] \right| = 2^{-\omega}.\]
	\end{enumerate}}
\end{definition}

\section{Novel Set-Systems and Vector Families}\label{construction}
In this section, we prove Theorem~\ref{mainThm} by constructing a novel set-system.
	\begin{proposition}
		\label{main_construction}
		Let $l \geq 2$ be an integer, and $m = \prod_{i=1}^{r} p_i^{\alpha_i}$ be a positive integer with $r > 1$ different prime divisors such that $\forall i \in \{1, \dots, r\}: p_i > l$. Suppose there exists an integer $t\geq 2$ and a uniform set-system $\mathcal{G}$ satisfying the conditions:
		\begin{enumerate}
			\item $\forall G\in\mathcal{G}:\vert G\vert = 0\bmod m$,
			\item $\forall t^\prime$ such that $2\leq t^\prime\leq t$, and for all distinct $G_1,\,G_2,\,\ldots,\,G_{t^\prime}\in\mathcal{G}$, it holds that:
			$$\left\vert\bigcap_{\tau=1}^{t^\prime} G_\tau\right\vert = \mu\bmod m,$$
			where $\mu \neq 0\bmod m$ and $\forall i \in \{1, \dots, r\}: \mu \in \{0,1\} \bmod p_i$,
			\item $\left\vert\bigcap_{G\in\mathcal{G}}G\right\vert\neq 0\bmod m$.
		\end{enumerate}
		Then, there exists a set-system $\mathcal{H}$ that is explicitly constructible from the set-system $\mathcal{G}$ such that:
		\begin{enumerate}[label=(\roman*)]			
			\item\label{L1I1} $\forall H_1, H_2\in\mathcal{H}$, either $\vert H_1\vert=\vert H_2\vert$, $\vert H_1\vert=l\vert H_2\vert$ or $l\vert H_1\vert=\vert H_2\vert$,
			\item\label{L1I2} $\mathcal{H}$ has $t$-wise restricted intersections modulo $m$ (see Definition~\ref{def2}).
		\end{enumerate}
	\end{proposition}
	
	\begin{proof}
		We start with $l$ uniform\footnote{all member sets have equal size} set systems $\mathcal{H}_1,\,\mathcal{H}_2,\,\ldots,\,\mathcal{H}_l$ satisfying the following properties:
		\begin{enumerate}
			\item $\forall H^{(i)}\in\mathcal{H}_i:\vert H^{(i)}\vert = 0\bmod m$,
			\item $\forall t^\prime$ such that $2\leq t^\prime\leq t$, and for all distinct $H^{(i)}_1,\,H^{(i)}_2,\,\ldots,\,H^{(i)}_{t^\prime}\in\mathcal{H}_i$, it holds that:
			$$\left\vert\bigcap_{\tau=1}^{t^\prime} H^{(i)}_\tau\right\vert = \mu\bmod m,$$
			where $\mu\neq 0\bmod m$ and $\forall z \in \{1, \dots, r\}: \mu \in \{0,1\} \bmod p_z$,
			\item $\forall i \in \{1, \dots, l\}: \left\vert\bigcap_{H^{(i)}\in\mathcal{H}_i}H^{(i)}\right\vert\neq 0\bmod m$,
			\item $\left\vert H^{(i)}\right\vert=\left\vert H^{(j)}\right\vert$ for all $H^{(i)}\in\mathcal{H}_i$, $H^{(j)}\in\mathcal{H}_j$,
			\item $\forall i,j \in \{1, \dots, l\}: \left\vert\bigcap_{H^{(i)}\in\mathcal{H}_i}H^{(i)}\right\vert=\left\vert\bigcap_{H^{(j)}\in\mathcal{H}_j}H^{(j)}\right\vert$.
		\end{enumerate}
		
		We begin by fixing bijections:
		$$f_{i,j}:\bigcap_{H^{(i)}\in\mathcal{H}_i}H^{(i)}\to \bigcap_{H^{(j)}\in\mathcal{H}_j}H^{(j)},$$
		such that $f_{i,i}$ is the identity and $f_{i,j}\circ f_{j,k}=f_{i,k}$ for all $1\leq i,j,k\leq l$. Using these bijections, we can identify the sets $\bigcap_{H^{(i)}\in\mathcal{H}_i}H^{(i)}$ and $\bigcap_{H^{(j)}\in\mathcal{H}_i}H^{(j)}$ with each other. Let:
		$$A=\bigcap_{H^{(1)}\in\mathcal{H}_1}H^{(1)}=\bigcap_{H^{(2)}\in\mathcal{H}_2}H^{(2)}=\cdots=\bigcap_{H^{(l)}\in\mathcal{H}_l}H^{(l)}.$$
		We shall treat the elements of the sets in $\mathcal{H}_i$ as being distinct from the elements of the sets in $\mathcal{H}_j$, except for the above identification of elements in $\bigcap_{H^{(i)}\in\mathcal{H}_i}H^{(i)}$ with elements in $\bigcap_{H^{(j)}\in\mathcal{H}_j}H^{(j)}$. Let $a=\vert A\vert$, and let $\beta_1,\,\beta_2,\,\ldots,\,\beta_{(l-1)a}$ be elements that are distinct from all the elements in the sets in $\mathcal{H}_1,\,\mathcal{H}_2,\,\ldots\,\,\mathcal{H}_l$. Define the set:
		$$B=\{\beta_1,\,\beta_2,\,\ldots,\,\beta_{(l-1)a}\},$$
		and consider a set system $\mathcal{H}$ which contains the following sets:
		\begin{itemize}
			\item $H^{(i)}$, where $H^{(i)}\in\mathcal{H}_i$ for some $i \in [l]$,
			\item $\bigcup_{i=1}^lH^{(i)}\cup B$, where $H^{(i)}\in\mathcal{H}_i$ for all $i\in [l]$.
		\end{itemize}
		Write the common size of the sets in the uniform set systems $\mathcal{H}_i~(1 \leq i \leq l)$ as $km$ for some $k>0$. Then, the following holds for all $H^{(i)}\in\mathcal{H}_i$,
		\begin{align*}
		\left\vert\bigcup_{i=1}^lH^{(i)}\cup B\right\vert&=\left\vert\bigcup_{i=1}^lH^{(i)}\right\vert+\vert B\vert=\sum_{i=1}^l\vert H^{(i)}\vert-(l-1)\vert A\vert+\vert B\vert\\
		&=l(km)-(l-1)a+(l-1)a=lkm,
		\end{align*}
		where the second equality comes from the fact that $H^{(i)}\cap H^{(j)}=A$ for all $i\neq j$. This proves that Condition~\ref{L1I1} holds. Moving on to the Condition~\ref{L1I2}: let $t_1,\,t_2,\,\ldots,\,t_{l+1}\geq 0$ be such that $2\leq t^\prime(=t_1+t_2+\cdots+t_{l+1})\leq t$. We shall consider the intersection of the sets:
		\begin{itemize}
			\item $H^{(i)}_\tau$ where $1\leq i\leq l$, $1\leq\tau\leq t_i$ and $H^{(i)}_\tau\in\mathcal{H}_i$,
			\item $\bigcup_{i=1}^l H_\tau^{\prime(i)}\cup B$ where $1\leq\tau\leq t_{l+1}$ and $H_\tau^{\prime(i)}\in\mathcal{H}_i$.
		\end{itemize}
		Assume that these sets form a non-degenerate family. Let:
		\begin{align*}
		\sigma=&\left\vert\bigcap_{i=1}^l\bigcap_{\tau=1}^{t_i} H^{(i)}_\tau\cap\bigcap_{\tau=1}^{t_{l+1}}(H_{\tau}^{\prime (1)}\cup H_\tau^{\prime (2)}\cup\cdots\cup H_\tau^{\prime (l)}\cup B)\right\vert\\
		=&\left\vert\bigcap_{i=1}^l\bigcap_{\tau=1}^{t_i}H^{(i)}_\tau\cap\bigcap_{\tau=1}^{t_{l+1}}(H_\tau^{\prime (1)}\cup H_\tau^{\prime (2)}\cup\cdots\cup H_\tau^{\prime (l)})\right\vert+\epsilon\vert B\vert,
		\end{align*}
		where $\epsilon=1$ if $t_1=t_2=\cdots=t_l=0$, and $\epsilon=0$ otherwise. If two or more of $t_1,\,t_2,\,\ldots,\,t_l$ are non-zero, then: $\sigma=\vert A\vert=a \neq 0\bmod m$. On the other hand, if exactly one of $t_1,\,t_2,\,\ldots,\,t_l$ is non-zero, then:
		$$\sigma=\left\vert\bigcap_{\tau=1}^{t_i}H^{(i)}_\tau\cap\bigcap_{\tau=1}^{t_{l+1}}H_{\tau}^{\prime (i)}\right\vert\neq 0\bmod m$$
		since $H^{(i)}_\tau$ (for $1\leq\tau\leq t_i$) and $H_\tau^{\prime (i)}$ (for $1\leq\tau\leq t_{l+1}$) are not all the same by the assumption of non-degeneracy. If $t_1=t_2=\cdots=t_l=0$, then we get:
		\begin{align*}
		\sigma&=\left\vert\bigcap_{\tau=1}^{t_{l+1}}(H_\tau^{\prime (1)}\cup H_\tau^{\prime (2)}\cup\cdots\cup H_\tau^{\prime (l)})\right\vert+\vert B\vert\\
		&=\sum_{i=1}^{l}\left\vert\bigcap_{\tau=1}^{t_{l+1}}H_\tau^{\prime (i)}\right\vert-(l-1)\vert A\vert+\vert B\vert = \sum_{i=1}^{l^\prime}\mu_i\bmod m,
		\end{align*}
		for some integer $l'$ such that $1\leq l^\prime\leq l$, and some set $\{\mu_i\}_{i=1}^{l'}$ such that for each $\mu_i$ and all primes $p$ such that $p ~|~ m$, it holds that: $\mu_i \in \{0,1\} \bmod p$. Since $\mu_i \neq 0\bmod m$ for all $1\leq i\leq l^\prime$, there must be some prime factor $p$ of $m$ for which at least one of the $\mu_i$'s satisfy $\mu_i = 1\bmod p$. Since $p$ is a prime factor of $m$, it satisfies: $p>l\geq l^\prime$. Hence, for $p$, we get:
		$$\sigma = \sum_{i=1}^{l^\prime}\mu_i \neq 0\bmod p.$$
		This proves Condition~\ref{L1I2}, and hence completes the proof. $\qed$
	\end{proof}
	
	\begin{remark}\label{remarkImp}
			Suppose that $\vert\mathcal{G}\vert=s$ and that the number of elements in the universe of $\mathcal{G}$ is $g$. Then, there are $ls$ sets of size $km$ and $s^l$ sets of size $lkm$ in $\mathcal{H}$. Therefore, we get: $\vert\mathcal{H}\vert=s^l+ls$. The universe of $\mathcal{H}$ has $lg$ elements, and for each $H\in\mathcal{H}$, exactly one of the following is true:
			\begin{itemize}
				\item $H$ is a proper subset of exactly $s^{l-1}$ sets and not a proper superset of any sets in $\mathcal{H}$,
				\item $H$ is a proper superset of exactly $l$ sets and not a proper subset of any sets in $\mathcal{H}$.
			\end{itemize}
	\end{remark}
	
	In order to explicitly construct set systems which, in addition to having the properties in Proposition \ref{main_construction}, have sizes superpolynomial in the number of elements, we first recall a result of Barrington et al. \cite{Barr[94]}, which Grolmusz \cite{Gro[00]} used to construct a superpolynomial uniform set-system.
	
	\begin{theorem}[\cite{Barr[94]}, Theorem 2.1]\label{Thm5}
		\label{BBR_theorem}
		Let $\{\alpha_i\}_{i=1}^r$ be $r > 1$ positive integers and $m = \prod_{i=1}^{r} p_i^{\alpha_i}$ be a positive integer with $r$ different prime divisors: $p_1, \dots, p_r$. For every integer $n\geq 1$, there exists an explicitly constructible polynomial $P$ in $n$ variables such that
		\begin{enumerate}
			\item $P(0,\,0,\,\ldots,\,0) = 0\bmod m$,
			\item $P(x)\neq 0\bmod m$ for all $x\in\{0,1\}^n$ such that $x\neq(0,\,0,\,\ldots,\,0)$,
			\item $\forall i \in [r]$ and $\forall x\in\{0,1\}^n$ such that $x\neq(0,\,0,\,\ldots,\,0)$, it holds that: $P(x) \in \{0,1\} \bmod p_i$. 
		\end{enumerate}
		The polynomial $P$ has degree $d=\max(p_1^{e_1},\,\ldots,\,p_r^{e_r})-1$ where $e_i~(\forall i \in [r])$ is the smallest integer that satisfies $p_i^{e_i}>\lceil n^{1/r}\rceil$.
	\end{theorem}
	
	Define $Q(x_1,\,x_2,\,\ldots,\,x_n)=P(1-x_1,\,1-x_2,\,\ldots,\,1-x_n)$. Then:
	\begin{enumerate}
		\item $Q(1,\,1,\,\ldots,\,1) = 0\bmod m$,
		\item $Q(x)\neq 0\bmod m$ for all $x\in\{0,1\}^n$ such that $x\neq(1,\,1,\,\ldots,\,1)$.
		\item $\forall i \in [r]$ and $\forall x\in\{0,1\}^n$ such that $x\neq(1,\,1,\,\ldots,\,1)$, it holds that: $Q(x) \in \{0,1\} \bmod p_i$.
	\end{enumerate}
	
	\begin{theorem}[\cite{Gro[00]}, Theorem 1.4, Lemma 3.1]
		\label{grolmusz}
		Let $\{\alpha_i\}_{i=1}^r$ be $r > 1$ positive integers and $m = \prod_{i=1}^{r} p_i^{\alpha_i}$ be a positive integer with $r$ different prime divisors: $p_1, \dots, p_r$. For every integer $n\geq 1$, there exists a uniform set system $\mathcal{G}$ over a universe of $g$ elements which is explicitly constructible from the polynomial $Q$ of degree $d$ such that
		\begin{enumerate}
			\item $g<\frac{2(m-1)n^{2d}}{d!}$ if $n\geq 2d$,
			\item $\vert\mathcal{G}\vert=n^n$,
			\item $\forall G\in\mathcal{G}$, $\vert G\vert = 0\bmod m$,
			\item $\forall G, H\in\mathcal{G}$ such that $G\neq H$, it holds that: $\vert G\cap H\vert = \mu\bmod m$, where $\mu\neq 0\bmod m$ and $\mu \in \{0,1\} \bmod p_i$ for all $i \in [r]$,
			\item\label{T5} $\left\vert\bigcap_{G\in\mathcal{G}}G\right\vert\neq 0\bmod m$.
		\end{enumerate}
	\end{theorem}
	
	Note that Condition~\ref{T5} follows from the fact that the following holds in Grolmusz's construction of superpolynomial set-systems:
	$$\left\vert\bigcap_{G\in\mathcal{G}}G\right\vert = Q(0,\,0,\,\ldots,\,0)\neq 0\bmod m.$$
	
	In fact, a straightforward generalization of the arguments in \cite{Gro[00]} proves the following theorem:
	\begin{theorem}
		\label{grolmusz_k_intersections}
		Let $\{\alpha_i\}_{i=1}^r$ be $r > 1$ positive integers and $m = \prod_{i=1}^{r} p_i^{\alpha_i}$ be a positive integer with $r$ different prime divisors: $p_1, \dots, p_r$. For all integers $t\geq 2$ and $n\geq 1$, there exists a uniform set system $\mathcal{G}$ over a universe of $g$ elements which is explicitly constructible from the polynomial $Q$ of degree $d$ such that
		\begin{enumerate}
			\item $g<\frac{2(m-1)n^{2d}}{d!}$ if $n\geq 2d$,
			\item $\vert\mathcal{G}\vert=n^n$,
			\item $\forall G\in\mathcal{G}$, $\vert G\vert = 0\bmod m$,
			\item $\forall t^\prime$ such that $2\leq t^\prime\leq t,$ and for all distinct $G_1,\,G_2,\,\ldots,\,G_{t^\prime}\in\mathcal{G}$, it holds that:
			$$\left\vert\bigcap_{\tau=1}^{t^\prime}G_\tau\right\vert = \mu\bmod m,$$
			where $\mu\neq 0\bmod m$ and $\mu \in \{0,1\} \bmod p_i$ for all $i \in [r]$,
			\item $\left\vert\bigcap_{G\in\mathcal{G}}G\right\vert\neq 0\bmod m$.
		\end{enumerate}
	\end{theorem}
	
	\begin{proof}
		We will follow the proof of Theorem 1.4 in \cite{Gro[00]}, but with a few minor changes. Write the polynomial $Q$ as
		$$Q(x_1,\,x_2,\,\ldots,\,x_n)=\sum_{i_1<i_2<\cdots<i_l}a_{i_1,\,i_2,\,\ldots,\,i_l}x_{i_1}x_{i_2}\cdots x_{i_l}$$
		Define
		$$\tilde{Q}(x_1,\,x_2,\,\ldots,\,x_n)=\sum_{i_1<i_2<\cdots<i_l}\tilde{a}_{i_1,\,i_2,\,\ldots,\,i_l}x_{i_1}x_{i_2}\cdots x_{i_l}$$
		where $\tilde{a}_{i_1,\,i_2,\,\ldots,\,i_l}$ is the remainder when $a_{i_1,\,i_2,\,\ldots,\,i_l}$ is divided by $m$.
		
		Let $[0,\,n-1]=\{0,\,1,\,\ldots,\,n-1\}$. Define the function $\delta:[0,\,n-1]^t\to\{0,1\}$ as
		$$\delta(u_1,\,u_2,\,\ldots,\,u_t)=
		\begin{cases}
		1 & \text{if }u_1=u_2=\cdots=u_t,\\
		0 & \text{otherwise}.
		\end{cases}
		$$
		
		For $y_1,\,y_2,\,\ldots,\,y_t\in [0,\,n-1]^n$, let
		$$a^{y_1,\,y_2,\,\ldots,\,y_t}=\tilde{Q}\left(\delta(y_{1,1},\,y_{2,1},\,\ldots,\,y_{t,1}),\,\ldots,\,\delta(y_{1,n},\,y_{2,n},\,\ldots,\,y_{t,n})\right)\bmod m.$$
		Then
		$$a^{y_1,\,y_2,\,\ldots,\,y_t}=\sum b^{y_1,\,y_2,\,\ldots,\,y_t}_{i_1,\,i_2,\,\ldots,\,i_l}$$
		where
		$$b^{y_1,\,y_2,\,\ldots,\,y_t}_{i_1,\,i_2,\,\ldots,\,i_l}=\prod_{j=1}^{l}\delta(y_{1,i_j},\,y_{2,i_j},\,\ldots,\,y_{t,i_j}).$$
		Each summand $b^{y_1,\,y_2,\,\ldots,\,y_t}_{i_1,\,i_2,\,\ldots,\,i_l}$ corresponds to a monomial of $\tilde{Q}$ and occurs with multiplicity $\tilde{a}_{i_1,\,i_2,\,\ldots,\,i_l}$ in the above sum.
		
		It is easy to check that there exists partitions $\mathcal{P}_{i_1,\,i_2,\,\ldots,\,i_l}$ of $[0,\,n-1]^n$ such that for all $y_1,\,y_2,\,\ldots,\,y_t\in [0,\,n-1]^n$,
		$$b^{y_1,\,y_2,\,\ldots,\,y_t}_{i_1,\,i_2,\,\ldots,\,i_l}=
		\begin{cases}
		1 & \text{if }y_1,\,y_2,\,\ldots,\,y_t\text{ belong to the same block of }\mathcal{P}_{i_1,\,i_2,\,\ldots,\,i_l},\\
		0 & \text{otherwise},
		\end{cases}$$
		and that the equivalence classes defined by the partition $\mathcal{P}_{i_1,\,i_2,\,\ldots,\,i_l}$ each has size $n^{n-l}$. We say that a block in the partition $\mathcal{P}_{i_1,\,i_2,\,\ldots,\,i_l}$ covers $y\in [0,\,n-1]^n$ if $y$ is an element of the block.
		
		We define a set system $\mathcal{G}$ as follows: the sets in $\mathcal{G}$ correspond to $y$ for $y\in [0,\, n-1]^n$, and the set corresponding to $y$ has elements given by the blocks that cover $y$.
		
		The set $y$ in the set system $\mathcal{G}$ has size equal to the number of blocks that cover $y$, which is equal to
		$$a^{y,\,y,\,\ldots,\,y}=\tilde{Q}(1,\,1,\,\ldots,\,1) = 0\bmod m.$$
		For any $2\leq t^\prime\leq t$, and $y_1,\,y_2,\,\ldots,\,y_{t^\prime}\in [0,\,n-1]^n$ distinct, some block of $\mathcal{P}_{i_1,\,i_2,\,\ldots,\,i_l}$ covers all of $y_1,\,y_2,\,\ldots,\,y_{t^\prime}$ if and only if $b^{y_1,\,y_2,\,\ldots,\,y_{t^\prime},\,\ldots,\,y_{t^\prime}}_{i_1,\,i_2,\,\ldots,\,i_l}=1$ (note that $y_{t^\prime}$ occurs in the superscript $t-t^\prime+1$ times). Hence, the number of such blocks is equal to:
		$$a^{y_1,\,y_2,\,\ldots,\,y_{t^\prime},\,\ldots,\,y_{t^\prime}}\neq 0\bmod m.$$
		
		Finally, we would like to have a bound on $g$, the number of elements in the universe of $\mathcal{G}$. By our construction, this is equal to the number of blocks. Since the partition $\mathcal{P}_{i_1,\,i_2,\,\ldots,\,i_l}$ defines $n^l$ equivalence classes, the number of blocks is given by
		\begin{align*}
		g=\sum_{i_1<i_2<\cdots<i_l}\tilde{a}_{i_1,\,i_2,\,\ldots,\,i_l}n^l&\leq\sum_{l=0}^d\binom{n}{l}(m-1)n^l<(m-1)\sum_{l=0}^d\frac{n^{2l}}{l!}\\
		&<\frac{2(m-1)n^{2d}}{d!},
		\end{align*}
		provided that $n\geq 2d$. $\qed$
	\end{proof}
	
	\begin{theorem}
		\label{superpolynomial_set_systems}
		Let $\{\alpha_i\}_{i=1}^r$ be $r > 1$ positive integers and $m = \prod_{i=1}^{r} p_i^{\alpha_i}$ be a positive integer with $r$ different odd prime divisors: $p_1, \dots, p_r$, and $l\geq 2$ be an integer such that $l<\min(p_1,\,\ldots,\,p_r)$. Then, for all integers $t\geq 2$ and $n\geq 1$, there exists an explicitly constructible non-uniform set-system $\mathcal{H}$, defined over a universe of $h$ elements, such that
		\begin{enumerate}
			\item $h<2l(m-1)n^{4mn^\frac{1}{r}}$ if $n\geq (4m)^{1+\frac{1}{r-1}}$,
			\item $\vert\mathcal{H}\vert=n^{ln}+ln^n$,
			\item\label{T5C3} $\forall H_1, H_2\in\mathcal{H}$, either $\vert H_1\vert=\vert H_2\vert$, $\vert H_1\vert=l\vert H_2\vert$ or $l\vert H_1\vert=\vert H_2\vert$,
			\item\label{T5C4} $\mathcal{H}$ has $t$-wise restricted intersections modulo $m$.
		\end{enumerate}
	\end{theorem}
	
	\begin{proof}
		By Theorem \ref{grolmusz_k_intersections}, there exists a uniform set-system $\mathcal{G}$ that satisfies conditions 1--3 of Proposition \ref{main_construction}, and is defined over a universe of $g$ elements, such that $\vert\mathcal{G}\vert=n^n$. Furthermore, we know that $g<\frac{2(m-1)n^{2d}}{d!}$ provided the condition $n\geq 2d$ is satisfied. From Theorem \ref{BBR_theorem}, $d=\max(p_1^{e_1},\,\ldots,\,p_r^{e_r})-1$ where $e_i$ is the smallest integer that satisfies $p_i^{e_i}>\lceil n^{1/r}\rceil$, from which we obtain the following inequality:
		$$d<\max(p_1,\,\ldots,\,p_r)\lceil n^{1/r}\rceil<2m n^{1/r}.$$
		Hence if $n\geq (4m)^{1+\frac{1}{r-1}}$, then $n^\frac{r-1}{r}\geq 4m\implies n\geq 4m n^{1/r}>2d$, and thus we have:
		$$g<\frac{2(m-1)n^{2d}}{d!}<2(m-1)n^{2d}<2(m-1)n^{4mn^\frac{1}{r}}.$$
		Applying Proposition \ref{main_construction} with the set-system $\mathcal{G}$, we obtain a set-system $\mathcal{H}$ satisfying Conditions~\ref{T5C3} and~\ref{T5C4}. It follows from Remark~\ref{remarkImp}, that the size of $\mathcal{H}$ is:
		$$\vert\mathcal{H}\vert=(n^n)^l+l(n^n)=n^{ln}+ln^n,$$
		and the number of elements in the universe of $\mathcal{H}$ is $h=l g<2l(m-1)n^{4mn^\frac{1}{r}}$ for $n\geq (4m)^{1+\frac{1}{r-1}}$. $\qed$
	\end{proof}
	\begin{corollary}[Same as Theorem~\ref{mainThm}]
		Let $\{\alpha_i\}_{i=1}^r$ be $r > 1$ positive integers and $m = \prod_{i=1}^{r} p_i^{\alpha_i}$ be a positive integer with $r$ different odd prime divisors: $p_1, \dots, p_r$, and $l\geq 2$ be an integer such that $l<\min(p_1,\,\ldots,\,p_r)$. Then, there exists $c>0$ such that for all integers $t\geq 2$ and $h\geq lm$, there exists an explicitly constructible non-uniform\footnote{member	sets do not	all have equal size} set-system $\mathcal{H}$, defined over a universe of $h$ elements, such that
		\begin{enumerate}
			\item $\vert\mathcal{H}\vert>\exp\left(c\dfrac{l(\log h)^r}{(\log\log h)^{r-1}}\right)+l\exp\left(c\dfrac{(\log h)^r}{(\log\log h)^{r-1}}\right)$,
			\item $\forall H_1, H_2\in\mathcal{H}$, either $\vert H_1\vert=\vert H_2\vert$, $\vert H_1\vert=l\vert H_2\vert$ or $l\vert H_1\vert=\vert H_2\vert$,
			\item $\mathcal{H}$ has $t$-wise restricted intersections modulo $m$.
		\end{enumerate}
	\end{corollary}
	\begin{proof}
		For small values of $h$, we can simply take $\mathcal{H}$ to be the set system
		\begin{align*}
			\big\{[m-1]\cup \{m\},\ [m-1]\cup \{m+1\},\ \ldots,\ [m-1]\cup \{m+l\},\ [lm]\big\},
		\end{align*}
		so it is enough to prove the statement for sufficiently large $h$. Choose $n$ as large as possible subject to the restriction $2l(m-1)n^{4mn^\frac{1}{r}}\leq h$. We may assume that $h$ is sufficiently large so that the condition $n\geq (4m)^{1+\frac{1}{r-1}}$ is satisfied. For $N=n+1$, it holds that:
		$$h<2l(m-1)N^{4mN^\frac{1}{r}}\implies N>e^{rW_0\left(\frac{1}{4rm}\log\frac{h}{2l(m-1)}\right)},$$
		where $W_0$ is the principal branch of the Lambert $W$ function~\cite{LambertOrg[58]}. Fix any $c_1$ such that $0<c_1<\frac{1}{4rm}$. Then, for $h$ sufficiently large, $n>e^{rW_0\left(c_1\log h\right)}$. Corless et al. \cite{Corless[96]} proved the following:
		$$W_0(x)=\log x-\log\log x+o(1),$$
		hence, it follows that there exists some $c_2$ such that for all sufficiently large $h$, it holds that:
		\begin{align*}
		n&>\exp\left(r\log\log h-r\log\log\log h+c_2\right)\\
		&=\frac{e^{c_2}(\log h)^r}{(\log\log h)^r}.
		\end{align*}
		This shows that there exists $c_3>0$ such that for sufficiently large $h$, we get:
		\begin{equation}\label{EqAbove}
			n^n>\exp\left(\frac{c_3(\log h)^r}{(\log\log h)^{r-1}}\right).
		\end{equation}
		Since the size of $\mathcal{H}$ is $\vert\mathcal{H}\vert=n^{ln}+ln^n$, it follows from \Cref{EqAbove} that:
		\[\vert\mathcal{H}\vert>\exp\left(c\dfrac{l(\log h)^r}{(\log\log h)^{r-1}}\right)+l\exp\left(c\dfrac{(\log h)^r}{(\log\log h)^{r-1}}\right). \eqno \qed\] 
	\end{proof}

\begin{definition}[Covering Vectors \cite{Vipin[20]}]\label{def22}
	\emph{Let $m, h > 0$ be positive integers, $S \subseteq \mathbb{Z}_m \setminus \{0\}$, and w$(\cdot)$ and $\langle \cdot, \cdot \rangle$ denote Hamming weight and inner product, respectively. We say that a subset $\mathcal{V} = \{\textbf{v}_i\}_{i=1}^N$ of vectors in $(\mathbb{Z}_m)^h$ forms an $S$-covering family of vectors if the following two conditions are satisfied: 
		\begin{itemize}
			\item $\forall i \in [N]$, it holds that: $\langle \textbf{v}_i, \textbf{v}_i \rangle = 0 \bmod m$,
			\item $\forall i,j \in [N]$, where $i \neq j$, it holds that: 
			\begin{align*}
				\langle \textbf{v}_i, \textbf{v}_j \rangle \bmod m &= 
				\begin{cases}
					0 \qquad \qquad \quad \text{if w}(\textbf{v}_i \circ \textbf{v}_j \bmod m) = 0 \bmod m, \\
					\in S \qquad \quad \quad \text{otherwise},
				\end{cases}
			\end{align*}
		\end{itemize}
	where $\circ$ denotes Hadamard/Schur product (see Definition~\ref{Hada}).}
\end{definition}

Recall from Theorem~\ref{mainThm} that $h, m, l$ are positive integers such that $2 \leq l<\min(p_1,\,\ldots,\,p_r)$ and $m = \prod_{i=1}^r p_i^{\alpha_i}$ has $r > 1$ different prime divisors: $p_1, \ldots, p_r$. Further, it follows trivially that the sizes of the pairwise intersections of the sets in $\mathcal{H}$ occupy at most $m-1$ residue classes modulo $m$. If each set $H_i \in \mathcal{H}$ is represented by a representative vector $\textbf{v}_i \in (\mathbb{Z}_m)^h$, then for the resulting subset $\mathcal{V}$ of vectors in $(\mathbb{Z}_m)^h$, the following result follows from Theorem~\ref{mainThm}.   
\begin{corollary}[to Theorem~\ref{mainThm}]\label{corImp}
	For the set-system $\mathcal{H}$ defined in Theorem~\ref{mainThm}, if each set $H_i \in \mathcal{H}$ is represented by a unique vector \emph{$\textbf{v}_i$} $\in (\mathbb{Z}_m)^h$, then for a set $S$ of size $m-1,$ the set of vectors $\mathcal{V} =$ \emph{$\{\textbf{v}_i\}$}${}^N_{i=1}$, formed by the representative vectors of all sets in $\mathcal{H}$, forms an $S$-covering family such that $$N > \exp\left(c\dfrac{l(\log h)^r}{(\log\log h)^{r-1}}\right)+l\exp\left(c\dfrac{(\log h)^r}{(\log\log h)^{r-1}}\right)$$ and $\forall i,j \in [N]$ it holds that \emph{$\langle \textbf{v}_i, \textbf{v}_j \rangle$}$ = |H_i \cap H_j|\bmod m$.
\end{corollary}

\section{Working Over Set-Systems via Vector Families}\label{work}
In this section, we explain how vector families and special inner products can be used to work with sets from different set-systems. We begin by recalling the following two properties (from Remark~\ref{remarkImp}) that hold for all sets in any set-system $\mathcal{H}$ that is defined by Theorem~\ref{mainThm}. 
\begin{itemize}
	\item $H$ is a proper subset of exactly $s^{l-1}$ sets and not a proper superset of any sets in $\mathcal{H}$,
	\item $H$ is a proper superset of exactly $l$ sets and not a proper subset of any sets in $\mathcal{H}$,
\end{itemize}
where $s \geq \exp \left( c \dfrac{(\log h)^r}{(\log \log h)^{r-1}} \right)$.\\

Let $\mathcal{V}\subseteq(\mathbb{Z}_m)^h$ be a family of covering vectors, consisting of representative vectors for the sets in a set-system $\mathcal{H}$. For all $i \in |\mathcal{H}|(=|\mathcal{V}|)$, let $\textbf{v}_i \in \mathcal{V}$ denote the representative vector for the set $H_i \in \mathcal{H}$. Recall from Corollary~\ref{corImp} that the following holds:
\[\langle \textbf{v}_i, \textbf{v}_j \rangle = |H_i \cap H_j| \bmod m.\]
We define a $k$-multilinear form on $\mathcal{V}^k$ as:
\begin{align*}
\langle \textbf{v}_1,\,\textbf{v}_2,\,\ldots,\,\textbf{v}_k\rangle_k &=\sum_{i=1}^h \textbf{v}_{1}[i]\textbf{v}_{2}[i]\cdots \textbf{v}_{k}[i]\\ &= \Big| \bigcap\limits_{i=1}^{k} H_i \Big|.
\end{align*}
We fix a representative vector $\textbf{v}\in\mathcal{V}$ for a fixed set $H \in \mathcal{H}$. For the rest of the sets $H_i \in \mathcal{H}$, we denote their respective representative vectors by $\textbf{v}_i \in \mathcal{V}$. Let $\textbf{v},\,\textbf{v}_1,\,\textbf{v}_2\in\mathcal{V}$, and $\textbf{v}_{i \cup j} \in \mathcal{V}$ denote the representative vector for the set $H_{i \cup j} = H_i \cup H_j$. Then, the following holds:
\begin{align}
\langle \textbf{v}, \textbf{v}_{1 \cup 2}\rangle&=\vert H \cap (H_1\cup H_2)\vert=\vert(H \cap H_1)\cup (H \cap H_2)\vert \nonumber \\
&=\vert H \cap H_1 \vert+\vert H \cap H_2 \vert-\vert H \cap H_1\cap H_2\vert \nonumber \\
&=\langle \textbf{v}, \textbf{v}_1\rangle + \langle \textbf{v}, \textbf{v}_2\rangle - \langle \textbf{v}, \textbf{v}_1, \textbf{v}_2\rangle_3.
\label{Eq1}
\end{align}
Define $F$ as:
$$F(x,y,z)=x+y-z,$$
i.e., the following holds:
$$F(\langle \textbf{v}, \textbf{v}_1\rangle,\langle \textbf{v}, \textbf{v}_2\rangle,\langle \textbf{v}, \textbf{v}_1, \textbf{v}_2\rangle_3)=\langle \textbf{v}, \textbf{v}_{1\cup 2}\rangle.$$
Note that the following also holds:
\begin{align*}
\vert H \cap (H_1\cap H_2)\vert &= \langle \textbf{v}, \textbf{v}_1\rangle + \langle \textbf{v}, \textbf{v}_2\rangle - \langle \textbf{v}, \textbf{v}_{1 \cup 2}\rangle \\ 
&= \vert H \cap H_1 \vert+\vert H \cap H_2\vert - |H \cap (H_1 \cup H_2)|.
\end{align*}
Consider the following simple extension of \Cref{Eq1}:
\begin{align*}
\langle \textbf{v}, \textbf{v}_1, \textbf{v}_{2 \cup 3}\rangle_3&=\vert H \cap H_1\cap(H_2\cup H_3)\vert=\vert(H\cap H_1\cap H_2)\cup (H\cap H_1\cap H_3)\vert\\
&=\vert H \cap H_1\cap H_2\vert+\vert H \cap H_1\cap H_3 \vert-\vert H\cap H_1\cap H_2\cap H_3\vert\\
&=\langle \textbf{v}, \textbf{v}_1, \textbf{v}_2\rangle_3 + \langle \textbf{v}, \textbf{v}_1, \textbf{v}_3\rangle_3 - \langle \textbf{v}, \textbf{v}_1, \textbf{v}_2, \textbf{v}_3\rangle_4.
\end{align*}
Therefore, we get:
$$F\left(\langle \textbf{v}, \textbf{v}_1,\textbf{v}_2\rangle_3,\langle \textbf{v}, \textbf{v}_1,\textbf{v}_3\rangle_3,\langle \textbf{v}, \textbf{v}_1, \textbf{v}_2,\textbf{v}_3\rangle_4\right)=\langle \textbf{v}, \textbf{v}_1, \textbf{v}_{2\cup 3}\rangle_3.$$
Note that the following also holds:
\begin{align*}
\vert H \cap (H_1\cap H_2\:\cap&\:H_3)\vert = \langle \textbf{v}, \textbf{v}_1, \textbf{v}_2\rangle_3 + \langle \textbf{v}, \textbf{v}_1, \textbf{v}_3\rangle_3 - \langle \textbf{v}, \textbf{v}_1, \textbf{v}_{1 \cup 2}\rangle_4 \\ 
&= \vert H \cap H_1\cap H_2\vert+\vert H \cap H_1\cap H_3 \vert -\vert H \cap H_1\cap(H_2\cup H_3)\vert.
\end{align*}
It follows by extension that $\langle \textbf{v}, \textbf{v}_{1\cup 2\cup \cdots \cup w}\rangle_w,$ can be computed from the $k$-multilinear forms $\langle \textbf{v}_1,\, \textbf{v}_2,\, \ldots,\, \textbf{v}_k\rangle_k$, for all $k \in [w+1]$ and all $\textbf{v}_i\in\mathcal{V}$. Hence, $\langle \textbf{v}_i, \textbf{v}_j \rangle = |H_i \cap H_j| \bmod m$ allows us to compute intersection of any sets $H_i,H_j \in \mathcal{H}$, and being able to compute the aforementioned function $F(x,y,z)$ allows us to perform unions and intersections of any arbitrary number of sets from $\mathcal{H}$. 

\begin{figure}[t!]
	\centering
	\includegraphics[width=.7\textwidth]{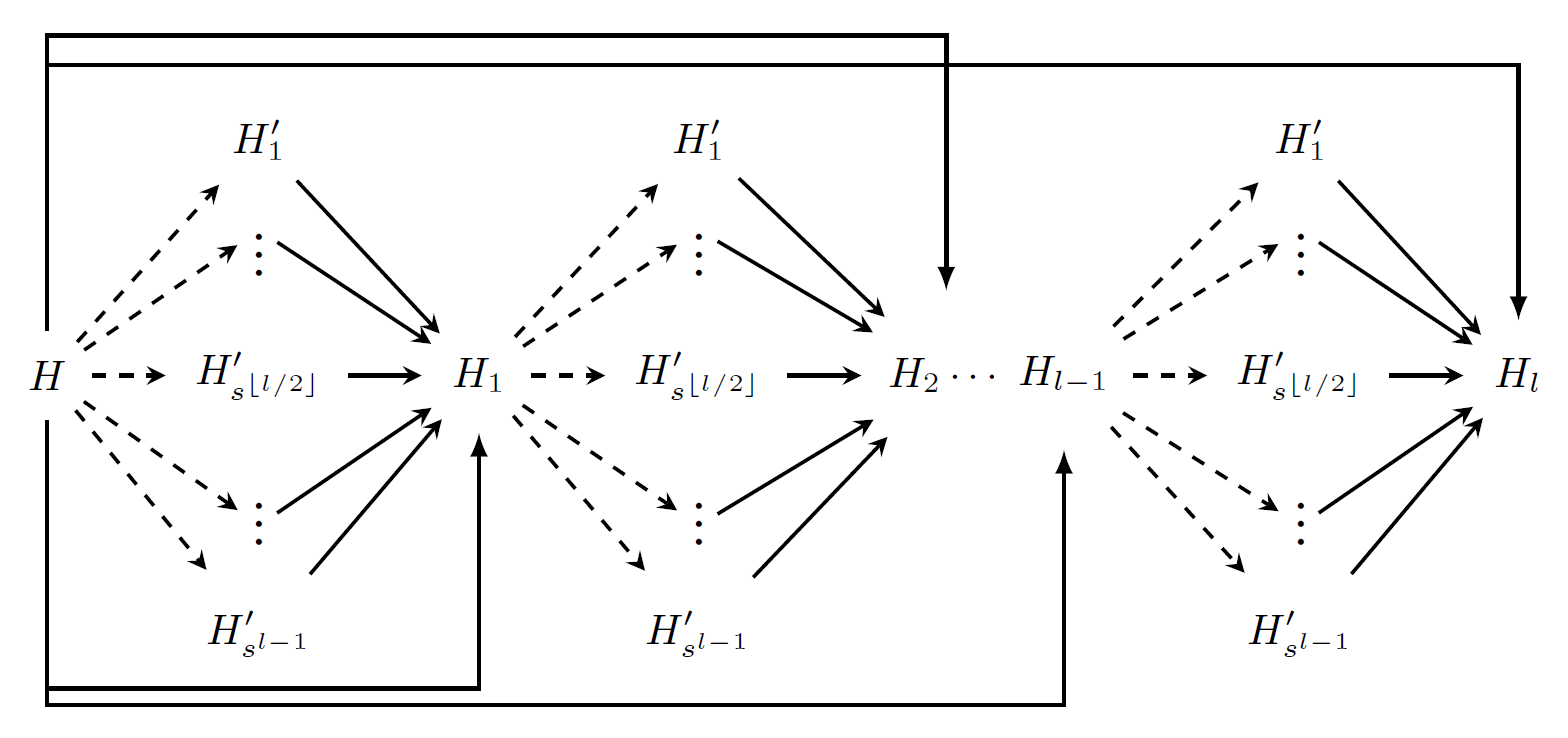}
	\caption{Supersets and subsets of a set $H \in \mathcal{H}, \mathcal{H}'$ within the two set-systems. $H_i$\sampleline{dashed,thick,->}$H_j$ denotes $H_i \subseteq H_j$, and $H_j$\sampleline{thick,->}$H_k$ denotes $H_j \supseteq H_k$. Since $\mathcal{H}$ and $\mathcal{H}'$ are defined over the identical universe of $h$ elements, superset-subset relations can hold even between the sets that exclusively belong to different set-systems.}\label{Fig1}
\end{figure}

Let $m = \prod_{i=1}^{r} p_i$ and $m' = \prod_{i=1}^{r'} p_i$ be positive integers, having $r$ and $r' > r$ different odd prime divisors, respectively. Recall from Theorem~\ref{mainThm}, that the universe of elements over which the set-system $\mathcal{H}$ is constructed is given by: $h \geq lm$, where $2 \leq l < \min(p_1, \dots, p_r)$. We construct two set-systems $\mathcal{H}$ and $\mathcal{H}'$ over $\mathbb{Z}_m$ and $\mathbb{Z}_{m'}$. Let the sets of parameters $\{h,l,m\}$ and $\{h',l',m'\}$ correspond to set-systems $\mathcal{H}$ and $\mathcal{H}'$, respectively. In order to ensure that the number of elements is same for both $\mathcal{H}$ and $\mathcal{H}'$, we set $h = h' =\max(lm,l'm')$. Since $m$ is a factor of $m'$, the following holds for all $H \in \mathcal{H}'$:
\[|H| = 0 \bmod m' = 0 \bmod m. \]
Note that for appropriate choice of the underlying set-system $\mathcal{G}$ (see Proposition~\ref{main_construction}), it holds that $|\mathcal{H} \cap \mathcal{H}'| > 0$. It follows from Remark~\ref{remarkImp} that the following two conditions hold for $H \in \mathcal{H} \cap \mathcal{H'}$:
\begin{itemize}
	\item $H$ is a proper subset of exactly $s^{l-1}$ sets and not a proper superset of any sets in $\mathcal{H}'$,
	\item $H$ is a proper superset of exactly $l$ sets and not a proper subset of any sets in $\mathcal{H}$.
\end{itemize}

Therefore, it holds for the representative $\textbf{v}$ of $H$ that: $\textbf{v} \in \mathcal{V} \cap \mathcal{V}'$. \Cref{Fig1} gives graphical depiction of the various subset and superset relationships of $H$ in $\mathcal{H}$ and $\mathcal{H}'$. Specifically, it shows the $s^{l-1}$ proper supersets $\{H'_1, \ldots, H'_{s^{l-1}} \} \in \mathcal{H}'$ of $H$, along with its $l$ proper subsets $\{H_1, \dots, H_l\} \in \mathcal{H}$.  

\begin{figure}[t!]
	\centering
	\includegraphics[width=.7\textwidth]{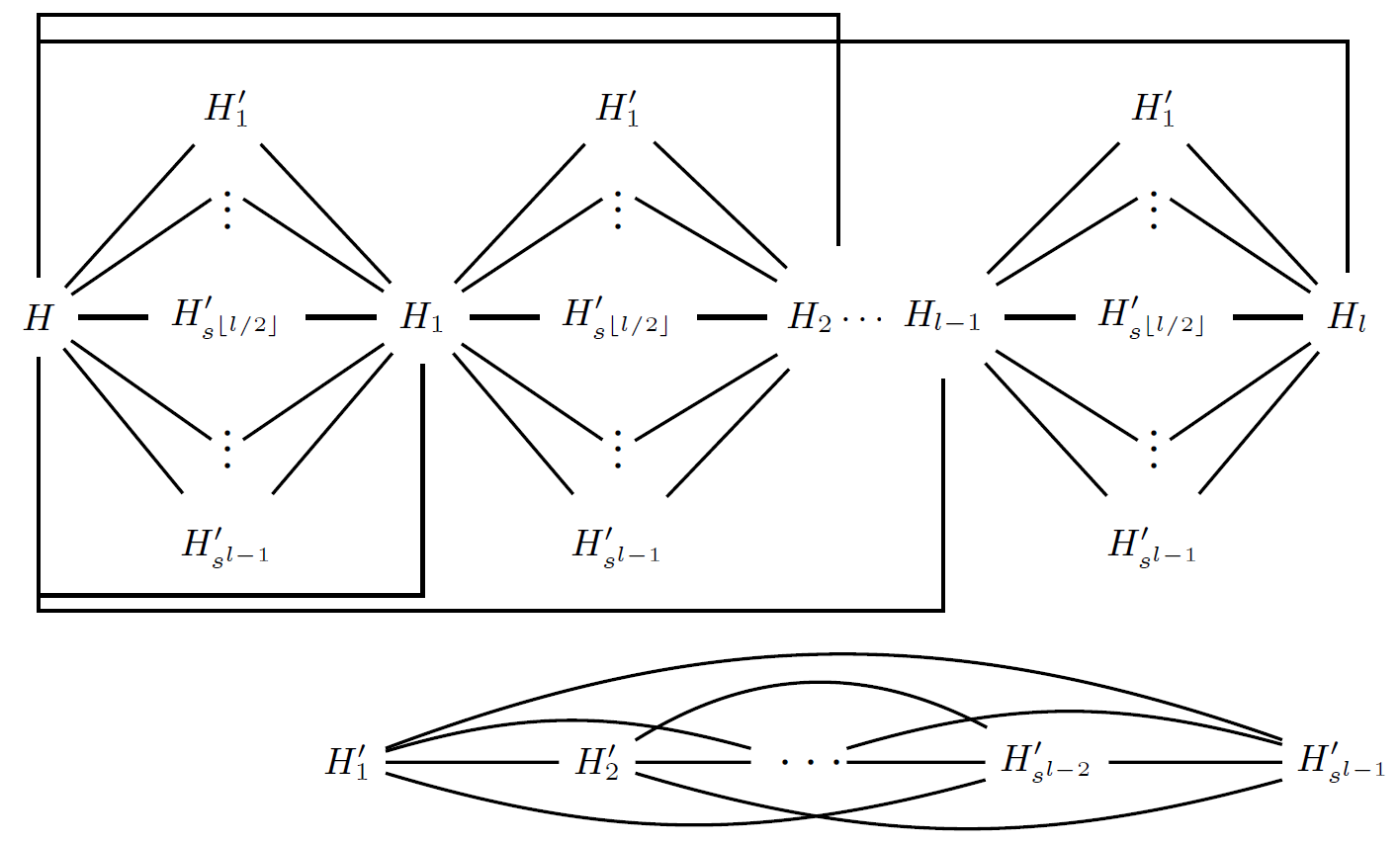}
	\caption{Hopping between supersets and subsets of a set $H \in \mathcal{H} \cap \mathcal{H}'$.}\label{Fig2}
\end{figure}

Let $\mathcal{V} \in \mathbb{Z}_m^h$ and $\mathcal{V}' \in \mathbb{Z}_{m'}^h$ be the covering vectors families (see Definition~\ref{def22}) that correspond to the set systems $\mathcal{H}$ and $\mathcal{H}'$, respectively. It is easy to see that for all $\textbf{v} \in \mathcal{V}$ and $\textbf{v}' \in \mathcal{V}'$, there exists some vector $\textbf{v}_\delta \in \mathbb{Z}^h$ such that $\textbf{v} + \textbf{v}_\delta = \textbf{v}'$. Since vector inner products are additive in the second argument, we can compute: $\langle \textbf{u}, \textbf{v} \rangle + \langle \textbf{u}, \textbf{v}_\delta \rangle = \langle \textbf{u}, \textbf{v}' \rangle$, and hence ``hop'' between the set-systems $\mathcal{H}$ and $\mathcal{H}'$. Having this ability along with our $k$-multilinear forms, allows us to ``hop'' within and between the set-systems $\mathcal{H}$ and $\mathcal{H}'$ via inner products of the corresponding vectors from covering vectors families $\mathcal{V}$ and $\mathcal{V}'$. 

For instance, given a set $H \in \mathcal{H} \cap \mathcal{H}'$, we can ``hop'' between the subsets and supersets of $H$ within the two set-systems. \Cref{Fig2} shows all such hops between all supersets and subsets of $H \in \mathcal{H} \cap \mathcal{H}'$. 

\section{Access Structure Encoding}\label{sec5}
In this section, we give an example procedure to encode any access structure. Let $\mathcal{P} = \{P_1, \ldots, P_\ell\}$ be a set of $\ell$ polynomial-time parties and $\mathrm{\Omega} \in \Gamma_0$ be any minimal authorized subset (see Definition~\ref{GammaDef}). Hence, each party $P_i \in \mathcal{P}$ can be identified as $P_i \in \mathrm{\Omega}$ or $P_i \in \mathcal{P} \setminus \mathrm{\Omega}$. We begin by giving an overview of the central idea of our scheme.

\subsection{\textbf{Central Idea}} 
If all parties in the minimal authorized subset $\mathrm{\Omega} \subseteq \mathcal{P}$ combine their access structure tokens $\{\mathrm{\mho}^{(\Gamma)}_i\}_{i \in \mathrm{\Omega}}$, they should arrive at a fixed set $H$, which represents $\mathrm{\Omega}$. From thereon, access structure token of each party $P_j \in \mathcal{P} \setminus \mathrm{\Omega}$ is generated such that no combination of their access structure tokens can reach $H$. Finally, the result of combining the access structure tokens, $\{\mathrm{\mho}^{(\Gamma)}_i\}_{i \in \mathcal{A}}$, of any authorized subset $\mathcal{A} \in \Gamma$, where $\Gamma = $ cl$(\mathrm{\Omega})$, takes us to some set $H_\phi \supseteq H$. As described in Section~\ref{work}, we can operate on the sets in our set-systems via their respective representative vectors, their inner products and $k$-multilinear forms. 

\subsection{\textbf{Example Procedure}}
Generate an integer $m = \prod_{i=1}^r p_r$ with $r>1$ prime divisors: $p_1, \ldots, p_r$ such that $|\mathrm{\Omega}| \ll \max(p_1, \ldots, p_r)$. Define a set-system $\mathcal{H}$ modulo $m$ (as described in Theorem~\ref{mainThm}) such that $l+|\mathrm{\Omega}| \ll \max(p_1, \ldots, p_r)$. Pick a set $H \xleftarrow{\; \$ \;} \mathcal{H}$ such that $H$ is a proper subset of exactly $s^{l-1}~(>\ell)$ sets and not a proper superset of any sets in $\mathcal{H}$. Let $\Im \subset \mathcal{H}$ denote the collection of sets in $\mathcal{H}$ that are supersets of $H$. Randomly generate a positive integer $\kappa$ such that $l+|\mathrm{\Omega}|+\kappa<\max(p_1, \ldots, p_r)$. Then, the following procedure is used to assign unique sets from $\mathcal{H}$ to the parties in $\mathcal{P}$. Without loss of generality, we assume that $\mathrm{\Omega} = \{P_1, P_2, \ldots, P_{|\mathrm{\Omega}|}\}.$

\begin{enumerate}
	\item The set for party $P_1$ is generated as: $$S_1 = H \sqcup ([|\mathrm{\Omega}|+\kappa] \setminus \{1\}).$$
	\item For each party $P_i ~(2 \leq i \leq |\mathrm{\Omega}|-1)$, generate its set as:
	\[S_i = H_i \sqcup ([|\mathrm{\Omega}|+\kappa] \setminus \{i\}),\] 
	where $H_i \xleftarrow{\; \$ \;} \Im$ is a superset of $H$.
	\item The set for party $P_{|\mathrm{\Omega}|}$ is generated as:
	\[S_{|\mathrm{\Omega}|} = H_{|\mathrm{\Omega}|} \sqcup ([|\mathrm{\Omega}| + \kappa] \setminus \{|\mathrm{\Omega}|, \ldots, |\mathrm{\Omega}|+\kappa \}) = H_{|\mathrm{\Omega}|} \sqcup [|\mathrm{\Omega}| - 1],\]
	where $H_{|\mathrm{\Omega}|} \xleftarrow{\; \$ \;} \Im$ is a superset of $H$.
	\item For each party $P_j \in \mathcal{P} \setminus \mathrm{\Omega}$, generate its set as:
	\[S_j = H_j \sqcup [|\mathrm{\Omega}|+\kappa], \]
	where $H_j \xleftarrow{\; \$ \;} \Im$ is a superset of $H$.
	\item Generate a ``special'' set $H_0 = H_\partial \sqcup [|\mathrm{\Omega}|+\kappa],$ where $H_\partial \xleftarrow{\; \$ \;} \Im$ and $H_\partial \neq H_i$ for all $i \in [|\mathrm{\Omega}|]$.
	\item For each $i \in [\ell]$, compute $H_0 \cap S_i$, and let $\textbf{s}_i$ denote the elements in $H_0 \cap S_i$. Let $\gamma$ a random permutation, then the access structure token for party $P_i \in \mathcal{P}$ is $\gamma(\textbf{s}_i)$.
\end{enumerate}   
Generating access structure tokens in this manner allows any subset of parties $\mathcal{A} \in \mathcal{P}$ to compute intersections of their respective sets $\{S_i\}_{i \in \mathcal{A}}$ by simply computing the inner products of $\{\textbf{s}\}_{i \in \mathcal{A}}$ modulo $m$. The way in which the sets $S_i$ are generated ensures that: $$\bigcap_{i \in \mathcal{A}} S_i = H$$ if and only if $\mathrm{\Omega} \subseteq \mathcal{A}$, i.e., $\mathcal{A} \in \Gamma$. By choosing a large enough maximum prime factor $\max(p_i)$, we can guarantee that the size of the intersection is never a multiple of $m$ unless $\bigcap\limits_{i \in \mathcal{A}} S_i = H$.

\begin{note}[From example procedure to a general procedure]\label{NoteImp}
	The space overhead of sending one unique vector to each party is $\mathrm{\Theta}(h)+\max(p_i) = \mathrm{\Theta}(h)$, where $h$ is the number of elements in the universe over which $\mathcal{H}$ is defined. We know from Section~\ref{work} that instead of vectors, the parties can be provided with inner products along with sizes of various unions and intersections. This allows the parties to compute the sizes of the intersections of their respective sets without revealing any information about the sets themselves. However, in order to perform
	unions (and the respective intersections) of $\ell$ sets, the parties need sizes of the intersections and unions of various combinations of sets $S_i~(1 \leq i \leq \ell)$, which increases the space overhead to $\approx 2^\ell$.
\end{note}

From hereon, we use $\mathrm{\mho}^{(\Gamma)}_i$ to denote the access structure token of party $P_i~(i \in [\ell])$. 

\begin{lemma}\label{lemma1}
	For every authorized subset of parties $\mathcal{A} \in \Gamma$, it holds that $|\bigcap\limits_{i \in \mathcal{A}} S_i| = 0 \bmod m$.
\end{lemma}
\begin{proof}
	It follows from the generation of $\{S_i\}_{i=1}^\ell$ that $\bigcap\limits_{i \in \mathcal{A}} S_i = H$. Hence, it follows that $|\bigcap\limits_{i \in \mathcal{A}} S_i| = 0 \bmod m. \qed$
\end{proof}

\begin{lemma}\label{lemma2}
	For every unauthorized subset of parties $\mathcal{B} \notin \Gamma$, it holds that $|\bigcap\limits_{i \in \mathcal{B}} S_i| \neq 0 \bmod m$.
\end{lemma}
\begin{proof}
	Since $\mathcal{B}$ is unauthorized, there exists $1\leq j\leq |\mathrm{\Omega}|$ such that $P_j\not\in\mathcal{B}$. Then $\bigcap\limits_{i \in \mathcal{B}} S_i=K\sqcup K'$, where $K$ is an intersection of certain supersets of $H$ (which might include $H$ itself), and $K'$ is a non-empty subset of $[|\mathrm{\Omega}|+\kappa]$. It follows that:
	$$1\leq\left|\bigcap\limits_{i \in \mathcal{B}} S_i\right|\bmod p\leq l+|\mathrm{\Omega}|+\kappa$$
	for $p=\max(p_1, \ldots, p_r)$, from which we obtain $|\bigcap\limits_{i \in \mathcal{B}} S_i| \neq 0 \bmod m$. $\qed$
\end{proof}

\section{\textsf{PRIM-LWE}}\label{prim} 
In this section, we present a new variant of LWE, called \textsf{PRIM-LWE}. We begin by describing and deriving some relevant results. 

Recall that $M_n(\mathbb{Z}_p)$ denotes the space of $n \times n$ matrices over $\mathbb{Z}_p$.
\begin{proposition}
	\label{determinants_of_matrices}
	Let $p$ be prime. Then, there are
	$$p^{n^2}-\prod_{k=0}^{n-1}(p^n-p^k)$$
	matrices in $M_n(\mathbb{Z}_p)$ with determinant $0$. And, for any non-zero $\alpha\in\mathbb{Z}_p$, there are
	$$p^{n-1}\prod_{k=0}^{n-2}(p^n-p^k)$$
	matrices with determinant $\alpha$.
\end{proposition}

\begin{proof}
	Clearly, there are $p^{n^2}$ matrices in $M_n(\mathbb{Z}_p)$. A matrix $\textbf{M}\in M_n(\mathbb{Z}_p)$ has non-zero determinant if and only if it has linearly independent columns. Hence there are
	$$\gamma(p):=(p^n-1)(p^n-p)(p^n-p^2)\cdots(p^n-p^{n-1})=\prod_{k=0}^{n-1}(p^n-p^k)$$
	such matrices, which shows that there are
	$$p^{n^2}-\gamma(p)=p^{n^2}-\prod_{k=0}^{n-1}(p^n-p^k)$$
	matrices with determinant $0$.
	
	To prove the second statement, consider the determinant map
	$$\det: GL_n(\mathbb{Z}_p)\to\mathbb{Z}_p^\ast,$$
	where $GL_n(\cdot)$ denotes the general linear group (see \cite{ConCurt[85]}, Chapter 2.1, p. x) and $\det$ is a group homomorphism. Hence, it follows that for any $\alpha\in\mathbb{Z}_p^\ast$, there are exactly
	$$\frac{\gamma(p)}{p-1}=p^{n-1}\prod_{k=0}^{n-2}(p^n-p^k)$$
	matrices with determinant $\alpha$. $\qed$
\end{proof}

\begin{corollary}
	\label{primitive_root_determinants}
	The fraction of matrices in $M_n(\mathbb{Z}_p)$ whose determinant generates $\mathbb{Z}_p^\ast$ is:
	$$\frac{\varphi(p-1)\prod_{k=2}^{n}(p^k-1)}{p^{\frac{1}{2}n(n+1)}},$$
	where $\varphi$ is Euler's totient function (see Theorem~\ref{Euler}).
\end{corollary}
\begin{proof}
	Note that $\mathbb{Z}_p^\ast$ is cyclic of order $p-1$, so it has exactly $\varphi(p-1)$ different generators. By Proposition \ref{determinants_of_matrices}, the required fraction is
	$$\varphi(p-1)\times\frac{p^{n-1}\prod_{k=0}^{n-2}(p^n-p^k)}{p^{n^2}}=\varphi(p-1)\times\frac{\prod_{k=2}^{n}(p^k-1)}{p^{\frac{1}{2}n(n+1)}}.\eqno \qed$$
\end{proof}

Let us recall two standard results, given as Propositions~\ref{infinite_products_1} and~\ref{infinite_products_2}, from the theory of infinite products (see~\cite{Knopp[56]} for more details).

\begin{proposition}
	\label{infinite_products_1}
	The infinite product $\prod_{k=1}^{\infty}a_k$ converges to a non-zero limit if and only if $\sum_{k=1}^{\infty}\log a_n$ converges.
\end{proposition}

\begin{proposition}
	\label{infinite_products_2}
	$\sum_{k=1}^{\infty}\log a_n$ converges absolutely if and only if $\sum_{k=1}^{\infty}(1-a_n)$ converges absolutely. Hence, if $\sum_{k=1}^{\infty}(1-a_n)$ converges absolutely, then $\prod_{k=1}^{\infty}a_k$ converges to a non-zero limit.
\end{proposition}

\begin{proposition}
	\label{primitive_root_determinants_lower_bound}
	There exists a constant $c=c(p)>0$, independent of $n$, such that the fraction of matrices in $M_n(\mathbb{Z}_p)$ whose determinant generates $\mathbb{Z}_p^\ast$ is bounded below by $c$ for all $n$.
\end{proposition}
\begin{proof}
	Let
	$$f_p(n)=\frac{\prod_{k=2}^n(p^k-1)}{p^{\frac{1}{2}n(n+1)}}=\frac{1}{p-1}\prod_{k=1}^{n}\frac{p^k-1}{p^k}.$$
	Since the infinite series
	$$\sum_{k=1}^\infty\left(1-\frac{p^k-1}{p^k}\right)=\sum_{k=1}^\infty\frac{1}{p^k}=\frac{1}{p-1}$$
	converges absolutely, by Proposition \ref{infinite_products_2}, $\lim_{n\rightarrow\infty} f_p(n)$ exists and is non-zero. Let $c^\prime=\lim_{n\rightarrow\infty} f_p(n)>0$.
	
	Note, furthermore, that $f_p(n+1)<f_p(n)$, so that $f_p(n)>c^\prime$ for all $n$. By Proposition \ref{primitive_root_determinants}, at least a $c^\prime\varphi(p-1)>0$ fraction of the matrices in $M_n(\mathbb{Z}_p)$ have determinants which are primitive roots of unity modulo $p$. $\qed$
\end{proof}

\begin{table}
	\centering
	\setlength{\tabcolsep}{12pt}
	\begin{tabular}{c|cccccccc}
		$p$ & 2 & 3 & 5 & 7 & 11 & 13 & 17 & 19 \\
		\hline
		$\varphi(p-1)$ & 1 & 1 & 2 & 2 & 4 & 4 & 8 & 6 \\
		$c(p)$ & 0.289 & 0.280 & 0.380 & 0.279 & 0.360 & 0.306 & 0.469 & 0.315
	\end{tabular}
	\vspace{5mm}
	\caption{Approximate values of $c(p)$}
	\label{primitive_root_determinants_table}
\end{table}

\begin{remark}
	\begin{enumerate}[label=(\roman*)]
		\item While the exact values of $c(p)=\lim_{n\rightarrow\infty}f_p(n)\varphi(p-1)$ appear difficult to determine, we have calculated some approximate values of $c(p)$ as shown in Table \ref{primitive_root_determinants_table}.
		\item It might seem from Table \ref{primitive_root_determinants_table} that $c(p)$ does not vary much with $p$. Nevertheless, it might in fact be the case that $\inf_{p\text{ prime}}c(p)=0$. A primorial prime is a prime of the form $p=\prod_{i=1}^k p_i +1$, where $p_1<p_2<\cdots<p_k$ are the first $k$ primes. For such a prime $p=\prod_{i=1}^k p_i +1$,
		$$c(p)=\frac{\varphi(p-1)}{p-1}\prod_{j=1}^{\infty}\frac{p^j-1}{p^j}<\frac{\varphi(p-1)}{p-1}=\prod_{j=1}^{k}\frac{p_j-1}{p_j}.$$
		It is an open problem whether or not there are infinitely many primorial primes, but heuristic arguments suggest that this should be the case. Suppose there actually are infinitely many such primes. Then, since
		$$\sum_{j=1}^{\infty}\left(1-\frac{p_j-1}{p_j}\right)=\sum_{j=1}^{\infty}\frac{1}{p_j}$$
		diverges \cite{Eynden80}, $c(p)$ diverges to $0$ as $k\rightarrow\infty$, which shows that $\inf_{p}c(p)=0$.
	\end{enumerate}
\end{remark}

We are now ready to define \textsf{PRIM-LWE}. First, we define:
$$M_n^{prim}(\mathbb{Z}_p)=\{\textbf{M}\in M_n(\mathbb{Z}_p): \det(\textbf{M})\text{ is a generator of }\mathbb{Z}_p^\ast\}.$$
Recall that for a vector $\textbf{s}\in \mathbb{Z}_p^n$ and a noise distribution $\chi$ over $\mathbb{Z}_p$, \textsf{LWE} distribution $D^{\textsf{LWE}}_{n,p,\textbf{s},\chi}$ is defined as the distribution over $\mathbb{Z}_p^n\times\mathbb{Z}_p$ that is obtained by choosing $\textbf{a} \xleftarrow{\; \$ \;} \mathbb{Z}_p^n, e\xleftarrow{\; \$ \;} \chi$, and outputting $(\textbf{a},\,\langle \textbf{a},\textbf{s}\rangle+e)$.

\begin{definition}[\textsf{PRIM-LWE}]
	\emph{Let $n\geq 1$ and $p\geq 2$. Then, for a matrix $S\in M_n^{prim}(\mathbb{Z}_p)$ and a noise distribution $\chi$ over $M_n(\mathbb{Z}_p)$, define the \textsf{PRIM-LWE} distribution $D^{\textsf{PRIM-LWE}}_{n,p,\textbf{S},\chi}$ to be the distribution over $M_n(\mathbb{Z}_p)\times M_n(\mathbb{Z}_p)$ obtained by choosing a matrix $\textbf{A} \xleftarrow{\; \$ \;} M_n(\mathbb{Z}_p)$ uniformly at random, $\textbf{E} \xleftarrow{\; \$ \;} \chi$, and outputting $(\textbf{A},\,\textbf{AS}+\textbf{E})$.}
\end{definition}

For distributions $\psi$ over $\mathbb{Z}_p^n$ and $\chi$ over $\mathbb{Z}_p$, the decision-\textsf{LWE}$_{n,p,\psi,\chi}$ problem is to distinguish between $(a,\,b)\leftarrow D^{\textsf{LWE}}_{n,p,\textbf{s},\chi}$ and a sample drawn uniformly from $\mathbb{Z}_p^n\times\mathbb{Z}_p$, where $\textbf{s}\leftarrow\psi$. Similarly, for distributions $\psi$ over $M_n^{prim}(\mathbb{Z}_p)$ and $\chi$ over $M_n(\mathbb{Z}_p)$, the decision-\textsf{PRIM-LWE}$_{n,p,\psi,\chi}$ problem is to distinguish between $(\textbf{A},\,\textbf{B})\leftarrow D^{\textsf{PRIM-LWE}}_{n,p,\textbf{S},\chi}$ and a sample drawn uniformly from $M_n(\mathbb{Z}_p)\times M_n(\mathbb{Z}_p)$, where $\textbf{S}\leftarrow\psi$.

\begin{theorem}
	Let $\psi$ and $\psi^\prime$ be the uniform distributions over $M_n^{prim}(\mathbb{Z}_p)$ and $\mathbb{Z}_p^n$ respectively. Suppose $\chi$ is the distribution over $M_n(\mathbb{Z}_p)$ obtained by selecting each entry of the matrix independently from the distribution $\chi^\prime$ over $\mathbb{Z}_p$. Then, solving decision-\textsf{PRIM-LWE}$_{n,p,\psi,\chi}$ is at least as hard as decision-\textsf{LWE}$_{n,p,\psi^\prime,\chi^\prime}$, up to an $O(n^2)$ factor.
\end{theorem}

\begin{proof}
	Let $\varepsilon$ be the advantage of an adversary in solving decision-\textsf{LWE}$_{n,p,\psi^\prime,\chi^\prime}$. By a standard hybrid argument, the advantage of distinguishing $(\textbf{A}, \textbf{AS}+\textbf{E})$ from a sample uniformly drawn from $M_n(\mathbb{Z}_p)\times\mathbb{Z}_p^n$ is at most $n\varepsilon$.
	
	A sample $(\textbf{A},\,\textbf{AS}+\textbf{E})$ where $\textbf{A}, \textbf{S} \xleftarrow{\; \$ \;} M_n(\mathbb{Z}_p)$ is the same as $n$ samples $(\textbf{A},\,\textbf{A}\textbf{s}_i+\textbf{E}_i)$, with $n$ different secrets $\textbf{s}_i~(i \in [n])$. Hence, the advantage of an adversary in distinguishing $(\textbf{A},\,\textbf{AS}+\textbf{E})$ from uniformly random is at most $n^2\varepsilon$.
	
	By Proposition \ref{primitive_root_determinants_lower_bound}: $c=\inf_{n\geq 1}\vert M_n^{prim}(\mathbb{Z}_p)\vert/\vert M_n(\mathbb{Z}_p)\vert>0$. Given $m=\lceil 1/c\rceil$ samples $(\textbf{A}_i,\,\textbf{A}_i\textbf{S}_i+\textbf{E}_i)$, where $\textbf{S}_i\xleftarrow{\; \$ \;} M_n(\mathbb{Z}_p)$,
	$$\Pr[\textbf{S}_i \in M_n^{prim}(\mathbb{Z}_p)\text{ for some }i]\geq 1-(1-c)^m\geq 1-e^{-cm}\geq 1-e^{-1}.$$
	Therefore, if $\textbf{S} \xleftarrow{\; \$ \;} M_n^{prim}(\mathbb{Z}_p)$, then the advantage of an adversary in distinguishing $(\textbf{A},\,\textbf{AS}+\textbf{E})$ from uniformly random $(\textbf{A}, \textbf{B})$ is at most:
	$$\frac{mn^2}{1-e^{-1}}\varepsilon=O(n^2)\varepsilon. \eqno \qed$$
\end{proof}

\section{Graph-Based Access Structure Hiding Verifiable Secret Sharing}\label{sec6}
In this section, we present the first graph-based access structure hiding verifiable secret sharing scheme, which is also the first LWE-based secret sharing scheme for general access structures. 

\begin{note}
The (loose) description that follows gives a high-level overview of the actual scheme, and its only purpose is to facilitate better understanding of the full scheme. For the sake of simplicity, we assume that the access structure tokens are generated as inner products and unions (as described in Note~\ref{NoteImp}).
\end{note}

Unlike the regular definition of STCON in the context of secret sharing, wherein parties are represented by edges in the graph, we denote parties by nodes in the graph.

\subsection{\textbf{High-Level Overview}}\label{high-level}
Based on a minimal authorized set $\mathrm{\Omega} \in \Gamma_0$, the dealer generates a connected DAG, $G = (V,E)$, where $|V| = \ell$, such that $G$ contains a source node $s$ and a sink node $t$. Each vertex/node in $G$ houses exactly one party with the STCON housing the parties in $\mathrm{\Omega}$, i.e., the number of nodes in the STCON is $|\mathrm{\Omega}|$. Figure~\ref{Fig3} gives an example graph wherein node $v_2$ denotes $s$ and node $v_6$ represents $t$, with $\mathrm{\Omega} = \{P_1, P_2, P_3, P_4\}$.

\begin{figure}
	\centering
	\begin{tikzpicture}
	
	\begin{scope}[every node/.style={circle,thick,draw}]
	\node[label={left:$t, P_6$}] (F) at (4,2.5) {$v_6$};
	\node[label={above:$s,P_1$}] (A) at (6,4) {$v_2$};
	\node[label={below:\small{$P_3$}}] (C) at (8,0.5) {$v_4$};
	\node[label={above:\small{$P_2$}}] (B) at (8.7,3) {$v_3$} ;
	\node[label={below:\small{$P_4$}}] (E) at (5,0.5) {$v_5$} ;
	\node[label={above:\small{$P_5$}}] (G) at (2,4) {$v_1$};
	\node[label={below:\small{$P_7$}}] (H) at (.8,1.2) {$v_7$};
	\end{scope}
	
	\begin{scope}[>={Stealth[black]},
	every node/.style={fill=white,circle}]
	\path [->,very thick] (A) edge node {$\textbf{D}_2$} (B);
	\path [->,very thick] (B) edge node {$\textbf{D}_3$} (C);
	\path [->,very thick, dashed, bend left] (A) edge node {$\textbf{D}$} (F);
	\path [->,very thick] (C) edge node {$\textbf{D}_4$} (E);
	\path [->,very thick] (E) edge node {$\textbf{D}_5$} (F);
	\end{scope}
	
	\begin{scope}[>={Stealth[black]},
	every node/.style={fill=white,circle}]
	\path [->,very thick] (G) edge node {$\textbf{D}_1$} (A);
	\path [->,very thick] (F) edge node {$\textbf{D}_6$} (H);
	\end{scope}
	
	\end{tikzpicture}
	\caption{Example DAG with STCON representing $\mathrm{\Omega} = \{P_1, P_2, P_3, P_4\}$.}\label{Fig3}
\end{figure}
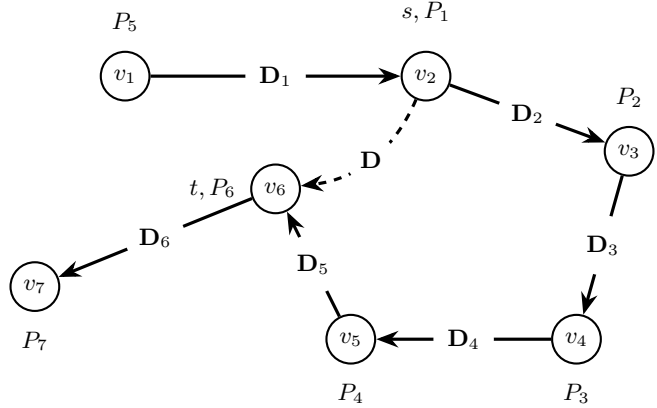

A unique matrix $\textbf{A}_v$ along with the corresponding `trapdoor information' $\tau_i$ is associated with each node $v \in V$ (i.e., party $P_v \in \mathcal{P})$, and ``encodings'' in the scheme are defined relative to the directed paths in $G$. Let $k$ be the secret to be shares. A ``small'' matrix $\textbf{S}$, such that $\det(\mathbf{S})=k$, is encoded with respect to a path $u \rightsquigarrow v$ via another ``small'' matrix $\textbf{D}_u$ such that $\textbf{D}_u \cdot \textbf{A}_u \approx \textbf{A}_v \cdot \textbf{S}_u$, where $\textbf{S}_u = \textbf{S}^{\mathrm{\mho}_u^{(\Gamma)}}$. Access structure token for party $P_u$ with respect to access structure $\Gamma$ is denoted by $\mathrm{\mho}^{(\Gamma)}_u$, and generated by following the procedures given in Section~\ref{sec5} and Note~\ref{NoteImp}. For the sake of simplicity, we assume that the access structure tokens are generated as inner products and unions (as described in Note~\ref{NoteImp}). For one randomly selected party $P_j \in \mathrm{\Omega}$, the share is generated as: $\textbf{S}_j = \textbf{S}^{\mathrm{\mho}_j^{(\Gamma)}+1}$. Given `trapdoor information' $\tau_u$ for $\textbf{A}_u$, encoding $\textbf{D}_u$ for share $\textbf{S}_u$ with respect to sink $v$ is generated such that: $$\textbf{D}_u \cdot \textbf{A}_u = \textbf{A}_v \cdot \textbf{S}_u + \textbf{E}_u,$$ where $\textbf{E}_u$ is a small LWE error matrix. It is easy to see that the LWE instance $\{\textbf{A}_u, \textbf{B}_u (= \textbf{A}_v \cdot \textbf{S}_u + \textbf{E}_u)\}$ remains hard for appropriate parameters and dimensions. Encodings relative to paths $v \rightsquigarrow w$ and $u \rightsquigarrow v$ can be multiplied to get an encoding relative to path $u \rightsquigarrow w$. Namely, given: 
$$\textbf{D}_v \cdot \textbf{A}_v = \textbf{A}_w \cdot \textbf{S}_v + \textbf{E}_v \quad  \text{ and } \quad \textbf{D}_u \cdot \textbf{A}_u = \textbf{A}_v \cdot \textbf{S}_u + \textbf{E}_u,$$ 
we obtain: 
\[\textbf{D}_v \cdot \textbf{D}_u \cdot \textbf{A}_u = \textbf{D}_v \cdot (\textbf{A}_v \cdot \textbf{S}_u+ \textbf{E}_u) = \textbf{A}_w \cdot \textbf{S}_v \cdot \textbf{S}_u+ \textbf{E}^\prime,\] such that the matrices, $\textbf{D}_v \cdot \textbf{D}_u$, $\textbf{S}_v \cdot \textbf{S}_u$ and $\textbf{E}^\prime$ remain small. Note that our procedures for generating and multiplying the encodings are similar to that of Gentry et al.~\cite{Gentry[15]}, but apart from that two schemes are completely different; their aim was to develop a multilinear map scheme and their tools do not include extremal set theory.

The final encoding for any set of parties is generated by combining their respective encodings according to the source-sink order of the nodes that house them. For example, in \Cref{Fig3}, the final encoding for the path from $s$ to $t$ is given by $\textbf{D} =\textbf{D}_2\textbf{D}_3\textbf{D}_4\textbf{D}_5$. Specifically, if $\textbf{A}_1$ is the matrix of party $P_1$, housed by $s$, and $\textbf{A}_6$ is the matrix assigned to party $P_6$, housed by $t$, then we can compute: 
\begin{align*}
	\textbf{D} \cdot \textbf{A}_1 \bmod N &= \textbf{A}_6 \cdot \textbf{S}^{1+\sum_{i\in \mathrm{\Omega}} \mathrm{\mho}_i^{(\Gamma)}} + \textbf{E}^\prime \bmod N\\ &= \textbf{A}_6 \cdot \textbf{S}^{1+cm} + \textbf{E}' \bmod N \qquad \text{ (using Lemma~\ref{lemma1})} \\ &= \textbf{A}_6 \cdot \textbf{S}^{1+c \varphi(N)} +\textbf{E}' \bmod N,
\end{align*}
where $N$ is some composite integer, $\varphi$ denotes Euler's totient function (see Theorem~\ref{Euler}) and $m = \varphi(N)$ is the modulo with respect to which the set-system (as described in Theorem~\ref{mainThm}) is defined. The scheme ensures that any authorized subset of parties $\mathcal{A} \supseteq \mathrm{\Omega}$ possesses the `trapdoor information' required to invert $\textbf{A}_6 \cdot \textbf{S}^{1+c \varphi(N)} + \textbf{E}' \bmod N$, and recover the secret $k=\det(\textbf{S})$. 

\subsection{\textbf{Detailed Scheme: Share Generation}} 
Let $m=\prod_{i=1}^r p_i~(\min(p_i) = 3)$ be a positive integer with $r>1$ odd prime divisors such that $2m+1$ is prime. Recall from Dirichlet's Theorem (see Theorem~\ref{Dr}) that there are infinitely many odd integers $m$ such that $2m+1$ is prime. As described by Theorem~\ref{mainThm}, define a set-system $\mathcal{H}$ modulo $m$ over a universe of $h$ elements such that for all $H_1, H_2 \in \mathcal{H}$, it holds that exactly one of the following three conditions is true
\begin{itemize}
	\item $|H_1| = |H_2| = \eta m$, where $\eta$ is some even integer,
	\item $|H_1| = l|H_2|$,
	\item $|H_2| = l|H_1|$,
\end{itemize}
where $l = 2$.

Let $m'=mp_{r'}$ be a positive integer, where $p_{r'}$ is an odd prime such that for all $i \in [r]$, it holds that: $p_{r'} \neq p_i$. According to Theorem~\ref{mainThm}, define a set-system $\mathcal{H}'$ modulo $m'$ over a universe of $h$ elements. Since $m$ is a factor of $m'$, the following holds for all $H \in \mathcal{H}'$:
\[|H| = 0 \bmod m' = 0 \bmod m. \]
Note that for appropriate choice of the underlying set-system $\mathcal{G}$ (see Proposition~\ref{main_construction}), it holds that $|\mathcal{H} \cap \mathcal{H}'| > 0$. Hence, we pick a set $H \in \mathcal{H} \cap \mathcal{H}'$ to generate access structure tokens. We know that the following holds for some $H \in \mathcal{H} \cap \mathcal{H}'$:
\begin{itemize}
	\item $H$ is a proper subset of exactly $s^{l-1}$ sets and not a proper superset of any sets in $\mathcal{H}'$,
	\item $H$ is a proper superset of exactly $l$ sets and not a proper subset of any sets in $\mathcal{H}$,
\end{itemize}
where $s \geq \exp \left( c \dfrac{(\log h)^r}{(\log \log h)^{r-1}} \right)$.

\begin{note}[Encoding $l+1$ monotone access structures via two moduli]\label{note} Let us examine the benefits of using two moduli and two set-systems. We know from Section~\ref{work} that access structure tokens operate over a fixed set $H$ and its $s^{l-1}$ proper supersets. Also recall that $H$ does not exactly represent the minimal authorized subset $\mathrm{\Omega}$, instead it is a randomly sampled set, picked to enforce the desired access structure $\Gamma = $ cl$(\mathrm{\Omega})$. Having a set $H$ with $s^{l-1}$ proper supersets in $\mathcal{H}'$ and $l$ proper subsets in $\mathcal{H}$ enables us to use \emph{carefully} generated access structure tokens to capture the $l$ minimal authorized subsets that are represented by subsets of $H$ in $\mathcal{H}$. Let $\widetilde{\mathcal{H}} \in \mathcal{H}$ denote the collection of $l$ proper subsets of $H$, and $\widehat{\mathcal{H}} \in \mathcal{H}'$ denote the $s^{l-1}$ proper supersets of $H$. Update these as: $\widetilde{\mathcal{H}} = \widetilde{\mathcal{H}} \cup H$ and $\widehat{\mathcal{H}} = \widehat{\mathcal{H}} \cup H$ to denote the collections of $l+1$ subsets and $s^{l-1}+1$ supersets of $H$, respectively. Further, let $\wp = \{\Gamma_1, \ldots, \Gamma_{l+1}\}$ denote the family of monotone access structures that originate from the family of minimal authorized subsets $\{\mathrm{\Omega}_i\}_{i=1}^{l+1}$, where $\Gamma_i = $ cl$(\mathrm{\Omega}_i)$ for $i \in [l+1]$. Let $\{\mathrm{\mho}^{(\wp)}_i\}_{i=1}^{\ell}$ denote the access structure tokens that capture the $l+1$ access structures in $\wp$. Then, for an access structure token combining function $f$, it follows that the following holds for all subsets of parties $\mathcal{A} \in \wp$:
\[f(\{\mathrm{\mho}^{(\wp)}\}_{i \in \mathcal{A}}) = |H \cap \widetilde{H}| = 0 \bmod m \quad \text{OR} \quad f(\{\mathrm{\mho}^{(\wp)}\}_{i \in \mathcal{A}}) = |H \cap \widehat{H}| = 0 \bmod m',\]
where $\widehat{H} \in \widehat{\mathcal{H}}$ and $\widetilde{H} \in \widetilde{\mathcal{H}}$. Since $\bmod ~m$ and $\bmod ~m'$ correspond to the set-systems $\mathcal{H}$ and $\mathcal{H}'$, respectively, we need two moduli to realize this functionality. Note that the access structure token generation procedure discussed in Section~\ref{work} is not suitable to achieve this goal but any procedure that operates over the sets in a manner that guarantees that the outputs of unions remain inside the collection of some fixed set-systems can be used to harness the power of two moduli to achieve significant improvements over the current known upper bound of $2^{.637\ell + o(\ell)}$ on the share size for secret sharing for general (monotone) access structures. 
\end{note}

We are now ready to present our detailed access structure hiding verifiable secret sharing scheme that works with two moduli. Since we do not have an access structure encoding procedure that satisfies the properties outlined in Note~\ref{note}, we use our access structure encoding procedures from Section~\ref{sec5}. As explained in \Cref{high-level}, we arrange the $\ell$ parties as nodes in a DAG $G$. Without loss of generality, we assume that the parties lie on a single directed path, as shown in Figure~\ref{Fig4}. Each party $P_i \in \mathcal{P}$ operates in: 
\begin{itemize}
	\item $\mathbb{Z}_q$ if $i = 1 \bmod 2$, 
	\item $\mathbb{Z}_{q^\prime}$ if $i = 0 \bmod 2.$
\end{itemize}

\begin{figure}[h!]
	\centering
	\begin{tikzpicture}[>=stealth,->,auto,node distance = 2.5cm,thick,shorten >=1pt]
	\tikzstyle{every state}=[thick, minimum size = 11mm]
	\node[state] (P1) [label=$P_1$] {$v_1$};
	\node[state] (P2) [right=1.4 of P1,label = {{\large \textbf{s}}, $P_i$}] {$v_i$};
	\node[state] (P3) [right=1.4 of P2, label={$P_{i+ \lceil |\mathrm{\Omega}/2| \rceil}$}] {$v_{i+ \lceil |\mathrm{\Omega}/2| \rceil}$};
	\node[state] (P4) [right=1.4 of P3,label = {{\large \textbf{t}}, $P_{i+ |\mathrm{\Omega}|}$}] {$v_{i+ |\mathrm{\Omega}|}$};
	\node[state] (Pl) [right=1.4 of P4, label={$P_\ell$}] {$v_\ell$};
	
	\path[every node/.style={fill=white,inner sep=1pt}]
	(P2)edge [bend right=40] node[left=0mm] {$\textbf{D}$} (P4);
	\path[dotted,->] (P1) edge node {} (P2);
	\path[dotted,->] (P2) edge node {} (P3);
	\path[dotted,->] (P3) edge node {} (P4);
	\path[dotted,->] (P4) edge node {} (Pl);
	\end{tikzpicture}
	\caption{Parties $\mathcal{P} = \{P_1, \ldots, P_\ell\}$ arranged as a simple DAG: a generalization of Figure~\ref{Fig3}.}\label{Fig4}
\end{figure}
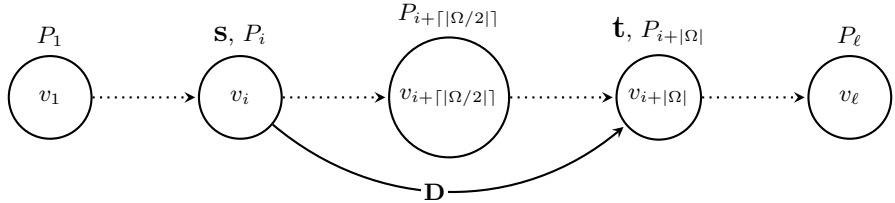

For a prime $p = 2m+1$, let $q=pc$ and $q^\prime=pc^\prime$, such that $p\nmid c, c^\prime$ and $|q'| = |q| + \epsilon(|q|)$, where $\epsilon$ is a negligible function (see Definition~\ref{Neg}). We ensure that $q,q^\prime=(d\lambda)^{\Theta(d)}$ such that the following holds:
$$q<q^\prime,\quad p\leq\sqrt{\log q}\quad\text{and}\quad\frac{2p-1}{2p}<\frac{q}{q^\prime}<\frac{2p}{2p+1}.$$ 
The other parameters are chosen as: $n=\Theta(d\lambda\log(d\lambda))$ and $w=\Theta(n\log q)=\Theta(d^2\lambda\log^2(d\lambda))$. 

The secret $k (\neq 0) \in \mathbb{Z}_p$ is a primitive root modulo $p$, and gets encoded using a $n \times n$ matrix $\mathbf{S}$ such that $||\mathbf{S}||<p$ and $\det(\mathbf{S})=k$. Given a minimal authorized subset $\mathrm{\Omega} \in \Gamma_0$, the dealer uses $H \in \mathcal{H} \cap \mathcal{H}'$ to generate $\ell$ access structure tokens $\{\mathrm{\mho}^{(\Gamma)}_i\}_{i=1}^\ell \in \mathbb{Z}_m \cup \mathbb{Z}_{m'} \setminus \{0\}$ (as defined by Note~\ref{NoteImp} in Section~\ref{sec5}) that capture the access structure $\Gamma = $ cl$(\mathrm{\Omega})$. For $s=\sqrt{n}, \sigma=\Theta(\sqrt{n\log q})=\Theta(\sqrt{n\log q^\prime})$, and security parameter $\lambda$, the dealer generates following for each party $P_i \in \mathcal{P}$: 
\begin{itemize}
	\item sample a $w \times n$ matrix $\mathbf{A}_i$ such that $||\mathbf{A}_i||<p$, 
	\item compute the `trapdoor information' $\tau_i$ for $\textbf{A}_i$ using the lattice-trapdoor generation algorithm from~\cite{Micci[12]} and a fixed generator matrix $G$,  
	\item sample a $w \times n$ matrix $\mathbf{E}_i$ from the discrete Gaussian distribution $\chi=D_{\mathbb{Z},s}$ subject to the restriction that $||\mathbf{E}_i||<s\sqrt{\lambda}$,
	\item except for a randomly picked party $P_j$, compute the share for party $P_i$ as: $\textbf{S}_i = \textbf{S}^{\mathrm{\mho}^{(\Gamma)}_i} \bmod p$. The share for party $P_j$ is generated as: $\textbf{S}_j = \textbf{S}^{\mathrm{\mho}^{(\Gamma)}_j+ 1} \bmod p,$
	\item use $\tau_i$ to compute a $w \times w$ encoding $\textbf{D}_i$ of $P_i$'s share $\textbf{S}_i$ such that the following relations hold (source-sink; see Figure~\ref{Fig4}):
		\begin{align*}
		\mathbf{D}_1\mathbf{A}_1&=\mathbf{A}_2\mathbf{S}_1+\mathbf{E}_1 \\
		\mathbf{D}_2\mathbf{A}_2&=\mathbf{A}_3\mathbf{S}_2+\mathbf{E}_2 \\
		&\ \vdots \\
		\mathbf{D}_{\ell-1}\mathbf{A}_{\ell-1}&=\mathbf{A}_\ell\mathbf{S}_{\ell-1}+\mathbf{E}_{\ell-1},
		\end{align*}
	where $||\mathbf{D}_i||<\sigma\sqrt{\lambda}$.
\end{itemize}

Since the entries of $\mathbf{A}_i$ are bounded by $p$, the following follows from our selection of $q$ and $q'$: 
$$\left\lfloor\frac{q^\prime}{q}\mathbf{A}_i\right\rceil=\mathbf{A}_i\text{ for odd }i,\quad\text{and}\quad\left\lfloor\frac{q}{q^\prime}\mathbf{A}_i\right\rceil=\mathbf{A}_i\text{ for even }i.$$
This means that we may naturally interpret the entries of $\mathbf{A}_i$'s as being in both $\mathbb{Z}_q$ and $\mathbb{Z}_{q^\prime}$. \\
\textit{Notations:} The following notations are used frequently throughout the rest of this section.
\begin{itemize}
	\item Without loss of generality, let $v_i$ be the node housing the party $P_i$ for all $i \in [r]$. 
	\item We use $\vv{\prod}$ to denote a product that is computed in the order that is defined by the relative positions of the nodes present in the given directed path, from source to sink. We call such products \emph{in-order}. For instance, the following denotes the in-order product of the `trapdoor information' of all parties in the DAG depicted in Figure~\ref{Fig3}:
	\[\vv{\prod_{i \in \mathcal{P}}} \tau_i = \tau_5 \tau_1 \tau_2 \tau_3 \tau_4 \tau_6 \tau_7.\]
	\item Similarly, if the multiplications are performed in the opposite order to what is defined by the given directed path; i.e., the multiplications are performed from sink to source, then we call it \emph{reverse-order} product and denote it as $\protect\cv{\prod}$. For example, the reverse-order product of the `trapdoor information' of all parties in the DAG depicted in Figure~\ref{Fig3} is: $\protect\cv{\prod\limits_{i \in \mathcal{P}}} \tau_i = \tau_7 \tau_6 \tau_4 \tau_3 \tau_2 \tau_1 \tau_5.$
	\item For a subset of parties $\mathcal{A} \subseteq \mathcal{P}$, which forms a directed path $\mathfrak{P}$ in the DAG $G$, let $P^{(\mathcal{A})}_\triangleright$   denote the party that is housed by the node that is at the beginning of $\mathfrak{P}$. Similarly, let $P^{(\mathcal{A})}_\triangleleft$ denote the party that is housed by the node that is at the end of $\mathfrak{P}$.
\end{itemize}
Let $\tau^{(\mathcal{A})}_\triangleleft$ denote the `trapdoor information' corresponding to the matrix $\textbf{A}^{(\mathcal{A})}_\triangleleft$ of party $P^{(\mathcal{A})}_\triangleleft$. Each party $P_i \in \mathcal{P}$ receives its share as: $\{\mathrm{\mho}^{(\Gamma)}_i, \mathrm{\Psi}_i^{(k)}\}$, where $\mathrm{\Psi}_i^{(k)} = \{\textbf{A}_i, \tilde{\tau}_i, \textbf{D}_i\}$ and $\tilde{\tau}_i$ is randomly sampled such that it holds for all subsets of parties $\mathcal{A} \subseteq \mathcal{P}$ that: $\tau^{(\mathcal{A})}_\triangleleft= \vv{\prod\limits_{i \in \mathcal{A}}} \tilde{\tau}_i$ (in $\mathbb{Z}_q$ or $\mathbb{Z}_{q'}$, depending on the value of $i \bmod 2$) if and only if $\mathcal{A} \supseteq \mathrm{\Omega}$, i.e., $\mathcal{A} \in \Gamma$. 

\subsection{\textbf{Secret Reconstruction and Correctness}} 
In order to reconstruct the secret, any subset of parties $\mathcal{A} \subseteq \mathcal{P}$ first combine their access structure tokens $\{\mathrm{\mho}^{(\Gamma)}_i\}_{i \in \mathcal{A}}$ and verify that: 
\[\sum\limits_{i \in \mathcal{A}} \mathrm{\mho}^{(\Gamma)}_i = 0 \bmod m  \quad \text{OR} \quad \sum\limits_{i \in \mathcal{A}} \mathrm{\mho}^{(\Gamma)}_i = 0 \bmod m'.\]
It follows from Section~\ref{sec5} that the access structure tokens can be generate such that above condition holds for any authorized subset of parties $\mathcal{A} \in \Gamma$, while for all unauthorized subsets $\mathcal{B} \notin \Gamma$, it holds that: $$\sum\limits_{i \in \mathcal{B}} \mathrm{\mho}^{(\Gamma)}_i \neq 0 \bmod m \quad \text{AND} \quad \sum\limits_{i \in \mathcal{B}} \mathrm{\mho}^{(\Gamma)}_i \neq 0 \bmod m'.$$

Once it is established that $\mathcal{A} \in \Gamma$, then the parties combine their encodings $\{\textbf{D}_i\}_{i \in \mathcal{A}}$, in the correct order as: 
\begin{equation}\label{EQQ}
\cv{\prod\limits_{i \in \mathcal{A}}} \textbf{D}_i \textbf{A}^{(\mathcal{A})}_\triangleright = \textbf{D} \textbf{A}^{(\mathcal{A})}_\triangleright= \textbf{A}^{(\mathcal{A})}_\triangleleft \prod_{i \in \mathcal{A}}\textbf{S}^{\mathrm{\mho}^{(\Gamma)}_i} + \textbf{E}^\prime = \textbf{A}^{(\mathcal{A})}_\triangleleft \textbf{S}^{\sum_{i \in \mathcal{A}}\mathrm{\mho}^{(\Gamma)}_i +1} + \textbf{E}^\prime. 
\end{equation}
Recall that $\textbf{A}^{(\mathcal{A})}_\triangleright$ and $\textbf{A}^{(\mathcal{A})}_\triangleleft$ respectively denote the matrices of the parties housed by the first and final nodes in the directed path formed by the nodes housing the parties in $\mathcal{A}$. Depending on the value of $i \mod 2$, each party $P_i \in \mathcal{P}$ operates within its respective $\bmod ~q$ or $\bmod~ q'$ world. Without loss of generality, let $P^{(\mathcal{A})}_\triangleleft$ operate in modulo $q$ world. Recall that for $\mathcal{A} \in \Gamma$, it holds that: $$\sum\limits_{i \in \mathcal{A}} \mathrm{\mho}^{(\Gamma)}_i = 0 \bmod m\quad \text{OR} \quad \sum\limits_{i \in \mathcal{A}} \mathrm{\mho}^{(\Gamma)}_i = 0 \bmod m',$$ i.e., it holds that: $$\sum\limits_{i \in \mathcal{A}} \mathrm{\mho}^{(\Gamma)}_i = c(p-1),$$ where $c \geq 1$ is an integer. Recall that the size of all sets in $\mathcal{H}$ and $\mathcal{H}'$ is an even multiple of $m$ and $p = 2m+1$. Therefore, sizes of the intersections between any subset-superset pairs must also be even multiples of $m$. Hence, for all $ \vartheta = 0 \bmod m$ and/or $\vartheta = 0 \bmod m'$, it holds that $\vartheta = c(p-1)$, where $c \geq 1$ is an integer. 

Recall that $q = pc$ and $q' = pc'$, where $p \nmid c, c'$. Hence, for authorized subsets of parties $\mathcal{A}\in \Gamma$, \Cref{EQQ} equates to:
\begin{equation}\label{EQ}
	\textbf{A}^{(\mathcal{A})}_\triangleleft \textbf{S}^{\left\langle \textbf{v}, \sum_{i \in \mathcal{A}}\textbf{v}_i \right\rangle +1} + \textbf{E}^\prime =\textbf{A}^{(\mathcal{A})}_\triangleleft \textbf{S}^{c(p-1) + 1} + \textbf{E}^\prime.
\end{equation}

We know that only authorized subset of parties $\mathcal{A} \in \Gamma$ can combine their \textit{trapdoor shares} $\{\tilde{\tau}_i\}_{i \in \mathcal{A}}$ to generate the trapdoor $\tau^{(\mathcal{A})}_\triangleleft$ required to invert $\textbf{A}_\mathcal{A}$. Hence, it follows from \Cref{EQQ,EQ} that for small $\textbf{E}^\prime$, the LWE inversion algorithm from~\cite{Micci[12]} can be used with $\tau^{(\mathcal{A})}_\triangleleft$ to compute matrix $\textbf{S}^{c(p-1)+1}$. We know from Fermat's little theorem (See \Cref{Fermat}) that: $$\det(\textbf{S})^{c(p-1)+1} = \det(\textbf{S})^{0+1} \bmod p.$$ Hence, the secret can be recovered as $\det(\textbf{S})^{c(p-1)+1} = \det(\textbf{S}) \bmod p = k$. Next, we prove that $\textbf{E}^\prime$ is indeed small for a bounded number of parties.

\begin{lemma}
	It holds that the largest $\mathbf{E}'$, computed by combining all encodings $\mathbf{D}_i$ as:
	$$\mathbf{D}_{\ell-1}\mathbf{D}_{\ell-2}\cdots\mathbf{D}_1\mathbf{A}_1=\mathbf{A}_\ell\mathbf{S}_{\ell-1}\mathbf{S}_{\ell-2}\cdots\mathbf{S}_1+\mathbf{E}^\prime,$$
	has entries bounded by $O\left(\sqrt{d^{6\ell-11}\lambda^{4\ell-5}\log^{6\ell-11}(d\lambda)}\right)$.
\end{lemma}

\begin{proof}
	In the setting depicted in Figure~\ref{Fig4}, if we combine the shares from parties $P_1$ and $P_2$, we obtain:
	\begin{align*}
	\mathbf{D}_2\mathbf{D}_1\mathbf{A}_1&=\mathbf{D}_2\mathbf{A}_2\mathbf{S}_1+\mathbf{D}_2\mathbf{E}_1 \\
	&=(\mathbf{A}_3\mathbf{S}_2+\mathbf{E}_2)\mathbf{S}_1+\mathbf{D}_2\mathbf{E}_1 \\
	&=\mathbf{A}_3\mathbf{S}_2\mathbf{S}_1+\mathbf{E}_2^\prime,
	\end{align*}
	where $\mathbf{E}_2^\prime=\mathbf{E}_2\mathbf{S}_1+\mathbf{D}_2\mathbf{E}_1$. Hence, it follows that:
	\begin{align*}
	||\mathbf{E}_2^\prime||&<n\cdot||\mathbf{E}_2||\cdot||\mathbf{S}_1||+m\cdot||\mathbf{D}_2||\cdot||\mathbf{E}_1|| \\
	&=O\left(\sqrt{d^7\lambda^7\log^7(d\lambda)}\right).
	\end{align*}
	
	If we now combine this with the share from party $P_3$, we obtain:
	\begin{align*}
	\mathbf{D}_3\mathbf{D}_2\mathbf{D}_1\mathbf{A}_1&=\mathbf{D}_3\mathbf{A}_3\mathbf{S}_2\mathbf{S}_1+\mathbf{D}_3\mathbf{E}_2^\prime \\
	&=(\mathbf{A}_4\mathbf{S}_3+\mathbf{E}_3)\mathbf{S}_2\mathbf{S}_1+\mathbf{D}_3\mathbf{E}_2^\prime \\
	&=\mathbf{A}_4\mathbf{S}_3\mathbf{S}_2\mathbf{S}_1+\mathbf{E}_3^\prime,
	\end{align*}
	where $\mathbf{E}_3^\prime=\mathbf{E}_3\mathbf{S}_2\mathbf{S}_1+\mathbf{D}_3\mathbf{E}_2^\prime$. Then,
	\begin{align*}
	||\mathbf{E}_3^\prime||&<n^2\cdot||\mathbf{E}_3||\cdot||\mathbf{S}_2||\cdot||\mathbf{S}_1||+m\cdot||\mathbf{D}_3||\cdot||\mathbf{E}_2^\prime|| \\
	&=O\left(\sqrt{d^{13}\lambda^{11}\log^{13}(d\lambda)}\right).
	\end{align*}
	
	Therefore, by induction, it follows for any $||\mathbf{E}^\prime||$ that:
	\begin{align*}
	||\mathbf{E}^\prime|| \leq ||\mathbf{E}_{\ell-1}^\prime||&=O\left(\sqrt{d^{6(\ell-1)-5}\lambda^{4(\ell-1)-1}\log^{6(\ell-1)-5}(d\lambda)}\right) \\
	&=O\left(\sqrt{d^{6\ell-11}\lambda^{4\ell-5}\log^{6\ell-11}(d\lambda)}\right). 
	\end{align*} $\qed$
\end{proof}

\subsection{\textbf{Secret and Share Verification}} After a successful secret reconstruction, any honest party $P_i \in \mathcal{A}$ in any authorized subset $\mathcal{A} \in \Gamma$ can verify the correctness of all shares $\{\textbf{S}_i\}_{i \in \mathcal{A}}$ from the reconstructed secret $k$. The verification is performed by removing $P_i$ from the directed path to $P^{(\mathcal{A})}_\triangleleft$ and then using $\tau^{(\mathcal{A})}_\triangleleft$ to invert the resulting \textsf{PRIM-LWE} instance. For example, if the directed path formed by the nodes housing the parties in $\mathcal{A}$ contains $P_i$\sampleline{thick,->}$P_j$\sampleline{thick,->}$P_t$ at the end, where $t = 1 \bmod 2$, then party $P_i$'s share can be verified by computing:
\[\textbf{D}_t \textbf{D}_j \textbf{A}_j = \textbf{A}_t \textbf{S}^{\mathrm{\mho}_j^{(\Gamma)} + \mathrm{\mho}_t^{(\Gamma)}} + \textbf{E}',\]
and inverting the output by using `trapdoor information' $\tau_t$  to find $\textbf{S}^{\mathrm{\mho}_j^{(\Gamma)} + \mathrm{\mho}_t^{(\Gamma)}}$, and use $\mathrm{\mho}_j^{(\Gamma)}, \mathrm{\mho}_t^{(\Gamma)}$ to verify its consistency with $k = \det(\textbf{S}) \bmod p$. 

Note that with our verification procedure, there can be a non-negligible probability of an inconsistent share to pass as valid due to random chance. For a setting with $\lceil \ell/2 \rceil$ malicious parties, the probability of an inconsistent share passing the verification of all honest parties is: $(1/p-1)^{\left\lfloor \frac{\ell}{2} \right\rfloor}.$ Hence, larger values of $p$ lead to smaller probabilities of verification failures by all honest parties. Note that unlike traditional VSS schemes, our scheme does not guarantee that all honest parties recover a consistent secret. Instead, it allows detection of malicious behavior without requiring any additional communication or cryptographic subroutines.

\subsection{\textbf{Maximum Share Size}} 
Since our access structure hiding verifiable secret sharing scheme works with minimal authorized subsets, its maximum share size is achieved when the access structure contains the largest possible number of minimal authorized subsets. For a set of $\ell$ parties, the maximum number of unique minimal authorized subsets in any access structure is $\binom{\ell}{\ell/2}$. Recall that the secret $k = \det(\textbf{S})$ belongs to $\mathbb{Z}_p$. Hence, for $|q| \approx |q'|$, it holds that $|k| \approx \sqrt{q}$. For each minimal authorized subset, every party $P_i \in \mathcal{P}$ receives a share of size $q(2mn + 1)$. Since $q = \poly(n)$ and the size of each access structure token is $\Theta(h)$ (see Section~\ref{sec5}), the maximum share size is (using results from~\cite{Das[20]}):
\begin{align*}
	\max\left(\mathrm{\Psi}^{(k)}_i\right) &\leq \binom{\ell}{\ell/2}(\sqrt{q}(2mn + 1) + \Theta(h))\\ &= (1+ o(1)) \dfrac{2^{\ell}}{\sqrt{\pi \ell/2}}(\sqrt{q}(2 q^{\varrho} + 1)+ \Theta(h)) \\ &= (1+ o(1)) \dfrac{2^{\ell}}{\sqrt{\pi \ell/2}}(2 q^{\varrho + 0.5} + \sqrt{q} + \Theta(h)),
\end{align*}
where $\varrho \leq 1$ is a constant and $h$ is the number of elements over which our set-systems are defined.\\
\textit{Possible Improvements:}
If one is able to realize the access structure encoding that is described in Note~\ref{note}, then the maximum share size drops by a factor of $l \geq 2$ by using that procedure instead of the one that we used in our scheme. Hence, the maximum share size (with respect to the secret size) of the resulting scheme would be:
\begin{align*}
\max\left(\mathrm{\Psi}^{(k)}_i\right) &\leq \dfrac{1}{l+1} \binom{\ell}{\ell/2}\sqrt{q}(2mn + 2)\\ &= \dfrac{1}{l+1} \left( (1+ o(1)) \dfrac{2^{\ell}}{\sqrt{\pi \ell/2}}\sqrt{q}(2 q^{\varrho} + 2) \right)\\ &= \dfrac{1}{l+1} \left( (1+ o(1)) \dfrac{2^{\ell}}{\sqrt{\pi \ell/2}}(2 q^{\varrho + 0.5} + 2\sqrt{q}) \right),
\end{align*}
where $\varrho \leq 1$ is a constant and $l \geq 2$ is as defined by Theorem~\ref{mainThm}.

\subsection{\textbf{Secrecy and Privacy}}
\Cref{lemma1,lemma2} establish perfect completeness and perfect soundness of our scheme, respectively (see Definition~\ref{MainDef}). We argued about perfect correctness while explaining the secret reconstruction procedure. Hence, we are left with proving computational secrecy and statistical hiding.

\begin{theorem}[Statistical Hiding]\label{profL}
	Every unauthorized subset $\mathcal{B} \notin \Gamma$ can identify itself to be outside $\Gamma$ by using its set of access structure tokens, $\{\mathrm{\mho}^{(\Gamma)}_i\}_{i \in \mathcal{B}}$. Given that the decision-LWE problem is hard, the following holds for all unauthorized subsets $\mathcal{B} \notin \Gamma$ and all access structures $\Gamma' \subseteq 2^{\mathcal{P}}$, where $\Gamma \neq \Gamma'$ and $\mathcal{B} \notin \Gamma'$:  
	\[\Big|\Pr[\Gamma~|~ \{\mathrm{\mho}^{(\Gamma)}_i\}_{i \in \mathcal{B}}] - \Pr[\Gamma'~|~ \{\mathrm{\mho}^{(\Gamma)}_i\}_{i \in \mathcal{B}}]\Big| = 2^{-\omega}, \]
	where $\omega = |\mathcal{P} \setminus \mathcal{B}|$ is the security parameter.
\end{theorem}
\begin{proof}
	It follows from Lemma~\ref{lemma2} that the following holds for all unauthorized subsets of parties $\mathcal{B} \notin \Gamma$:
	\[ \sum\limits_{i \in \mathcal{B}} \mathrm{\mho}^{(\Gamma)}_i \neq 0 \bmod m,\] 
	i.e., any unauthorized subset $\mathcal{B} \notin \Gamma$ can use its access structure tokens to identify itself as outside of the access structure $\Gamma$. The security parameter $\omega = |\mathcal{P} \setminus \mathcal{B}|$ accounts for this minimum information that is available to any unauthorized subset $\mathcal{B} \notin \Gamma$.
	
	We know that the set $H \xleftarrow{\; \$ \;} \mathcal{H}$ is randomly sampled. Furthermore, the access structure tokens $\{\mathrm{\mho}^{(\Gamma)}_i\}_{i \in \mathcal{B}}$ given to the parties are permuted according to a random permutation $\gamma$. Hence, it follows from the randomness of $H$ and $\gamma$ that:
	\[\Big|\Pr[\Gamma~|~ \{\mathrm{\mho}^{(\Gamma)}_i\}_{i \in \mathcal{B}}] - \Pr[\Gamma'~|~ \{\mathrm{\mho}^{(\Gamma)}_i\}_{i \in \mathcal{B}}]\Big| = 2^{-\omega}. \eqno \qed\]
\end{proof}

\begin{theorem}[Computational Secrecy]\label{ThMM}
	Given that decision-LWE problem is hard, it holds for every unauthorized subset $\mathcal{B} \notin \Gamma$ and all different secrets $k_1, k_2 \in \mathcal{K}$ that the distributions $\{\mathrm{\Psi}^{(k_1)}\}_{i \in \mathcal{B}}$ and $\{\mathrm{\Psi}^{(k_2)}\}_{i \in \mathcal{B}}$ are computationally indistinguishable with respect to the security parameter $\varepsilon \cdot |\mathcal{B}|$, where $\varepsilon$ denotes the advantage of a polynomial-time adversary against a \textsf{PRIM-LWE} instance.
\end{theorem}
\begin{proof}
	Recall that $\mathrm{\Psi}_i^{(k)} = \{\textbf{A}_i, \tilde{\tau}_i, \textbf{D}_i\}$. We know that the `trapdoor shares' $\{\tilde{\tau}_i\}_{i \in [\ell]}$ are generated randomly such that the following holds only for authorized subsets of parties $\mathcal{A} \in \Gamma$: $$\tau^{(\mathcal{A})}_\triangleleft= \vv{\prod\limits_{i \in \mathcal{A}}} \tilde{\tau}_i.$$
	Hence, it follows from one-time pad that $\{\tilde{\tau}_i\}_{i \in \mathcal{B}}$ leaks no information about the trapdoors of any matrix $\{\textbf{A}_i\}_{i \in \mathcal{B}}$. It follows from the hardness of \textsf{PRIM-LWE} that the pairs $(\textbf{A}_i, \textbf{D}_i)_{i \in \mathcal{B}}$ do not leak any non-negligible information to any unauthorized subset of parties $\mathcal{B} \notin \Gamma$ because it cannot reconstruct the correct trapdoor $\tau^{(\mathcal{B})}_\triangleleft$. Hence, it follows that the distributions $\{\mathrm{\Psi}^{(k_1)}\}_{i \in \mathcal{B}}$ and $\{\mathrm{\Psi}^{(k_2)}\}_{i \in \mathcal{B}}$ are computationally indistinguishable with respect to the security parameter $\varepsilon \cdot |\mathcal{B}|$, where $\varepsilon$ denotes the advantage of a polynomial-time adversary against a \textsf{PRIM-LWE} instance. $\qed$
\end{proof}

\section{Conclusion}\label{conclude}
Secret sharing is a foundational and versatile tool with direct applications to many useful cryptographic protocols. Its applications include multiple privacy-preserving techniques, but the privacy-preserving guarantees of secret sharing itself have not received adequate attention. In this paper, we bolstered the privacy-preserving guarantees and verifiability of secret sharing by extending a recent work of Sehrawat and Desmedt~\cite{Vipin[20]} wherein they introduced hidden access structures that remain unknown until some authorized subset of parties assembles. Unlike the solution from~\cite{Vipin[20]}, our scheme tolerates malicious parties and supports all possible monotone access structures. We introduced an approach to combine the learning with errors (LWE) problem with our novel superpolynomial sized set-systems to realize secret sharing for \textit{all} monotone hidden access structures. Our scheme is the first secret sharing solution to support malicious behavior identification and share verifiability in malicious-majority settings. It is also the first LWE-based secret sharing scheme for general access structures. As the building blocks of our scheme, we constructed a novel set-system with restricted intersections and introduced a new variant of the LWE problem, called \textsf{PRIM-LWE}, wherein the secret matrix is sampled from the set matrices whose determinants are generators of $\mathbb{Z}_q^*$, where $q$ is the LWE modulus. We also gave concrete directions for future work that will reduce our scheme's share size to be smaller than the best known upper bound for secret sharing over general (i.e., all monotone) access structures.

\bibliographystyle{plainurl}
\bibliography{CO2020}

\begin{thebibliography}{100}

\bibitem{Abraham[08]}
Ittai Abraham, Danny Dolev, and Joseph~Y. Halpern.
\newblock An almost-surely terminating polynomial protocol for asynchronous
  {B}yzantine agreement with optimal resilience.
\newblock In {\em ACM Symposium on Principles of Distributed Computing
  {(PODC)}}, pages 405--414, 2008.

\bibitem{AgaMazu[16]}
Abhishek Agarwal and Arya Mazumdar.
\newblock Security in locally repairable storage.
\newblock {\em IEEE Transactions on Information Theory}, 62(11):6204--6217,
  2016.

\bibitem{Shweta[11]}
Shweta Agrawal, David~Mandell Freeman, and Vinod Vaikuntanathan.
\newblock Functional encryption for inner product predicates from learning with
  errors.
\newblock In {\em ASIACRYPT}, pages 21--40, 2011.

\bibitem{Ajtai[96]}
Mikl\'{o}s Ajtai.
\newblock Generating hard instances of lattice problems (extended abstract).
\newblock In {\em STOC}, pages 99--108, 1996.

\bibitem{Ajtai[99]}
Mikl\'{o}s Ajtai.
\newblock Generating hard instances of the short basis problem.
\newblock In {\em ICALP}, pages 1--9, 1999.

\bibitem{Ajtai[97]}
Mikl\'{o}s Ajtai and Cynthia Dwork.
\newblock A public-key cryptosystem with worst-case/average-case equivalence.
\newblock In {\em STOC}, pages 284--293, 1997.

\bibitem{AjtaiFagin[90]}
Mikl\'{o}s Ajtai and Ronald Fagin.
\newblock Reachability is harder for directed than for undirected finite
  graphs.
\newblock {\em Journal of Symbolic Logic}, 55(1):113--150, 1990.

\bibitem{Ajtai[01]}
Mikl\'{o}s Ajtai, Ravi Kumar, and D.~Sivakumar.
\newblock A sieve algorithm for the shortest lattice vector problem.
\newblock In {\em STOC}, pages 601--610, 2001.

\bibitem{Adi[09]}
Adi Akavia, Shafi Goldwasser, and Vinod Vaikuntanathan.
\newblock Simultaneous hardcore bits and cryptography against memory attacks.
\newblock In {\em TCC}, pages 474--495, 2009.

\bibitem{NIST[20]}
Gorjan Alagic, Jacob Alperin-Sheriff, Daniel Apon, David Cooper, Quynh Dang,
  John Kelsey, Yi-Kai Liu, Carl Miller, Dustin Moody, Rene Peralta, Ray
  Perlner, Angela Robinson, and Daniel Smith-Tone.
\newblock Status report on the second round of the nist post-quantum
  cryptography standardization process, {NIST}.
\newblock URL:
  \url{https://nvlpubs.nist.gov/nistpubs/ir/2020/NIST.IR.8309.pdf}.

\bibitem{Tesla[20]}
Erdem Alkim, Paulo S. L.~M. Barreto, Nina Bindel, Juliane Kramer, Patrick
  Longa, and Jefferson~E. Ricardini.
\newblock The lattice-based digital signature scheme {qTESLA}, April 2020.
\newblock URL: \url{https://eprint.iacr.org/2019/085.pdf}.

\bibitem{AnaFan[19]}
Prabhanjan Ananth, Xiong Fan, and Elaine Shi.
\newblock Towards attribute-based encryption for {RAM}s from {LWE}: Sub-linear
  decryption, and more.
\newblock In {\em ASIACRYPT}, pages 112--141, 2019.

\bibitem{AnanJai[16]}
Prabhanjan Ananth, Aayush Jain, Huijia Lin, Christian Matt, and Amit Sahai.
\newblock Indistinguishability obfuscation without multilinear maps: New
  paradigms via low degree weak pseudorandomness and security amplification.
\newblock In {\em CRYPTO}, pages 284--332, 2019.

\bibitem{Ross[96]}
Ross~J. Anderson.
\newblock The eternity service.
\newblock In {\em PRAGOCRYPT}, pages 242--252, 1996.

\bibitem{Apple[19]}
Benny Applebaum, Amos Beimel, Oriol Farr\`{a}s, Oded Nir, and Naty Peter.
\newblock Secret-sharing schemes for general and uniform access structures.
\newblock In {\em EUROCRYPT}, pages 441--471, 2019.

\bibitem{Benny[20]}
Benny Applebaum, Amos Beimel, Oded Nir, and Naty Peter.
\newblock Better secret sharing via robust conditional disclosure of secrets.
\newblock In {\em ACM SIGACT Symposium on Theory of Computing (STOC)}, pages
  280--293, 2020.

\bibitem{Benny[09]}
Benny Applebaum, David Cash, Chris Peikert, and Amit Sahai.
\newblock Fast cryptographic primitives and circular-secure encryption based on
  hard learning problems.
\newblock In {\em CRYPTO}, 2009.

\bibitem{Fred[14]}
Frederico Araujo, Kevin Hamlen, Sebastian Biedermann, and Stefan Katzenbeisser.
\newblock From patches to honey-patches: Lightweight attacker misdirection,
  deception, and disinformation.
\newblock In {\em ACM SIGSAC Conference on Computer and Communications
  Security}, pages 942--953, 2014.

\bibitem{VaruDar[17]}
Varunya Attasena, J\'{e}r\^{o}me Darmont, and Nouria Harbi.
\newblock Secret sharing for cloud data security: a survey.
\newblock {\em The VLDB Journal}, 26:657--681, 2017.

\bibitem{Backes[13]}
Michael Backes, Amit Datta, and Aniket Kate.
\newblock Asynchronous computational {VSS} with reduced communication
  complexity.
\newblock In {\em Cryptographers' Track at the {RSA} Conference (CT-RSA)},
  pages 259--276, 2013.

\bibitem{Ban[15]}
Abhishek Banerjee, Georg Fuchsbauer, Chris Peikert, Krzysztof Pietrzak, and
  Sophie Stevens.
\newblock Key-homomorphic constrained pseudorandom functions.
\newblock In {\em TCC}, pages 31--60, 2015.

\bibitem{Ban[14]}
Abhishek Banerjee and Chris Peikert.
\newblock New and improved key-homomorphic pseudorandom functions.
\newblock In {\em CRYPTO}, pages 353--370, 2014.

\bibitem{Ban[12]}
Abhishek Banerjee, Chris Peikert, and Alon Rosen.
\newblock Pseudorandom functions and lattices.
\newblock In {\em EUROCRYPT}, pages 719--737, 2012.

\bibitem{Barr[94]}
David A.~Mix Barrington, Richard Beigel, and Steven Rudich.
\newblock Representing boolean functions as polynomials modulo composite
  numbers.
\newblock {\em Computational Complexity}, 4:367--382, 1994.

\bibitem{Basu[19]}
Soumya Basu, Alin Tomescu, Ittai Abraham, Dahlia Malkhi, Michael~K. Reiter, and
  Emin~G\"{u}n Sirer.
\newblock Efficient verifiable secret sharing with share recovery in {BFT}
  protocols.
\newblock In {\em {ACM SIGSAC} Conference on Computer and Communications
  Security (CCS)}, pages 2387--2402, 2019.

\bibitem{Beham[13]}
Michael Beham, Marius Vlad, and Hans~P. Reiser.
\newblock Intrusion detection and honeypots in nested virtualization
  environments.
\newblock In {\em 43rd Annual IEEE/IFIP International Conference on Dependable
  Systems and Networks (DSN)}, pages 1--6, 2013.

\bibitem{Beimel[11]}
Amos Beimel.
\newblock Secret-sharing schemes: A survey.
\newblock {\em Coding and Cryptology, Third International Workshop, IWCC},
  pages 11--46, 2011.

\bibitem{Beimel[15]}
Amos Beimel, Yuval Ishai, Ranjit Kumaresan, and Eyal Kushilevitz.
\newblock On the cryptographic complexity of the worst functions.
\newblock In {\em TCC}, pages 317--342, 2014.
\newblock Full Version (2017) available at:
  \url{https://www.microsoft.com/en-us/research/wp-content/uploads/2017/03/BIKK.pdf}.

\bibitem{Beimel[12]}
Amos {Beimel}, Yuval {Ishai}, Eyal {Kushilevitz}, and Ilan {Orlov}.
\newblock Share conversion and private information retrieval.
\newblock In {\em IEEE 27th Conference on Computational Complexity}, pages
  258--268, 2012.

\bibitem{AmHuss[20]}
Amos Beimel and Hussien Othman.
\newblock Evolving ramp secret sharing with a small gap.
\newblock In {\em EUROCRYPT}, pages 529--555, 2020.

\bibitem{Beimel[08]}
Amos Beimel and Anat Paskin.
\newblock On linear secret sharing for connectivity in directed graphs.
\newblock In {\em SCN}, pages 172--184, 2008.

\bibitem{BellKil[12]}
Mihir Bellare, Eike Kiltz, Chris Peikert, and Brent Waters.
\newblock Identity-based (lossy) trapdoor functions and applications.
\newblock In {\em EUROCRYPT}, pages 228--245, 2012.

\bibitem{Ben[88]}
Michael Ben-Or, Shafi Goldwasser, and Avi Wigderson.
\newblock Completeness theorems for non-cryptographic fault-tolerant
  distributed computation.
\newblock In {\em STOC}, pages 1--10, 1988.

\bibitem{FrabAks[21]}
Fabrice Benhamouda, Akshay Degwekar, Yuval Ishai, and Tal Rabin.
\newblock On the local leakage resilience of linear secret sharing schemes.
\newblock {\em Journal of Cryptology}, 34(10), 2021.

\bibitem{Beut[89]}
A.~Beutelspacher.
\newblock How to say `no'.
\newblock In {\em EUROCRYPT}, pages 491--496, 1989.

\bibitem{Hayo[19]}
Hayo~BaanSauvik Bhattacharya, Scott Fluhrer, Oscar Garcia-Morchon, Thijs
  Laarhoven, Ronald Rietman, Markku-Juhani~O. Saarinen, Ludo Tolhuizen, and
  Zhenfei Zhang.
\newblock Round5: Compact and fast post-quantum public-key encryption.
\newblock In {\em PQCrypto}, pages 83--102, 2019.

\bibitem{BlaCat[84]}
G.~R. Blakley and Catherine Meadows.
\newblock Security of ramp schemes.
\newblock In {\em CRYPTO}, pages 242--268, 1984.

\bibitem{Blakley[79]}
G.R. Blakley.
\newblock Safeguarding cryptographic keys.
\newblock {\em American Federation of Information Processing}, 48:313--318,
  1979.

\bibitem{BlumMicali[84]}
Manuel Blum and Silvio Micali.
\newblock How to generate cryptographically strong sequences of pseudo-random
  bits.
\newblock {\em SIAM J. on Computing}, 13(4):850--864, 1984.

\bibitem{Blundo[96]}
C.~Blundo and D.~R. Stinson.
\newblock Anonymous secret sharing schemes.
\newblock {\em Designs, Codes and Cryptography}, 2:357--390, 1996.

\bibitem{Blundo[92]}
Carlo Blundo, Alfredo~De Santis, Luisa Gargano, and Ugo Vaccaro.
\newblock On the information rate of secret sharing schemes.
\newblock In {\em CRYPTO}, pages 149--169, 1992.

\bibitem{Peter[09]}
Peter Bogetoft, Dan~Lund Christensen, Ivan Damg\r{a}rd, Martin Geisler, Thomas
  Jakobsen, Mikkel Kroigaard, Janus~Dam Nielsen, Jesper~Buus Nielsen, Kurt
  Nielsen, Jakob Pagter, Michael Schwartzbach, and Tomas Toft.
\newblock Secure multiparty computation goes live.
\newblock In {\em Financial Cryptography and Data Security}, pages 325--343,
  2009.

\bibitem{KimDan[17]}
Dan Boneh, Sam Kim, and Hart Montgomery.
\newblock Private puncturable {PRF}s from standard lattice assumptions.
\newblock In {\em EUROCRYPT}, pages 415--445, 2017.

\bibitem{Boneh[13]}
Dan Boneh, Kevin Lewi, Hart Montgomery, and Ananth Raghunathan.
\newblock Key homomorphic {PRF}s and their applications.
\newblock In {\em CRYPTO}, pages 410--428, 2013.

\bibitem{Bos[18]}
Joppe Bos, L\'{e}o Ducas, Eike Kiltz, Tancr\`{e}de Lepoint, Vadim Lyubashevsky,
  John~M. Schanckk, Peter Schwabe, Gregor Seiler, and Damien Stehl\'{e}.
\newblock {CRYSTALS}-kyber: A {CCA}-secure module-lattice-based {KEM}.
\newblock In {\em IEEE European Symposium on Security and Privacy (Euro S\&P)},
  pages 353--367, 2018.

\bibitem{Bos[16]}
Joppe~W. Bos, Craig Costello, L\'{e}o Ducas, Ilya Mironov, Michael Naehrig,
  Valeria Nikolaenko, Ananth Raghunathan, and Douglas Stebila.
\newblock Frodo: Take off the ring! practical, quantum-secure key exchange from
  {LWE}.
\newblock In {\em ACM SIGSAC Conference on Computer and Communications
  Security}, pages 1006--1018, 2016.

\bibitem{WBos[15]}
Joppe~W. Bos, Craig Costello, Michael Naehrig, and Douglas Stebila.
\newblock Post-quantum key exchange for the tls protocol from the ring learning
  with errors problem.
\newblock In {\em IEEE Symposium on Security and Privacy (S\&P)}, pages
  553--570, 2015.

\bibitem{Joppe[13]}
Joppe~W. Bos, Kristin Lauter, Jake Loftus, and Michael Naehrig.
\newblock Improved security for a ring-based fully homomorphic encryption
  scheme.
\newblock In {\em IMA International Conference on Cryptography and Coding},
  pages 45--64, 2013.

\bibitem{Boyen[17]}
Xavier Boyen and Qinyi Li.
\newblock All-but-many lossy trapdoor functions from lattices and applications.
\newblock In {\em CRYPTO}, pages 298--331, 2017.

\bibitem{Dott[18]}
Zvika Brakerski and Nico D\"{o}ttling.
\newblock Two-message statistically sender-private {OT} from {LWE}.
\newblock In {\em TCC}, pages 370--390, 2018.

\bibitem{Brak[14]}
Zvika Brakerski, Craig Gentry, and Vinod Vaikuntanathan.
\newblock (leveled) fully homomorphic encryption without bootstrapping.
\newblock {\em ACM Transactions on Computation Theory}, 6(3), July 2014.

\bibitem{Zvika[13]}
Zvika Brakerski, Adeline Langlois, Chris Peikert, Oded Regev, and Damien
  Stehl\'{e}.
\newblock Classical hardness of learning with errors.
\newblock In {\em STOC}, pages 575--584, 2013.

\bibitem{RotBra[17]}
Zvika Brakerski, Rotem Tsabary, Vinod Vaikuntanathan, and Hoeteck Wee.
\newblock Private constrained {PRF}s (and more) from {LWE}.
\newblock In {\em TCC}, pages 264--302, 2017.

\bibitem{Vinod[11]}
Zvika Brakerski and Vinod Vaikuntanathan.
\newblock Fully homomorphic encryption from ring-{LWE} and security for key
  dependent messages.
\newblock In {\em CRYPTO}, pages 505--524, 2011.

\bibitem{Zvika[15]}
Zvika Brakerski and Vinod Vaikuntanathan.
\newblock {C}onstrained {K}ey-{H}omomorphic {PRF}s from {S}tandard {L}attice
  {A}ssumptions -- or: {H}ow to {S}ecretly {E}mbed a {C}ircuit in {Y}our {PRF}.
\newblock In {\em TCC}, pages 1--30, 2015.

\bibitem{ZviVin[16]}
Zvika Brakerski, Vinod Vaikuntanathan, Hoeteck Wee, and Daniel Wichs.
\newblock Obfuscating conjunctions under entropic ring {LWE}.
\newblock In {\em ACM Conference on Innovations in Theoretical Computer
  Science}, pages 147--156, 2016.

\bibitem{Brick[89]}
Ernest~F. Brickell.
\newblock Some ideal secret sharing schemes.
\newblock {\em Journal of Combin. Math. and Combin. Comput.}, 6:105--113, 1989.

\bibitem{Bruna[20]}
Joan Bruna, Oded Regev, Min~Jae Song, and Yi~Tang.
\newblock Continuous {LWE}, 2020.
\newblock \href {http://arxiv.org/abs/2005.09595} {\path{arXiv:2005.09595}}.

\bibitem{Bryant[17]}
Darryn Bryant and Daniel Horsley.
\newblock Steiner triple systems without parallel classes.
\newblock {\em SIAM J. Discrete Math}, 31(4):693--696, 2017.

\bibitem{Cachin[02]}
Christian Cachin, Klaus Kursawe, Anna Lysyanskaya, and Reto Strobl.
\newblock Asynchronous verifiable secret sharing and proactive cryptosystems.
\newblock In {\em ACM CCS}, pages 88--97, 2002.

\bibitem{RanChen[17]}
Ran Canetti and Yilei Chen.
\newblock Constraint-hiding constrained {PRF}s for {NC}${}^1$ from {LWE}.
\newblock In {\em EUROCRYPT}, pages 446--476, 2017.

\bibitem{Canetti[93]}
Ran Canetti and Tal Rabin.
\newblock Fast asynchronous byzantine agreement with optimal resilience.
\newblock In {\em ACM Symposium on Theory of Computing {(STOC)}}, pages 42--51,
  1993.

\bibitem{Capo[93]}
Renato~M. Capocelli, Alfredo~De Santis, Luisa Gargano, and Ugo Vaccaro.
\newblock On the size of shares for secret sharing schemes.
\newblock {\em Journal of Cryptology}, 6(3):157--168, 1993.

\bibitem{Cascudo[17]}
Ignacio Cascudo and Bernardo David.
\newblock {SCRAPE}: Scalable randomness attested by public entities.
\newblock In {\em Applied Cryptography and Network Security (ACNS)}, pages
  537--556, 2017.

\bibitem{Chaum[88]}
David Chaum, Claude Cr\'{e}peau, and Ivan~Bjerre Damg\r{a}rd.
\newblock Multiparty unconditionally secure protocols.
\newblock In {\em STOC}, pages 11--19, 1988.

\bibitem{ChaoJun[20]}
Chaochao Chen, Jun Zhou, Li~Wang, Xibin Wu, Wenjing Fang, Jin Tan, Lei Wang,
  Xiaoxi Ji, Alex Liu, Hao Wang, and Cheng Hong.
\newblock When homomorphic encryption marries secret sharing: Secure
  large-scale sparse logistic regression and applications in risk control,
  2020.
\newblock \href {http://arxiv.org/abs/2008.08753} {\path{arXiv:2008.08753}}.

\bibitem{HaoRon[06]}
Hao Chen and Ronald Cramer.
\newblock Algebraic geometric secret sharing schemes and secure multi-party
  computations over small fields.
\newblock In {\em CRYPTO}, pages 521--536, 2006.

\bibitem{Chen[19]}
Yilei Chen, Nicholas Genise, and Pratyay Mukherjee.
\newblock Approximate trapdoors for lattices and smaller hash-and-sign
  signatures.
\newblock In {\em ASIACRYPT}, pages 3--32, 2019.

\bibitem{ChenLi[15]}
Zhenhua Chen, Shundong Li, Youwen Zhu, Jianhua Yan, and Xinli Xu.
\newblock A cheater identifiable multi-secret sharing scheme based on the
  chinese remainder theorem.
\newblock {\em Sec. and Commun. Netw.}, 8(18):3592–3601, 2015.

\bibitem{Chor[85]}
Benny Chor, Shafi Goldwasser, Silvio Micali, and Baruch Awerbuch.
\newblock Verifiable secret sharing and achieving simultaneity in the presence
  of faults.
\newblock In {\em FOCS}, pages 383--395, 1985.

\bibitem{YaoGuo[21]}
Yao-Hsin Chou, Guo-Jyun Zeng, Xing-Yu Chen, and Shu-Yu Kuo.
\newblock Multiparty weighted threshold quantum secret sharing based on the
  chinese remainder theorem to share quantum information.
\newblock {\em Scientific Reports}, 11, 2021.

\bibitem{AshC[11]}
Ashish Choudhury.
\newblock Simple and asymptotically optimal $t$-cheater identifiable secret
  sharing scheme.
\newblock Cryptology ePrint Archive, Report 2011/330, 2011.
\newblock \url{https://eprint.iacr.org/2011/330}.

\bibitem{AshArpAsh[11]}
Ashish Choudhury, Arpita Patra, B.V. Ashwinkumar, Kannan Srinathan, and
  C.~Pandu Rangan.
\newblock Secure message transmission in asynchronous networks.
\newblock {\em Journal of Parallel and Distributed Computing},
  71(8):1067--1074, 2011.

\bibitem{Gu[17]}
Gu~Chunsheng.
\newblock Multilinear maps using a variant of ring-{LWE}.
\newblock Cryptology ePrint Archive, Report 2017/342, 2017.
\newblock \url{https://eprint.iacr.org/2017/342}.

\bibitem{CleDan[99]}
Richard Cleve, Daniel Gottesman, and Hoi-Kwong Lo.
\newblock How to share a quantum secret.
\newblock {\em Phys. Rev. Lett.}, 83:648--651, 1999.

\bibitem{Col[92]}
Charles Colbourn, Spyros~S. Magliveras, and Rudolf~A. Mathon.
\newblock Transitive {S}teiner and {K}irkman triple systems of order 27.
\newblock {\em Mathematics of Computation}, 58(197):441--450, 1992.

\bibitem{Charles[06]}
Charles~J. Colbourn and Jeffrey~H. Dinitz.
\newblock {\em Handbook of {C}ombinatorial {D}esigns}.
\newblock Discrete Mathematics and Its Applications. Chapman and Hall/CRC,
  2006.

\bibitem{Tri[99]}
Charles~J. Colbourn and Alex Rosa.
\newblock {\em Triple Systems}.
\newblock Oxford Mathematical Monographs. Clarendon Press and Oxford University
  Press, 1999.

\bibitem{ConCurt[85]}
J.~H. Conway, R.~T. Curtis, S.~P. Norton, R.~A. Parker, and R.~A. Wilson.
\newblock {\em Atlas of Finite Groups: Maximal Subgroups and Ordinary
  Characters for Simple Groups}.
\newblock Clarendon Press, Oxford, England, 1985.

\bibitem{Corless[96]}
R.~M. Corless, G.~H. Gonnet, D.~E.~G. Hare, D.~J. Jeffrey, and D.~E. Knuth.
\newblock On the {L}ambert{W} function.
\newblock {\em Adv. Comput. Math.}, 5:329--359, 1996.

\bibitem{CousinDi[18]}
David~Bruce Cousins, Giovanni~Di Crescenzo, Kamil~Doruk G\"{u}r, Kevin King,
  Yuriy Polyakov, Kurt Rohloff, Gerard~W. Ryan, and Erkay Savas.
\newblock Implementing conjunction obfuscation under entropic ring {LWE}.
\newblock In {\em IEEE Symposium on Security and Privacy (S\&P)}, pages
  354--371, 2018.

\bibitem{Cramer[00]}
Ronald Cramer, Ivan Damg\r{a}rd, and Ueli Maurer.
\newblock General secure multi-party computation from any linear secret-sharing
  scheme.
\newblock In {\em EUROCRYPT}, pages 316--334, 2000.

\bibitem{CramDam[15]}
Ronald Cramer, Ivan~Bjerre Damg\r{a}rd, and Jesper~Buus Nielsen.
\newblock {\em Secure Multiparty Computation and Secret Sharing}.
\newblock Cambridge University Press, 2015.

\bibitem{Csi[96]}
L\'{a}szl\'{o} Csirmaz.
\newblock The dealer's random bits in perfect secret sharing schemes.
\newblock {\em Studia Sci. Math. Hungar.}, 32(3-4):429--437, 1996.

\bibitem{Csi[97]}
L\'{a}szl\'{o} Csirmaz.
\newblock The size of a share must be large.
\newblock {\em Journal of Cryptology}, 10(4):223--231, 1997.

\bibitem{DaiZha[18]}
Mingjun Dai, Shengli Zhang, Hui Wang, and Shi Jin.
\newblock A low storage room requirement framework for distributed ledger in
  blockchain.
\newblock {\em IEEE Access}, 6:22970--22975, 2018.

\bibitem{Jan[18]}
Jan-Pieter D'Anvers, Angshuman Karmakar, Sujoy~Sinha Roy, and Frederik
  Vercauteren.
\newblock Saber: Module-{LWR} based key exchange, {CPA}-secure encryption and
  {CCA}-secure {KEM}.
\newblock In {\em AFRICACRYPT}, pages 282--305, 2018.

\bibitem{Das[20]}
Shagnik Das.
\newblock A brief note on estimates of binomial coefficients.
\newblock URL: \url{http://page.mi.fu-berlin.de/shagnik/notes/binomials.pdf}.

\bibitem{Daza[07]}
Vanesa Daza and Josep Domingo-Ferrer.
\newblock On partial anonymity in secret sharing.
\newblock In {\em European Public Key Infrastructure Workshop}, pages 193--202,
  2007.

\bibitem{Lou[20]}
Louis De{B}iasio and Michael Tait.
\newblock Large monochromatic components in 3-edge-colored {S}teiner triple
  systems.
\newblock {\em Mathematics of Computation}, 28(6):428--444, 2020.

\bibitem{DanTho[15]}
Daniel Demmler, Thomas Schneider, and Michael Zohner.
\newblock {ABY} - a framework for efficient mixed-protocol secure two-party
  computation.
\newblock 2015.

\bibitem{YvoFrank[89]}
Yvo Desmedt and Yair Frankel.
\newblock Threshold cryptosystems.
\newblock In {\em CRYPTO}, pages 307--315, 1989.

\bibitem{Yvo[21]}
Yvo Desmedt, Songbao Mo, and Arkadii~M. Slinko.
\newblock Framing in secret sharing.
\newblock {\em IEEE Transactions on Information Forensics and Security},
  16:2836--2842, 2021.

\bibitem{Yvo[89]}
Yvo~G. Desmedt and Yair Frankel.
\newblock Shared generation of authenticators and signatures (extended
  abstract).
\newblock In {\em CRYPTO}, pages 457--469, 1991.

\bibitem{Danny[93]}
Danny Dolev, Cynthia Dwork, Orli Waarts, and Moti Yung.
\newblock Perfectly secure message transmission.
\newblock {\em Journal of the ACM (JACM)}, 40(1):17--47, 1993.

\bibitem{Dorn[04]}
Maximillian Dornseif, Thorsten Holz, and Christian~N. Klein.
\newblock {NoSEBrEaK} - attacking honeynets.
\newblock In {\em IEEE SMC Information Assurance Workshop}, pages 123--129,
  2004.

\bibitem{ConstantDr[15]}
Constantin~C\u{a}t\u{a}lin Dr\u{a}gan and Ferucio~Lauren\c{t}iu \c{T}iplea.
\newblock Key-policy attribute-based encryption for general boolean circuits
  from secret sharing and multi-linear maps.
\newblock In {\em BalkanCryptSec}, pages 112--133, 2015.

\bibitem{Dili[17]}
L\'{e}o Ducas, Eike Kiltz, Tancrède Lepoint, Vadim Lyubashevsky, Peter
  Schwabe, Gregor Seiler, and Damien Stehl\'{e}.
\newblock {CRYSTALS}-dilithium: A lattice-based digital signature scheme, 2017.
\newblock URL: \url{https://eprint.iacr.org/2017/633.pdf}.

\bibitem{Dung[16]}
Dung~Hoang Duong, Pradeep~Kumar Mishra, and Masaya Yasuda.
\newblock Efficient secure matrix multiplication over {LWE}-based homomorphic
  encryption.
\newblock {\em Tatra Mountains Mathematical Publications}, 67(1):69 -- 83,
  2016.

\bibitem{Zeev[11]}
Zeev Dvir, Parikshit Gopalan, and Sergey Yekhanin.
\newblock Matching vector codes.
\newblock {\em SIAM Journal on Computing}, 40(4):1154--1178, 2011.

\bibitem{Zeev[15]}
Zeev Dvir and Sivakanth Gopi.
\newblock 2-server pir with sub-polynomial communication.
\newblock In {\em STOC}, pages 577--584, 2015.

\bibitem{Klim[09]}
Klim Efremenko.
\newblock 3-query locally decodable codes of subexponential length.
\newblock In {\em STOC}, pages 39--44, 2009.

\bibitem{Erdos[61]}
P.~Erd\H{o}s, C.~Ko, and R.~Rado.
\newblock Intersection theorems for systems of finite sets.
\newblock {\em The Quarterly Journal of Mathematics}, 12(48):313--320, 1961.

\bibitem{Eynden80}
Charles~Vanden Eynden.
\newblock Proofs that $\sum 1/p$ diverges.
\newblock {\em The American Mathematical Monthly}, 87(5):394--397, 1980.

\bibitem{Fan[12]}
Junfeng Fan and Frederik Vercauteren.
\newblock Somewhat practical fully homomorphic encryption.
\newblock Cryptology ePrint Archive, Report 2012/144, 2012.
\newblock \url{https://eprint.iacr.org/2012/144}.

\bibitem{FangChen[20]}
Wenjing Fang, Chaochao Chen, Jin Tan, Chaofan Yu, Yufei Lu, Li~Wang, Lei Wang,
  Jun Zhou, and Alex X.
\newblock A hybrid-domain framework for secure gradient tree boosting, 2020.
\newblock \href {http://arxiv.org/abs/2005.08479} {\path{arXiv:2005.08479}}.

\bibitem{SebaTal[10]}
Sebastian Faust, Tal Rabin, Leonid Reyzin, Eran Tromer, and Vinod
  Vaikuntanathan.
\newblock Protecting circuits from leakage: the computationally-bounded and
  noisy cases.
\newblock In {\em EUROCRYPT}, pages 135--156, 2010.

\bibitem{Feld[88]}
P.~Feldman and S.~Micali.
\newblock An optimal algorithm for synchronous byzantine agreemet.
\newblock In {\em STOC}, pages 639--648, 1988.

\bibitem{Feld[87]}
Paul Feldman.
\newblock A practical scheme for non-interactive verifiable secret sharing.
\newblock In {\em FOCS}, pages 427--438, 1987.

\bibitem{Fer[19]}
Asaf Ferber and Matthew Kwan.
\newblock Almost all {S}teiner triple systems are almost resolvable.
\newblock {\em arXiv: preprint}, 2019.
\newblock \href {http://arxiv.org/abs/1907.06744} {\path{arXiv:1907.06744}}.

\bibitem{Fincke[85]}
U.~Fincke and M.~Pohst.
\newblock Improved methods for calculating vectors of short length in a
  lattice, including a complexity analysis.
\newblock {\em Mathematics of Computation}, 44(170):463--471, 1985.

\bibitem{MarcArno[21]}
Marc Fischlin and Arno Mittelbach.
\newblock An overview of the hybrid argument.
\newblock Cryptology ePrint Archive, Report 2021/088, 2021.
\newblock \url{https://eprint.iacr.org/2021/088}.

\bibitem{FALCON[20]}
Pierre-Alain Fouque, Jeffrey Hoffstein, Paul Kirchner, Vadim Lyubashevsky,
  Thomas Pornin, Thomas Prest, Thomas Ricosset, Gregor Seiler, William Whyte,
  and Zhenfei Zhang.
\newblock {FALCON}: Fast-fourier lattice-based compact signatures over {NTRU}.
\newblock URL: \url{https://falcon-sign.info/falcon.pdf}.

\bibitem{Frankl[16]}
Peter Frankl and Norihide Tokushige.
\newblock Invitation to intersection problems for finite sets.
\newblock {\em J. Combinatorial Theory Series A}, 144, 2016.

\bibitem{Gama[06]}
Nicolas Gama, Nick Howgrave-Graham, Henrik Koy, and Phong~Q. Nguyen.
\newblock Rankin's constant and blockwise lattice reduction.
\newblock In {\em CRYPTO}, pages 112--130, 2006.

\bibitem{NicMal[16]}
Nicolas Gama, Malika Izabach\`{e}ne, Phong~Q. Nguyen, and Xiang Xie.
\newblock Structural lattice reduction: Generalized worst-case to average-case
  reductions and homomorphic cryptosystems.
\newblock In {\em EUROCRYPT}, pages 528--558, 2016.

\bibitem{Gama[13]}
Nicolas Gama and Phong~Q. Nguyen.
\newblock Finding short lattice vectors within {M}ordell's inequality.
\newblock In {\em STOC}, pages 207--216, 2008.

\bibitem{Gama[10]}
Nicolas Gama, Phong~Q. Nguyen, and Oded Regev.
\newblock Lattice enumeration using extreme pruning.
\newblock In {\em EUROCRYPT}, pages 257--278, 2010.

\bibitem{Grg[13]}
Sanjam Garg, Craig Gentry, and Shai Halevi.
\newblock Candidate multilinear maps from ideal lattices.
\newblock In {\em EUROCRYPT}, pages 1--17, 2013.

\bibitem{Gasc[17]}
Adri\`{a} Gasc\'{o}n, Phillipp Schoppmann, Borja Balle, Mariana Raykova, Jack
  Doerner, Samee Zahur, and David Evans.
\newblock Privacy-preserving distributed linear regression on high-dimensional
  data.
\newblock In {\em Privacy Enhancing Technologies}, pages 345--364, 2017.

\bibitem{Genn[98]}
Rosario Gennaro, Michael~Oser Rabin, and Tal Rabin.
\newblock Simplified {VSS} and fast-track multiparty computations with
  applications to threshold cryptography.
\newblock In {\em {ACM} symposium on Principles of distributed computing
  (PODC)}, pages 101--111, 1998.

\bibitem{Gentry[15]}
Craig Gentry, Sergey Gorbunov, and Shai Halevi.
\newblock Graph-induced multilinear maps from lattices.
\newblock In {\em Theory of Cryptography}, pages 498--527, 2015.

\bibitem{Gen[08]}
Craig Gentry, Chris Peikert, and Vinod Vaikuntanathan~Craig Gentry.
\newblock Trapdoors for hard lattices and new cryptographic constructions.
\newblock In {\em STOC}, pages 197--206, 2008.

\bibitem{GolMic[87]}
O.~Goldreich, S.~Micali, and A.~Wigderson.
\newblock How to play {ANY} mental game.
\newblock In {\em STOC}, pages 218--229, 1987.

\bibitem{GoldLevin[89]}
Oded Goldreich and Leonid Levin.
\newblock A hard-core predicate for all one-way functions.
\newblock In {\em STOC}, pages 25--31, 1989.

\bibitem{Oded[91]}
Oded Goldreich, Silvio~M Micali, and Avi Wigderson.
\newblock Proofs that yield nothing but their validity or all languages in np
  have zero-knowledge proof systems.
\newblock {\em Journal of the ACM}, 38(3):690--728, 1991.

\bibitem{Gold[13]}
Shafi Goldwasser, Yael Kalai, Raluca~Ada Popa, Vinod Vaikuntanathan, and
  Nickolai Zeldovich.
\newblock Reusable garbled circuits and succinct functional encryption.
\newblock In {\em STOC}, pages 555--564, 2013.

\bibitem{Gold[82]}
Shafi Goldwasser and Silvio~M Micali.
\newblock Probabilistic encryption \& how to play mental poker keeping secret
  all partial information.
\newblock In {\em STOC}, pages 365--377, 1982.

\bibitem{GoldMic[84]}
Shafi Goldwasser and Silvio~M Micali.
\newblock Probabilistic encryption.
\newblock {\em J. of Computer and System Sciences}, 28(2):270--299, 1984.

\bibitem{VipAshK[18]}
Vipul Goyal and Ashutosh Kumar.
\newblock Non-malleable secret sharing.
\newblock In {\em STOC}, pages 685--698, 2018.

\bibitem{Goyal[06]}
Vipul Goyal, Omkant Pandey, Amit Sahai, and Brent Waters.
\newblock Attribute-based encryption for fine-grained access control of
  encrypted data.
\newblock In {\em {ACM} {CCS}}, pages 89--98, 2006.

\bibitem{Gratzer[03]}
George Gr\"{a}tzer.
\newblock {\em General Lattice Theory (Second Edition)}.
\newblock Birkh\"{a}user Basel, 2003.

\bibitem{Gratzer[09]}
George Gr\"{a}tzer.
\newblock {\em Lattice Theory: First Concepts and Distributive Lattices}.
\newblock Dover Books on Mathematics. Dover Publications, 2009.

\bibitem{Gro[00]}
Vince Grolmusz.
\newblock Superpolynomial size set-systems with restricted intersections mod 6
  and explicit ramsey graphs.
\newblock {\em Combinatorica}, 20:71--86, 2000.

\bibitem{Charles[20]}
Charles Grover, Cong Ling, and Roope Vehkalahti.
\newblock Non-commutative ring learning with errors from cyclic algebras, 2020.
\newblock \href {http://arxiv.org/abs/2008.01834} {\path{arXiv:2008.01834}}.

\bibitem{Mida[03]}
Mida Guillermoand, Keith~M. Martin, and Christine~M. O'Keefe.
\newblock Providing anonymity in unconditionally secure secret sharing schemes.
\newblock {\em Designs, Codes and Cryptography}, 28:227--245, 2003.

\bibitem{AdnanEsam[19]}
Adnan Gutub, Nouf Al-Juaid, and Esam Khan.
\newblock Counting-based secret sharing technique for multimedia applications.
\newblock {\em Multimedia Tools and Applications}, 78:5591–5619, 2019.

\bibitem{Hal[17]}
Shai Halevi, Tzipora Halevi, and Victor Shoup.
\newblock Implementing {BP}-obfuscation using graph-induced encoding.
\newblock In {\em CCS}, pages 783--798, 2017.

\bibitem{HalTea[04]}
Joseph Halpern and Vanessa Teague.
\newblock Rational secret sharing and multiparty computation: extended
  abstract.
\newblock In {\em STOC}, pages 623--632, 2004.

\bibitem{Hardy[80]}
G.~H. Hardy and E.~M. Wright.
\newblock {\em An Introduction to the Theory of Numbers, 5th Edition}.
\newblock Oxford University Press, 1980.

\bibitem{Mic[98]}
Michael Harkavy, J.~Doug Tygar, and Hiroaki Kikuchi.
\newblock Electronic auctions with private bids.
\newblock In {\em Proceedings of the 3rd conference on USENIX Workshop on
  Electronic Commerce}, page~6, 1998.

\bibitem{MasaTake[18]}
Masahito Hayashi and Takeshi Koshiba.
\newblock Universal construction of cheater-identifiable secret sharing against
  rushing cheaters based on message authentication.
\newblock In {\em IEEE International Symposium on Information Theory (ISIT)},
  pages 2614--2618, 2018.

\bibitem{Osama[12]}
Osama Hayatle, Amr Youssef, and Hadi Otrok.
\newblock Dempster-{S}hafer evidence combining for (anti)-honeypot
  technologies.
\newblock {\em Information Security Journal: A Global Perspective},
  21(6):306--316, 2012.

\bibitem{Amir[95]}
Amir Herzberg, Stanislaw Jarecki, Hugo Krawczyk, and Moti Yung.
\newblock Proactive secret sharing or: How to cope with perpetual leakage.
\newblock In {\em CRYPTO}, pages 339--352, 1995.

\bibitem{MarkBui[99]}
Mark Hillery, Vladim\'{i}r Bu\u{z}ek, and Andr\'e Berthiaume.
\newblock Quantum secret sharing.
\newblock {\em Phys. Rev. A}, 59:1829--1834, 1999.

\bibitem{HirtMau[97]}
Martin Hirt and Ueli Maurer.
\newblock Complete characterization of adversaries tolerable in secure
  multi-party computation (extended abstract).
\newblock In {\em PODC}, pages 25--34, 1997.

\bibitem{Hof[12]}
Dennis Hofheinz.
\newblock All-but-many lossy trapdoor functions.
\newblock In {\em EUROCRYPT}, pages 209--227, 2012.

\bibitem{Holz[05]}
Thorsten Holz and Frederic Raynal.
\newblock Detecting honeypots and other suspicious environments.
\newblock In {\em IEEE SMC Information Assurance Workshop}, pages 29--36, 2005.

\bibitem{HasRus[99]}
Johan H\r{a}stad, Russell Impagliazzo, Leonid~A. Levin, and Michael Luby.
\newblock A pseudorandom generator from any one-way function.
\newblock {\em SIAM J. on Computing}, 28(4):1364--1396, 1999.

\bibitem{HsiHsu[07]}
Shang-Lin Hsieh, Lung-Yao Hsu, and I-Ju Tsai.
\newblock A copyright protection scheme for color images using secret sharing
  and wavelet transform.
\newblock {\em International Journal of Computer and Information Engineering},
  1(10):3172--3178, 2007.

\bibitem{HuLi[18]}
Chunqiang Hu, Wei Li, Xiuzhen Cheng, Jiguo Yu, Shengling Wang, and Rongfang
  Bie.
\newblock A secure and verifiable access control scheme for big data storage in
  clouds.
\newblock {\em IEEE Transactions on Big Data}, 4(3):341--355, 2018.

\bibitem{Sorin[07]}
Sorin Iftene.
\newblock General secret sharing based on the chinese remainder theorem with
  applications in e-voting.
\newblock {\em Electronic Notes in Theoretical Computer Science (ENTCS)},
  186:67--84, 2007.

\bibitem{RusNisan[94]}
Russell Impagliazzo, Noam Nisan, and Avi Wigderson.
\newblock Pseudorandomness fornetwork algorithms.
\newblock In {\em STOC}, pages 356--364, 1994.

\bibitem{IshRafa[12]}
Yuval Ishai, Rafail Ostrovsky, and Hakan Seyalioglu.
\newblock Identifying cheaters without an honest majority.
\newblock In {\em TCC}, pages 21--38, 2012.

\bibitem{YuvAmi[03]}
Yuval Ishai, Amit Sahai, and David Wagner.
\newblock Private circuits: Securing hardware against probing attacks.
\newblock In {\em CRYPTO}, pages 463--481, 2003.

\bibitem{Ito[87]}
Mitsuru Ito, Akira Saito, and Takao Nishizeki.
\newblock Secret sharing scheme realizing general access structure.
\newblock In {\em Globecom}, pages 99--102, 1987.

\bibitem{JoySab[20]}
Dintomon Joy, M.~Sabir, Bikash~K. Behera, and Prasanta~K. Panigrahi.
\newblock Implementation of quantum secret sharing and quantum binary voting
  protocol in the {IBM} quantum computer.
\newblock {\em Quantum Information Processing}, 19, 2020.

\bibitem{Karch[93]}
M.~Karchmer and A.~Wigderson.
\newblock On span programs.
\newblock In {\em Structure in Complexity Theory Conference}, pages 102--111,
  1993.

\bibitem{KarninGreene[83]}
E.~Karnin, J.~Greene, and M.~Hellman.
\newblock On secret sharing systems.
\newblock {\em IEEE Transactions on Information Theory}, 29(1):35--41, 1983.

\bibitem{Kate[10]}
Aniket Kate, Gregory~M. Zaverucha, and Ian Goldberg.
\newblock Constant-size commitments to polynomials and their applications.
\newblock In {\em ASIACRYPT}, pages 177--194, 2010.

\bibitem{Katz[06]}
Jonathan Katz and Chiu-Yuen Koo.
\newblock On expected constant-round protocols for byzantine agreement.
\newblock In {\em CRYPTO}, pages 445--462, 2006.

\bibitem{KatzVadim[21]}
Jonathan Katz and Vadim Lyubashevsky.
\newblock {\em Lattice-Based Cryptography}.
\newblock Cryptography and Network Security. CRC Press, 2021.

\bibitem{Katz[09]}
Jonathan Katz and Vinod Vaikuntanathan.
\newblock Smooth projective hashing and password-based authenticated key
  exchange from lattices.
\newblock In {\em ASIACRYPT}, pages 636--652, 2009.

\bibitem{Yung[04]}
Aggelos Kiayias and Moti Yung.
\newblock The vector-ballot e-voting approach.
\newblock In {\em International Conference on Financial Cryptography}, pages
  72--89, 2004.

\bibitem{KimSam[19]}
Sam Kim.
\newblock Multi-authority attribute-based encryption from {LWE} in the {OT}
  model.
\newblock Cryptology ePrint Archive, Report 2019/280, 2019.
\newblock \url{https://eprint.iacr.org/2019/280}.

\bibitem{KimWu[17]}
Sam Kim and David~J. Wu.
\newblock Watermarking cryptographic functionalities from standard lattice
  assumptions.
\newblock In {\em CRYPTO}, pages 503--536, 2017.

\bibitem{KimWu[19]}
Sam Kim and David~J. Wu.
\newblock Watermarking {PRF}s from lattices: Stronger security via extractable
  {PRF}s.
\newblock In {\em CRYPTO}, pages 335--366, 2019.

\bibitem{KimKiran[19]}
Yongjune Kim, Ravi~Kiran Raman, Young-Sik Kim, Lav~R. Varshney, and Naresh~R.
  Shanbhag.
\newblock Efficient local secret sharing for distributed blockchain systems.
\newblock {\em IEEE Communications Letters}, 23(2):282--285, 2019.

\bibitem{Kishi[02]}
Wataru Kishimoto, Koji Okada, Kaoru Kurosawa, and Wakaha Ogata.
\newblock On the bound for anonymous secret sharing schemes.
\newblock {\em Discrete Applied Mathematics}, 121:193--202, 2002.

\bibitem{Knopp[56]}
Konrad Knopp.
\newblock {\em Infinite Sequences and Series}.
\newblock Dover Publications, 1956.

\bibitem{KomaYog[14]}
Ilan Komargodski, Moni Naor, and Eylon Yogev.
\newblock Secret-sharing for {NP}.
\newblock In {\em ASIACRYPT}, pages 254--273, 2014.

\bibitem{KomaMoni[16]}
Ilan Komargodski, Moni Naor, and Eylon Yogev.
\newblock How to share a secret, infinitely.
\newblock In {\em TCC}, pages 485--514, 2016.

\bibitem{IlanMark[18]}
Ilan Komargodski and Mark Zhandry.
\newblock Cutting-edge cryptography through the lens of secret sharing.
\newblock {\em Information and Computation}, 263:75--96, 2018.

\bibitem{Kra[95]}
E.S. Kramer and R.~Mathon.
\newblock Proper {S}(t,$\mathcal{K}, v)$'s for $t \geq 3, v \leq 16,
  |\mathcal{K}| > 1$ and their extensions.
\newblock {\em J. Combin. Des.}, pages 411--425, 1995.

\bibitem{Hugo[93]}
Hugo Krawczyk.
\newblock Secret sharing made short.
\newblock In {\em CRYPTO}, pages 136--146, 1993.

\bibitem{Kraw[04]}
Neal {Krawetz}.
\newblock Anti-honeypot technology.
\newblock {\em IEEE Security \& Privacy}, 2(1):76--79, 2004.

\bibitem{Kulkarni[12]}
Saurabh Kulkarni, Madhumitra Mutalik, Prathamesh Kulkarni, and Tarun Gupta.
\newblock Honeydoop - a system for on-demand virtual high interaction
  honeypots.
\newblock In {\em International Conference for Internet Technology and Secured
  Transactions}, pages 743--747, 2012.

\bibitem{KaoSat[95]}
Kaoru Kurosawa, Satoshi Obana, and Wakaha Ogata.
\newblock $t$-cheater identifiable $(k, n)$ threshold secret sharing schemes.
\newblock In {\em CRYPTO}, pages 410--423, 1995.

\bibitem{Kwan[20]}
Matthew Kwan.
\newblock Almost all {S}teiner triple systems have perfect matchings.
\newblock {\em Proceedings of the London Mathematical Society},
  121(6):1468--1495, 2020.

\bibitem{ChiLein[89]}
Chi-Sung Laih, Lein Harn, Jau-Yien Lee, and Tzonelih Hwang.
\newblock Dynamic threshold scheme based on the definition of cross-product in
  an n-dimensional linear space.
\newblock In {\em CRYPTO}, pages 286--292, 1989.

\bibitem{LambertOrg[58]}
Johann~Heinrich Lambert.
\newblock Observationes variae in mathesin puram.
\newblock {\em Acta Helv. Phys. Math. Anat. Bot. Med.}, 3(5), 1758.

\bibitem{Lamport[82]}
Leslie Lamport, Robert Shostak, and Marshall Pease.
\newblock The byzantine generals problem.
\newblock {\em ACM Transactions on Programming Languages and Systems (TOPLAS)},
  4(3):382--401, 1982.

\bibitem{Ade[15]}
Adeline Langlois and Damien Stehl\'{e}.
\newblock Worst-case to average-case reductions for module lattices.
\newblock {\em Designs, Codes and Cryptography}, 75(3):565--599, 2015.

\bibitem{LLL[82]}
A.~K. Lenstra, H.~W.~Lenstra Jr., and L.~Lov\'{a}sz.
\newblock Factoring polynomials with rational coefficients.
\newblock {\em Mathematische Annalen}, 261:515--534, 1982.

\bibitem{QinLiu[21]}
Qin Liao, Haijie Liu, Lingjin Zhu, and Ying Guo.
\newblock Quantum secret sharing using discretely modulated coherent states.
\newblock {\em Phys. Rev. A}, 103:032410, 2021.

\bibitem{Huijia[16]}
Huijia Lin.
\newblock Indistinguishability obfuscation from constant-degree graded encoding
  schemes.
\newblock In {\em EUROCRYPT}, pages 28--57, 2016.

\bibitem{Liu[68]}
Chung~Laung Liu.
\newblock {\em Introduction to Combinatorial Mathematics}.
\newblock Computer science series. McGraw-Hill, USA, 1968.

\bibitem{LiuStoc[18]}
Tianren Liu and Vinod Vaikuntanathan.
\newblock Breaking the circuit-size barrier in secret sharing.
\newblock In {\em ACM SIGACT Symposium on Theory of Computing (STOC)}, pages
  699--708, 2018.

\bibitem{Liu[17]}
Tianren Liu, Vinod Vaikuntanathan, and Hoeteck Wee.
\newblock Conditional disclosure of secrets via non-linear reconstruction.
\newblock In {\em CRYPTO}, pages 758--790, 2017.

\bibitem{Liu[18]}
Tianren Liu, Vinod Vaikuntanathan, and Hoeteck Wee.
\newblock Towards breaking the exponential barrier for general secret sharing.
\newblock In {\em EUROCRYPT}, pages 567--596, 2018.

\bibitem{Adriana[12]}
Adriana L\'{o}pez-Alt, Eran Tromer, and Vinod Vaikuntanathan.
\newblock On-the-fly multiparty computation on the cloud via multikey fully
  homomorphic encryption.
\newblock In {\em STOC}, pages 1219--1234, 2012.

\bibitem{LuHou[20]}
Changbin Lu, Fuyou Miao, Junpeng Hou, Wenchao Huang, and Yan Xiong.
\newblock A verifiable framework of entanglement-free quantum secret sharing
  with information-theoretical security.
\newblock {\em Quantum Information Processing}, 19, 2020.

\bibitem{LuMiao[18]}
Changbin Lu, Fuyou Miao, Junpeng Hou, and Keju Meng.
\newblock Verifiable threshold quantum secret sharing with sequential
  communication.
\newblock {\em Quantum Information Processing}, 17, 2018.

\bibitem{Lu[18]}
Xianhui Lu, Yamin Liu, Zhenfei Zhang, Dingding Jia, Haiyang Xue, Jingnan He,
  Bao Li, and Kunpeng Wang.
\newblock Lac: Practical ring-lwe based public-key encryption with byte-level
  modulus.
\newblock Cryptology ePrint Archive, Report 2018/1009, 2018.
\newblock \url{https://eprint.iacr.org/2018/1009}.

\bibitem{Vad[09]}
Vadim Lyubashevsky.
\newblock Fiat-{S}hamir with aborts: Applications to lattice and
  factoring-based signatures.
\newblock In {\em ASIACRYPT}, pages 598--616, 2009.

\bibitem{Vad[12]}
Vadim Lyubashevsky.
\newblock Lattice signatures without trapdoors.
\newblock In {\em EUROCRYPT}, pages 738--755, 2012.

\bibitem{Reg[10]}
Vadim Lyubashevsky, Chris Peikert, and Oded Regev.
\newblock On ideal lattices and learning with errors over rings.
\newblock In {\em EUROCRYPT}, pages 1--23, 2010.

\bibitem{Lyub[15]}
Vadim Lyubashevsky and Daniel Wichs.
\newblock Simple lattice trapdoor sampling from a broad class of distributions.
\newblock In {\em Public-Key Cryptography (PKC)}, pages 716--730, 2015.

\bibitem{Urmila[18]}
Urmila Mahadev.
\newblock Classical verification of quantum computations.
\newblock In {\em IEEE Symposium on Foundations of Computer Science (FOCS)},
  pages 259--267, 2018.

\bibitem{MartPate[11]}
Keith~M. Martin, Maura~B. Paterson, and Douglas~R. Stinson.
\newblock Error decodable secret sharing and one-round perfectly secure message
  transmission for general adversary structures.
\newblock {\em Cryptography and Communications}, 3:65--86, 2011.

\bibitem{McSar[81]}
R.~J. McEliece and D.~V. Sarwate.
\newblock On sharing secrets and reed-solomon codes.
\newblock {\em Commun. ACM}, 24(9):583–584, 1981.

\bibitem{McEliece[81]}
R.~J. McEliece and Dilip~V Sarwate.
\newblock On sharing secrets and {R}eed-{S}olomon codes.
\newblock {\em Communications of the ACM}, 24(9):583--584, 1981.

\bibitem{SihAhm[20]}
Sihem Mesnager, Ahmet Sınak, and O\"{g}uz Yayla.
\newblock Threshold-based post-quantum secure verifiable multi-secret sharing
  for distributed storage blockchain.
\newblock {\em Mathematics}, 8(12), 2020.

\bibitem{DanMicci[09]}
Daniele Micciancio.
\newblock Cryptographic functions from worst-case complexity assumptions.
\newblock In {\em The LLL Algorithm}, Information Security and Cryptography,
  pages 427--452. Springer, Berlin, Heidelberg, 2009.

\bibitem{MicciGold[02]}
Daniele Micciancio and Shafi Goldwasser.
\newblock {\em Complexity of Lattice Problems - A Cryptographic Perspective},
  volume 671 of {\em The Springer International Series in Engineering and
  Computer Science}.
\newblock Springer US, 2002.

\bibitem{Shafi[02]}
Daniele Micciancio and Shafi Goldwasser.
\newblock {\em Complexity of Lattice Problems: A Cryptographic Perspective}.
\newblock The Springer International Series in Engineering and Computer Science
  (671). Springer US, 2002.

\bibitem{Micci[12]}
Daniele Micciancio and Chris Peikert.
\newblock Trapdoors for lattices: Simpler, tighter, faster, smaller.
\newblock In {\em EUROCRYPT}, pages 700--718, 2012.

\bibitem{Micci[13]}
Daniele Micciancio and Chris Peikert.
\newblock Hardness of {SIS} and {LWE} with small parameters.
\newblock In {\em CRYPTO}, pages 21--39, 2013.

\bibitem{DaniPan[10]}
Daniele Micciancio and Panagiotis Voulgaris.
\newblock A deterministic single exponential time algorithm for most lattice
  problems based on voronoi cell computation.
\newblock In {\em STOC}, pages 351--358, 2010.

\bibitem{DanPan[10]}
Daniele Micciancio and Panagiotis Voulgaris.
\newblock Faster exponential time algorithms for the shortest vector problem.
\newblock In {\em ACM-SIAM Symposium on Discrete Algorithms (SODA)}, pages
  1468--1480, 2010.

\bibitem{ArnoMarc[21]}
Arno Mittelbach and Marc Fischlin.
\newblock {\em The Theory of Hash Functions and Random Oracles-An Approach to
  Modern Cryptography}.
\newblock Information Security and Cryptography. Springer International
  Publishing, 2021.

\bibitem{PayPet[18]}
Payman Mohassel and Peter Rindal.
\newblock {ABY} 3: a mixed protocol framework for machine learning.
\newblock In {\em ACM CCS}, pages 35--52, 2018.

\bibitem{PayYup[17]}
Payman Mohassel and Yupeng Zhang.
\newblock Secure{ML}: A system for scalable privacy-preserving machine
  learning.
\newblock In {\em IEEE Symposium on Security and Privacy (S\&P)}, pages 19--38,
  2017.

\bibitem{Pat[17]}
Patrick Morris.
\newblock Random {S}teiner triple systems.
\newblock Master's thesis, Freie Universit\"{a}t Berlin, 2017.

\bibitem{MoniNaor[06]}
Moni Naor.
\newblock Secret sharing for access structures beyond {P}.
\newblock
  \url[Slides]{http://www.wisdom.weizmann.ac.il/~naor/PAPERS/minicrypt.html},
  2006.

\bibitem{Naor[06]}
Moni Naor and Avishai Wool.
\newblock Access control and signatures via quorum secret sharing.
\newblock In {\em {ACM} conference on Computer and communications security
  {(CCS)}}, pages 157--168, 1996.

\bibitem{Eyal[17]}
Eyal Neemany.
\newblock Honeypot buster: A unique red-team tool, 2017.
\newblock Javelin Networks.
\newblock URL:
  \url{https://jblog.javelin-networks.com/blog/the-honeypot-buster/}.

\bibitem{Hamid[19]}
Hamid Nejatollahi, Nikil Dutt, Sandip Ray, Francesco Regazzoni, Indranil
  Banerjee, and Rosario Cammarota.
\newblock Post-quantum lattice-based cryptography implementations: A survey.
\newblock {\em ACM Comput. Surv.}, 51(6), 2019.

\bibitem{Nguyen[09]}
Phong~Q. Nguyen and Damien Stehl\'{e}.
\newblock An {LLL} algorithm with quadratic complexity.
\newblock {\em SIAM Journal on Computing}, 39(3):874--903, 2009.

\bibitem{PhongStern[01]}
Phong~Q. Nguyen and Jacques Stern.
\newblock The two faces of lattices in cryptology.
\newblock In {\em International Cryptography and Lattices Conference}, pages
  146--180, 2001.

\bibitem{Ngu[10]}
Phong~Q. Nguyen and Brigitte Vall\'{e}e.
\newblock {\em The {LLL} Algorithm: Survey and Applications}.
\newblock Information Security and Cryptography. Springer US, 2010.

\bibitem{Vid[08]}
Phong~Q. Nguyen and T.~Vidick.
\newblock Sieve algorithms for the shortest vector problem are practical.
\newblock {\em Journal of Mathematical Cryptology}, 2(2):181--207, 2008.

\bibitem{Naom[91]}
Noam Nisan.
\newblock Pseudorandom bits for constant depth circuits.
\newblock {\em Combinatorica}, 11(1):63--70, 1991.

\bibitem{Naom[92]}
Noam Nisan.
\newblock Pseudorandom generators for space-bounded computation.
\newblock {\em Combinatorica}, 12(4):449--461, 1992.

\bibitem{NaomAvi[92]}
Noam Nisan and Avi Wigderson.
\newblock Hardness vs randomness.
\newblock {\em J. of Computer and System Sciences}, 49(2):149--167, 1994.

\bibitem{MehrDoug[12]}
Mehrdad Nojoumian and Douglas~R. Stinson.
\newblock Social secret sharing in cloud computing using a new trust function.
\newblock In {\em International Conference on Privacy, Security and Trust},
  pages 161--167, 2012.

\bibitem{SatoOb[11]}
Satoshi Obana.
\newblock Almost optimum $t$-cheater identifiable secret sharing schemes.
\newblock In {\em EUROCRYPT}, pages 284--302, 2011.

\bibitem{SatTos[06]}
Satoshi Obana and Toshinori Araki.
\newblock Almost optimum secret sharing schemes secure against cheating for
  arbitrary secret distribution.
\newblock In {\em ASIACRYPT}, pages 364--379, 2006.

\bibitem{Pat[08]}
Patric~R.J. \"{O}sterg\r{a}rd and Olli Pottonen.
\newblock There exists no {S}teiner system {S}(4,5,17).
\newblock {\em Journal of Combinatorial Theory, Series A}, 115(8):1570--1573,
  2008.

\bibitem{Patra[14]}
Arpita Patra, Ashish Choudhury, and C.~Pandu Rangan.
\newblock Asynchronous byzantine agreement with optimal resilience.
\newblock {\em Distributed Computing}, 27:111--146, 2014.

\bibitem{Tor[01]}
Torben~Pryds Pedersen.
\newblock Non-interactive and information-theoretic secure verifiable secret
  sharing.
\newblock In {\em CRYPTO}, pages 129--140, 2001.

\bibitem{Pei[09]}
Chris Peiker{t}.
\newblock Public-key cryptosystems from the worst-case shortest vector problem.
\newblock {\em STOC}, pages 333--342, 2009.

\bibitem{Peikert[16]}
Chris Peikert.
\newblock A decade of lattice cryptography.
\newblock {\em Foundations and Trends in Theoretical Computer}, 10:283--424,
  2016.

\bibitem{Pei[06]}
Chris Peikert and Alon Rosen.
\newblock Efficient collision-resistant hashing from worst-case assumptions on
  cyclic lattices.
\newblock In {\em TCC}, pages 145--166, 2006.

\bibitem{Sina[19]}
Chris Peikert and Sina Shiehian.
\newblock Noninteractive zero knowledge for {NP} from (plain) learning with
  errors.
\newblock In {\em CRYPTO}, pages 89--114, 2019.

\bibitem{Pei[08]}
Chris Peikert, Vinod Vaikuntanathan, and Brent Waters.
\newblock A framework for efficient and composable oblivious transfer.
\newblock In {\em CRYPTO}, pages 554--571, 2008.

\bibitem{PeiW[08]}
Chris Peikert and Brent Waters.
\newblock Lossy trapdoor functions and their applications.
\newblock In {\em STOC}, pages 187--196, 2008.

\bibitem{Phillips[92]}
Steven~J. Phillips and Nicholas~C. Phillips.
\newblock Strongly ideal secret sharing schemes.
\newblock {\em Journal of Cryptology}, 5:185--191, Oct. 1992.

\bibitem{Pila[17]}
Hossein Pilaram and Taraneh Eghlidos.
\newblock An efficient lattice based multi-stage secret sharing scheme.
\newblock {\em IEEE Transactions on Dependable and Secure Computing},
  14(1):2--8, 2017.

\bibitem{Pipp[89]}
Nicholar Pippenger and Joel Spencer.
\newblock Asymptotic behavior of the chromatic index for hypergraphs.
\newblock {\em J. Combin.Theory Ser. A}, 51(1):24--42, 1989.

\bibitem{Poha[81]}
Michael Pohst.
\newblock On the computation of lattice vectors of minimal length, successive
  minima and reduced bases with applications.
\newblock {\em ACM SIGSAM Bulletin}, 15(1):37--44, 1981.

\bibitem{Deng[07]}
Ying pu~Deng, Li~feng Guo, and Mu~lan Liu.
\newblock Constructions for anonymous secret sharing schemes using
  combinatorial designs.
\newblock {\em Acta Mathematicae Applicatae Sinica}, 23:67--78, January 2007.

\bibitem{QinWall[18]}
Huawang Qin, Wallace K.~S. Tang, and Raylin Tso.
\newblock Rational quantum secret sharing.
\newblock {\em Scientific Reports}, 8, 2018.

\bibitem{Quach[20]}
Willy Quach.
\newblock {UC}-secure {OT} from {LWE}, revisited.
\newblock In {\em SCN}, pages 192--211, 2020.

\bibitem{Qua[18]}
Willy Quach, Daniel Wichs, and Giorgos Zirdelis.
\newblock Watermarking {PRF}s under standard assumptions: Public marking and
  security with extraction queries.
\newblock In {\em TCC}, pages 669--698, 2018.

\bibitem{TalR[06]}
Tal Rabin.
\newblock A simplified approach to threshold and proactive {RSA}.
\newblock In {\em CRYPTO}, pages 89--104, 2006.

\bibitem{Rabin[89]}
Tal Rabin and Michael Ben-Or.
\newblock Verifiable secret sharing and multiparty protocols with honest
  majority (extended abstract).
\newblock In {\em STOC}, pages 73--85, 1989.

\bibitem{AMDD[2020]}
A.~M. Raigorodskii and D.~D. Cherkashin.
\newblock Extremal problems in hypergraph colourings.
\newblock {\em Russ. Math. Surv.}, 75(1):89--146, 2020.

\bibitem{RavLavR[18]}
Ravi~Kiran Raman and Lav~R. Varshney.
\newblock Distributed storage meets secret sharing on the blockchain.
\newblock In {\em Information Theory and Applications Workshop (ITA)}, pages
  1--6, 2018.

\bibitem{RaviLav[18]}
Ravi~Kiran Raman and Lav~R. Varshney.
\newblock Dynamic distributed storage for blockchains.
\newblock In {\em IEEE International Symposium on Information Theory (ISIT)},
  pages 2619--2623, 2018.

\bibitem{Chau[71]}
D.~K. Ray-Chaudhuri and R.~M. Wilson.
\newblock Solution of {K}irkman's schoolgirl problem.
\newblock {\em Combinatorics (Proc. Sympos. Pure Math., Vol. XIX, Univ.
  California, Los Angeles, Calif., 1968), Amer. Math.Soc., Providence, R.I.},
  pages 187--203, 1971.

\bibitem{Reg[05]}
Oded Regev.
\newblock On lattices, learning with errors, random linear codes, and
  cryptography.
\newblock In {\em 37th annual ACM Symposium on Theory of Computing (STOC)},
  pages 84--93, 2005.

\bibitem{RonAdi[01]}
Ron Rivest, Adi Shamir, and Yael~Tauman Kalai.
\newblock How to leak a secret.
\newblock In {\em ASIACRYPT}, pages 552--565, 2001.

\bibitem{Miruna[17]}
Miruna Ro\c{s}ca, Amin Sakzad, Damien Stehl\'{e}, and Ron Steinfeld.
\newblock Middle-product learning with errors.
\newblock In {\em CRYPTO}, pages 283--297, 2017.

\bibitem{RothG[12]}
Guy~N. Rothblum.
\newblock How to compute under $\mathcal{AC}^\mathbf{0}$ leakage without secure
  hardware.
\newblock In {\em CRYPTO}, pages 552--569, 2012.

\bibitem{Neil[06]}
Neil~C. Rowe.
\newblock Measuring the effectiveness of honeypot counter-counterdeception.
\newblock In {\em Hawaii International Conference on System Sciences
  (HICSS'06)}, pages 129c--129c, 2006.

\bibitem{Rowe[06]}
Neil~C. Rowe, Binh~T. Duong, and E.~John Custy.
\newblock Fake honeypots: A defensive tactic for cyberspace.
\newblock In {\em IEEE Information Assurance Workshop}, pages 223--230, 2006.

\bibitem{Markus[10]}
Markus R\"{u}ckert.
\newblock Lattice-based blind signatures.
\newblock In {\em ASIACRYPT}, pages 413--430, 2010.

\bibitem{Santis[94]}
Alfredo~De Santis, Yvo Desmedt, Yair Frankel, and Moti Yung.
\newblock How to share a function securely.
\newblock In {\em STOC}, pages 522--533, 1994.

\bibitem{Schnorr[87]}
C.-P. Schnorr.
\newblock A hierarchy of polynomial lattice basis reduction algorithms.
\newblock {\em Theoretical Computer Science}, 53(2-3):201--224, 1987.

\bibitem{Schnorr[94]}
C.-P. Schnorr and M.~Euchner.
\newblock Lattice basis reduction: Improved practical algorithms and solving
  subset sum problems.
\newblock {\em Mathematical Programming}, 66:181--199, 1994.

\bibitem{Schnorr[95]}
C.~P. Schnorr and H.~H. H\"{o}rner.
\newblock Attacking the {C}hor-{R}ivest cryptosystem by improved lattice
  reduction.
\newblock In {\em EUROCRYPT}, pages 1--12, 1995.

\bibitem{Berry[99]}
Berry Schoenmakers.
\newblock A simple publicly verifiable secret sharing scheme and its
  application to electronic voting.
\newblock In {\em CRYPTO}, pages 148--164, 1999.

\bibitem{VipinThesis[19]}
Vipin~Singh Sehrawat.
\newblock {\em Privacy Enhancing Cryptographic Constructs for Cloud and
  Distributed Security}.
\newblock PhD thesis, The University of Texas at Dallas, 2019.

\bibitem{Vipin[19]}
Vipin~Singh Sehrawat and Yvo Desmedt.
\newblock Bi-{H}omomorphic {L}attice-{B}ased {PRF}s and {U}nidirectional
  {U}pdatable {E}ncryption.
\newblock In {\em CANS}, volume 11829, pages 3--23. LNCS, Springer, 2019.

\bibitem{Vipin[20]}
Vipin~Singh Sehrawat and Yvo Desmedt.
\newblock Access structure hiding secret sharing from novel set systems and
  vector families.
\newblock In {\em COCOON}, volume 12273, pages 246--261. LNCS, Springer, 2020.
\newblock {F}ull version:.
\newblock \href {http://arxiv.org/abs/2008.07969} {\path{arXiv:2008.07969}}.

\bibitem{Sehrawat[17]}
Vipin~Singh Sehrawat, Yogendra Shah, Vinod~Kumar Choyi, Alec Brusilovsky, and
  Samir Ferdi.
\newblock Certificate and signature free anonymity for {V2V} communications.
\newblock In {\em IEEE Vehicular Networking Conference (VNC)}, pages 139--146,
  2017.

\bibitem{Shamir[79]}
Adi Shamir.
\newblock How to share a secret.
\newblock {\em Commun. ACM}, 22:612--613, 1979.

\bibitem{Shankar[08]}
Bhavani Shankar, Kannan Srinathan, and Chandrasekaran~Pandu Rangan.
\newblock Alternative protocols for generalized oblivious transfer.
\newblock In {\em International Conference on Distributed Computing and
  Networking {(ICDCN)}}, pages 304--309, 2008.

\bibitem{ShiChao[16]}
Haoyi Shi, Chao Jiang, Wenrui Dai, Xiaoqian Jiang, Yuzhe Tang, Lucila
  OhnoMachado, and Shuang Wang.
\newblock Secure multi-party computation grid logistic regression
  ({SMAC-GLORE}).
\newblock {\em BMC medical informatics and decision making}, 16, 2016.

\bibitem{Gusta[88]}
Gustavus~J. Simmons.
\newblock How to (really) share a secret.
\newblock In {\em CRYPTO}, pages 390--448, 1988.

\bibitem{GustSimm[89]}
Gustavus~J. Simmons.
\newblock Prepositioned shared secret and/or shared control schemes.
\newblock In {\em EUROCRYPT}, pages 436--467, 1989.

\bibitem{Sperner[28]}
Emanuel Sperner.
\newblock {E}in {S}atz \"{u}ber {U}ntermengen einer endlichen {M}enge.
\newblock {\em Mathematische Zeitschrift}, 27:544--548, 1928.

\bibitem{Lance[03]}
Lance Spitzner.
\newblock {\em Honeypots: Tracking Hackers}.
\newblock Addison-Wesley, 2003.

\bibitem{Stadler[96]}
Markus Stadler.
\newblock Publicly verifiable secret sharing.
\newblock In {\em EUROCRYPT}, pages 190--199, 1996.

\bibitem{Damien[09]}
Damien Stehl\'{e}, Ron Steinfeld, Keisuke Tanaka, and Keita Xagawa.
\newblock Efficient public key encryption based on ideal lattices (extended
  abstract).
\newblock In {\em ASIACRYPT}, pages 617--635, 2009.

\bibitem{Steiner[53]}
Jakob Steiner.
\newblock Combinatorische {A}ufgaben.
\newblock {\em J. Reine Angew. Math.}, 45, 1853.

\bibitem{Stein[07]}
Ron Steinfeld, Josef Pieprzyk, and Huaxiong Wang.
\newblock Lattice-based threshold changeability for standard shamir
  secret-sharing schemes.
\newblock {\em IEEE Transactions on Information Theory}, 53(7):2542--2559,
  2007.

\bibitem{Stinson[87]}
D.~R. Stinson and S.~A. Vanstone.
\newblock A combinatorial approach to threshold schemes.
\newblock In {\em CRYPTO}, pages 330--339, 1987.

\bibitem{Stoll[89]}
Clifford Stoll.
\newblock {\em The Cuckoo's Egg: Tracking a Spy Through the Maze of Computer
  Espionage}.
\newblock Doubleday, 1989.

\bibitem{SutraOm[20]}
Kartick Sutradhar and Hari Om.
\newblock Efficient quantum secret sharing without a trusted player.
\newblock {\em Quantum Information Processing}, 19, 2020.

\bibitem{SatKei[13]}
Satoshi Takahashi and Keiichi Iwamura.
\newblock Secret sharing scheme suitable for cloud computing.
\newblock In {\em 2013 IEEE 27th International Conference on Advanced
  Information Networking and Applications (AINA)}, pages 530--537, 2013.

\bibitem{Tassa[11]}
Tamir Tassa.
\newblock Generalized oblivious transfer by secret sharing.
\newblock {\em Designs, Codes and Cryptography}, 58:11--21, 2011.

\bibitem{Teir[94]}
Luc Teirlinck.
\newblock Some new 2-resolvable steiner quadruple systems.
\newblock {\em Designs, Codes and Cryptography}, pages 5--10, 1994.

\bibitem{Tompa[89]}
Martin Tompa and Heather Woll.
\newblock How to share a secret with cheaters.
\newblock {\em Journal of Cryptology}, 1:133--138, 1989.

\bibitem{Joni[17]}
Joni Uitto, Sampsa Rauti, Samuel Laur\'{e}n, and Ville Lepp\"{a}nen.
\newblock A survey on anti-honeypot and anti-introspection methods.
\newblock In {\em WorldCIST: Recent Advances in Information Systems and
  Technologies}, pages 125--134, 2017.

\bibitem{Hussein[15]}
Huseyin Ulusoy, Murat Kantarcioglu, Bhavani Thuraisingham, and Latifur Khan.
\newblock Honeypot based unauthorized data access detection in mapreduce
  systems.
\newblock In {\em IEEE International Conference on Intelligence and Security
  Informatics (ISI)}, pages 126--131, 2015.

\bibitem{Dijk[95]}
Marten van Dijk.
\newblock On the information rate of perfect secret sharing schemes.
\newblock {\em Des. Codes Cryptography}, 6(2):143--169, 1995.

\bibitem{Vrable[05]}
Michael Vrable, Justin Ma, Jay Chen, David Moore, Erik Vandekieft, Alex~C.
  Snoeren, Geoffrey~Michael Voelker, and Stefan Savage.
\newblock Scalability, fidelity, and containment in the potemkin virtual
  honeyfarm.
\newblock {\em ACM SIGOPS Operating Systems Review}, 39(5):148--162, 2005.

\bibitem{Wang[17]}
Lihua Wang, Yoshinori Aono, and Le~Trieu Phong.
\newblock A new secure matrix multiplication from ring-{LWE}.
\newblock In {\em CANS}, pages 93--111, 2017.

\bibitem{Wang[10]}
Ping Wang, Lei Wu, Ryan Cunningham, and Cliff~Changchun Zou.
\newblock Honeypot detection in advanced botnet attacks.
\newblock {\em International Journal of Information and Computer Security},
  4(1):30--51, 2010.

\bibitem{WangFan[19]}
Zhedong Wang, Xiong Fan, and Feng-Hao Liu.
\newblock {FE} for inner products and its application to decentralized {ABE}.
\newblock In {\em PKC}, pages 97--127, 2019.

\bibitem{Waterss[11]}
Brent Waters.
\newblock Ciphertext-policy attribute-based encryption: An expressive,
  efficient, and provably secure realization.
\newblock In {\em Public-Key Cryptography (PKC)}, pages 53--70, 2011.

\bibitem{HoWee[12]}
Hoeteck Wee.
\newblock Dual projective hashing and its applications — lossy trapdoor
  functions and more.
\newblock In {\em EUROCRYPT}, pages 246--262, 2012.

\bibitem{JonYan[10]}
Jonathan Weir and WeiQi Yan.
\newblock A comprehensive study of visual cryptography.
\newblock {\em Lecture Notes in Computer Science}, 6010:70--105, 2010.

\bibitem{NungJia[13]}
Ching-Nung Yang and Jia-Bin Lai.
\newblock Protecting data privacy and security for cloud computing based on
  secret sharing.
\newblock In {\em International Symposium on Biometrics and Security
  Technologies}, pages 259--266, 2013.

\bibitem{YangYvo[10]}
Qiushi Yang and Yvo Desmedt.
\newblock General perfectly secure message transmission using linear codes.
\newblock In {\em ASIACRYPT}, pages 448--465, 2010.

\bibitem{Yao[82]}
Andrew~C. Yao.
\newblock Protocols for secure computations.
\newblock In {\em FOCS}, pages 160--164, 1982.

\bibitem{Sergey[08]}
Sergey Yekhanin.
\newblock Towards 3-query locally decodable codes of subexponential length.
\newblock {\em Journal of the ACM (JACM)}, 55(1):1--16, 2008.

\bibitem{Yuca[99]}
Joseph~L. Yucas.
\newblock Extending {AG}(4, 2) to {S}(4,\{5, 6\}, 17).
\newblock {\em J. Combin. Des.}, pages 113--117, 1999.

\bibitem{Yuca[02]}
Joseph~L. Yucas.
\newblock Extensions of {PG}(3, 2) with bases.
\newblock {\em Australas. J. Combin.}, pages 125--131, 2002.

\bibitem{JiaZhen[20]}
Jiang Zhang and Zhenfeng Zhang.
\newblock {\em Lattice-Based Cryptosystems - A Design Perspective}.
\newblock Data Structures and Information Theory. Springer Singapore, 2020.

\bibitem{HongWei[12]}
Hong Zhong, Xiaodong Wei, and Runhua Shi.
\newblock A novel anonymous secret sharing scheme based on {BP} artificial
  neural network.
\newblock In {\em International Conference on Natural Computation}, pages
  366--370, 2012.

\bibitem{Zou[06]}
Cliff~Changchun Zou and Ryan Cunningham.
\newblock Honeypot-aware advanced botnet construction and maintenance.
\newblock In {\em International Conference on Dependable Systems and Networks
  (DSN)}, pages 199--208, 2006.

\end{thebibliography}

\end{document}